\documentclass[10pt,a4paper]{article}
\usepackage[utf8]{inputenc}
\usepackage{amsmath, amsthm}
\usepackage{amsfonts}
\usepackage{amssymb}
\usepackage{bbm}
\usepackage{hyperref}
\newtheorem*{thmunb}{Theorem}	

\newcommand{\Cin}{\underline{C}_{in}}
\newcommand{\rtheta}{\Re \theta}

\newcommand{\itheta}{\Im \theta}
\newcommand{\rxi}{\Re \xi}
\newcommand{\rphi}{\Re \phi}
\newcommand{\iphi}{\Im \phi}
\newcommand{\ixi}{\Im \xi}
\newcommand{\usi}{u_{\mathcal{S}_{i^+}}}

\newcommand{\ttheta}{\tilde{\theta}}
\newcommand{\txi}{\tilde{\xi}}

\newcommand{\tphi}{\tilde{\phi}}
\newcommand{\ep}{\epsilon}
\newcommand{\uch}{u_{\CH}}

\newcommand{\exd}{e^{\frac{K_-}{4}v}}
\newcommand{\uS}{u_{\mathcal{S}}(v)}
\newcommand{\vS}{v_{\mathcal{S}}(u)}

\newcommand{\A}{\mathcal{A}}
\newcommand{\T}{\mathcal{T}}
\newcommand{\R}{\mathcal{R}}
\newcommand{\uR}{u_{\mathcal{R}}}

\newcommand{\CH}{\mathcal{CH}_{i^+}}

\newtheorem{theo}{Theorem}

\usepackage{slashed}
\usepackage[left=2cm,right=2cm,top=2cm,bottom=2cm]{geometry}

\usepackage{authblk}
\usepackage{enumerate}

\usepackage{enumerate}

\theoremstyle{plain}
\newtheorem{thm}{Theorem}[section]

\newtheorem{lemma}[thm]{Lemma}

\newtheorem{prop}[thm]{Proposition}

\newtheorem{cor}[thm]{Corollary}
\newtheorem{open}[thm]{Open problem}
\newtheorem{conjecture}[thm]{Conjecture}

\theoremstyle{remark}
\newtheorem{rmk}{Remark}[section]
\newtheorem*{rmks}{Remark}

\theoremstyle{definition}

\newcommand{\RR}{\mathbb{R}}

\newcommand{\barL}{\underline{L}}
\newcommand{\rd}{\partial}
\newcommand{\ls}{\lesssim}

\theoremstyle{plain}

\theoremstyle{remark}

\theoremstyle{definition}

\usepackage{hyperref, todonotes, mathtools}	

\mathtoolsset{showonlyrefs}

\newcommand{\vR}{v_{\mathcal{R}}}

\usepackage{bm}

\usepackage{color}

\usepackage{float}

\addtocounter{tocdepth}{-2}
\usepackage{graphicx}
\usepackage{authblk}
 
\numberwithin{equation}{section}

\title{The  coexistence of null and spacelike singularities \\inside spherically symmetric black holes}

\author[1]{Maxime~Van~de~Moortel\thanks{maxime.vandemoortel@rutgers.edu}}
\affil[1]{\small  Department of Mathematics, Rutgers University, 
	Hill~Center,~New~Brunswick~NJ~08854,~United~States~of~America \vskip.1pc \  }

\date{\today}

\begin{document}
	\maketitle
	\thispagestyle{empty}
	
	\thispagestyle{empty}
	
	\begin{abstract}  
		In  our previous work  \emph{[Van de Moortel, The breakdown of weak null singularities, Duke Mathematical Journal 172 (15), 2957-3012, 2023]}, we showed that  dynamical black holes formed in charged spherical collapse generically feature both a null weakly singular Cauchy horizon and a stronger (presumably spacelike) singularity, 
		confirming a longstanding conjecture in the physics literature. However, this previous result, based on  a contradiction argument, did not provide  quantitative estimates on  
		the stronger singularity.

		In this study, we adopt a  new approach by analyzing   local initial data inside the black hole that are  consistent with a breakdown of the Cauchy horizon. We  prove that the remaining portion is spacelike  and obtain  sharp spacetime estimates  near the null-spacelike transition. 	Notably, we show that 
		the Kasner exponents of the spacelike portion are positive, in  contrast to the well-known Oppenheimer--Snyder model of gravitational collapse. Moreover, these exponents degenerate to $(1,0,0)$ towards the null-spacelike transition. 
		
		Our result provides the first quantitative instances of a null-spacelike singularity transition inside  a black hole.  In our companion paper \cite{bif2}, we moreover apply our analysis to carry out the construction of a large class of asymptotically flat one or two-ended black holes featuring coexisting null and spacelike singularities.

	\end{abstract}

	\section{Introduction}

	Understanding the strength of the singularity in the interior of an astrophysical black hole is of crucial importance in the mathematics of gravitation. In particular, the statement of the Strong Cosmic Censorship conjecture, at the heart of the very question of determinism in General Relativity,  relies 	on characterizing the   \emph{generic} black hole singularities, see  \cite{ChristoSCC,PenroseSCC} or the review \cite{review}. In the context of gravitational collapse, the  first example of a dynamical black hole was constructed in the celebrated paper of Oppenheimer--Snyder \cite{OppenheimerSnyder}  considering spherically-symmetric solutions of  the Einstein equation in the presence of dust. The singularity inside Oppenheimer--Snyder black holes is spacelike, 
	and coincides in parts with the  singularity $\mathcal{S}=\{r=0\}$ inside the  Schwarzschild black hole solution, which is well-known for its spaghettification
	\footnote{Spaghettification  is typically understood as the manifestation of infinite tidal deformations with a dilation in the radial direction and a contraction in orthoradial directions jointly experienced in finite time by infalling observers \cite{gravitation}.
	} experienced by all infalling observers when reaching $\mathcal{S}$, see e.g.\ \cite{Hawking,gravitation}: \begin{equation}\label{Schwarz}
		g_S = -(1-\frac{2M}{r}) dt^2+ (1-\frac{2M}{r})^{-1} dr^2+ r^2 ( d\theta^2+ \sin^2(\theta) d\varphi^2).
	\end{equation} The prevalence  of  spacelike singularities of Schwarzschild-type (together with other considerations related to the blue-shift instability of the Kerr black hole interior, see Section~\ref{CH.section} for more detail) had given credence to the  ambitious statement that in generic gravitational collapse, the black hole terminal boundary is \emph{everywhere spacelike}, see  \cite{KerrStab,claylecturenotes,JonathanICM,review} for further discussion. The groundbreaking work of Dafermos--Luk \cite{KerrStab}, however, showed that this conjecture is false for the Einstein equations in vacuum, in proving that small perturbations of the Kerr black hole admit a Cauchy horizon, which is a null terminal boundary component (not spacelike). Ironically, this begs the question as to whether the terminal boundary of generic black holes even\footnote{There exists, in fact, two-ended asymptotically flat black holes with no spacelike singularities \cite{nospacelike}; however, our attention here is mostly focused on the setting of gravitational collapse in which black holes are one-ended asymptotically flat, see Section~\ref{WCC.section}.} \emph{contains a spacelike portion}.  A gravitational collapse scenario in which the Cauchy horizon is the only boundary component and ``closes-off'' the spacetime as depicted in Figure~\ref{fig:disproof} is indeed consistent with the result of \cite{KerrStab}, see \cite{MihalisICM}. 	This problem was, however, resolved in the author's  work \cite{breakdown}  for a spherically-symmetric model of gravitational collapse (see Section~\ref{WCC.section} and Section~\ref{breakdown.section}) given by the  Einstein--Maxwell--Klein--Gordon system
	
	\begin{equation} \label{1.1}   Ric_{\mu \nu}(g)- \frac{1}{2}R(g)g_{\mu \nu}= \mathbb{T}^{EM}_{\mu \nu}+  \mathbb{T}^{KG}_{\mu \nu} ,    \end{equation} 
	\begin{equation} \label{2.1} \mathbb{T}^{EM}_{\mu \nu}=2\left(g^{\alpha \beta}F _{\alpha \nu}F_{\beta \mu }-\frac{1}{4}F^{\alpha \beta}F_{\alpha \beta}g_{\mu \nu}\right),
	\end{equation}
	\begin{equation} \label{3.1} \mathbb{T}^{KG}_{\mu \nu}= 2\left( \Re(D _{\mu}\phi \overline{D _{\nu}\phi}) -\frac{1}{2}(g^{\alpha \beta} D _{\alpha}\phi \overline{D _{\beta}\phi} + m ^{2}|\phi|^2  )g_{\mu \nu} \right), \end{equation} \begin{equation} \label{4.1} \nabla^{\mu} F_{\mu \nu}= \frac{ q_{0} }{2}i (\phi \overline{D_{\nu}\phi} -\overline{\phi} D_{\nu}\phi) , \; F=dA ,
	\end{equation} \begin{equation} \label{5.1} g^{\mu \nu} D_{\mu} D_{\nu}\phi = m ^{2} \phi , 	\end{equation} 	where $\phi$ is a scalar field of charge $q_0 \neq 0$ and of mass $m^2 \geq 0$ and $D_{\mu}= \nabla_{\mu}+iq_0 A_{\mu}$ is the gauge derivative.  The system \eqref{1.1}--\eqref{5.1} has  been extensively studied in spherical symmetry and is generally considered as one of the most adequate spherically symmetric models for gravitational collapse \cite{Kommemi} emulating many features of the non-spherically-symmetric Einstein equations, following Wheeler's idea that the effect of electromagnetic charge in the Einstein equations is analogous to that of angular momentum, see \cite{Dafermos:2004jp,Dafermos.Wheeler,review}.
	
	In  \cite{Moi}, the author proved that a generic black hole  solution of \eqref{1.1}--\eqref{5.1} in spherical symmetry admits a (null) Cauchy horizon $\CH$. It is then proven   in \cite{Moi4} that $\CH$ is  weakly singular. Finally, it is shown in \cite{breakdown} that, due to this weak singularity, the Cauchy horizon $\CH$ cannot be the only component of the terminal boundary due to a novel phenomenon identified as ``the breakdown of weak null singularities'':  \begin{thmunb}[\cite{breakdown}]
		The interior of a spherically symmetric dynamical black hole  solution of \eqref{1.1}--\eqref{5.1} admits a weakly singular Cauchy horizon $\CH$ which must break down and give rise to an additional singularity.
	\end{thmunb}
	In other words, the Penrose diagram of Figure~\ref{fig:disproof} is impossible (see Section~\ref{breakdown.section} for further details).
	While the exact nature of this additional singularity is not obtained in \cite{breakdown}, it can be proven that it is either a locally-naked singularity or a singularity $\mathcal{S}=\{r=0\}$ foliated by spheres of zero area-radius $r$. 	However, due to the fact that \cite{breakdown} proceeds by contradiction, it is not known whether $\mathcal{S}$ is spacelike, let alone how $(g,\phi)$ behave quantitatively near $\mathcal{S}$. Conjecturally, locally-naked singularities are non-generic (see \cite{ChristoCQG,Kommemi,review}, footnote~\ref{footnote3} and Section~\ref{WCC.section}) hence the breakdown of weak null singularities of \cite{breakdown} strongly suggests the presence of a \emph{spacelike} singularity $\mathcal{S}$. It also gives rise to the following fundamental open problem on black hole dynamics in gravitational collapse as already introduced by Dafermos \cite{Dafermos:2004jp} in 2004, see also \cite{MihalisICM,Kommemi}.


	\begin{open}[\cite{Dafermos:2004jp,MihalisICM,Kommemi,breakdown}] \label{spacelike.conj}
		Show that the  black hole terminal boundary  in generic gravitational collapse admits both a null Cauchy horizon $\CH$ which is weakly singular, and a spacelike singularity $\mathcal{S}$.
	\end{open}
	
	The main result of this paper (see  Theorem~\ref{thm.I} below) resolves Open Problem~\ref{spacelike.conj} for the spherically-symmetric model \eqref{1.1}--\eqref{5.1}, for spacetimes that do not have a locally-naked singularity\footnote{\label{footnote3}This assumption is natural, as it is conjectured that a \emph{generic black hole} does not admit any such locally-naked singularity. This statement is related to the Weak Cosmic Censorship Conjecture, see Section~\ref{WCC.section} for an extended discussion.} and obey certain  decay assumptions \eqref{decay.intro} in the black hole interior (we moreover construct a large class of gravitational collapse spacetimes satisfying these assumptions in our companion paper \cite{bif2}, see  already Theorem~\ref{thm.II}). 
	For such spacetimes,  we demonstrate that the terminal boundary  consists of a null Cauchy horizon $\CH$ and  a singularity $\mathcal{S}=\{r=0\}$, as depicted in the Penrose diagram of Figure~\ref{fig:spacelikeconj}. Moreover, we show that $\mathcal{S}$ is spacelike near $\CH$ and provide a comprehensive quantitative description of $(g,\phi)$ in the vicinity of $\CH\cap \mathcal{S}$. 
	Arguably, the most surprising feature of our analysis is that $\mathcal{S}$ differs significantly from the Schwarzschild-like singularity of Oppenheiner--Snyder's black hole;  indeed,  infalling observers experience infinite \emph{tidal contractions} in all directions at $\mathcal{S}$, as opposed to Schwarzschild's spaghettification. We will explain in Section~\ref{Kasner.section} why it is natural to conjecture that this phenomenon persists for solutions of \eqref{1.1}-\eqref{5.1} \emph{outside of spherical symmetry}\footnote{\label{BKLf}Conjectures for the Einstein vacuum equations, however, are more delicate to formulate in view of the chaotic dynamics that are expected near  spacelike singularities from  the celebrated BKL scenario \cite{BKL1,BKL2,mixmaster}, see Section~\ref{Kasner.section}.}.

	Our main theorem  can, in fact, be formulated  as a local result (with no reference to the global topology of spacetime) on a causal rectangle on which we \emph{assume that the 
		Cauchy horizon breaks down}. Our statement  is that the rest of the terminal boundary then contains a spacelike component $\mathcal{S}$, which we describe quantitatively  in  detail using an Eddington--Finkelstein type coordinate $v$. We provide a simplified version of the result below:

	\begin{theo} \label{thm.I} Consider  local initial data in the interior of a black hole  consisting of an ingoing cone $\Cin$ and an outgoing cone $C_{out}$ terminating at the sphere of a weakly singular Cauchy horizon $\CH$, and denote  $\mathcal{B}$ the terminal boundary of the resulting solution of \eqref{1.1}--\eqref{5.1}.  Assume   a  breakdown of  the Cauchy horizon and no locally-naked singularity, i.e., $ \CH = \{v=+\infty\} \underset{\neq}{\subset} \mathcal{B}  $, 
		and that there exists $s>1$ such that the following  hold: \begin{equation}\label{decay.intro}
			v^{-s} \lesssim	|D_v \phi|_{|C_{out}}(v) \lesssim v^{-s},\  |\Im(\bar{\phi} D_v \phi)|_{|C_{out}}(v) \ll v^{-s},\  	|D^2_{vv} \phi|_{|C_{out}}(v) \lesssim v^{-s-1} \text{ as } v\rightarrow +\infty.
		\end{equation}		Then, $\mathcal{B}$ contains a  spacelike singularity $\mathcal{S} \neq \emptyset$ intersecting $\CH$ as depicted in Figure~\ref{fig:local} and the metric near $\CH \cap \mathcal{S}$ is approximated by a Kasner metric of $v$-dependent positive Kasner exponents  $(1-2p(u,v), p(u,v),p(u,v))$
		with $$ p(u,v) \approx \frac{1}{v} 	\text{ as } v \rightarrow +\infty.$$
	\end{theo}
	We refer to Section~\ref{Kasner.section} for more details on Kasner-like metrics with $v$-dependent  exponents; see, in particular, the estimates \eqref{kasner.deg.intro}, \eqref{kasner.deg.intro2}. Note that the Kasner exponent $p(u,v)$ degenerates to $0$ as $v\rightarrow+\infty$.
	\begin{rmk}
		Theorem~\ref{thm.I} is formulated as a result conditional on $ \CH \underset{\neq}{\subset} \mathcal{B}  $  (which requires  the absence of an outgoing locally-naked singularity emanating from $\Cin$)
		for the sake of greater generality. However, we highlight that  $ \CH \underset{\neq}{\subset} \mathcal{B}  $ is satisfied if $\Cin$ is trapped and its endpoint is a collapsed sphere of zero-area radius  as depicted in Figure~\ref{fig:local}. We furthermore \emph{construct a  large class of initial data} for which $\Cin$ satisfies these conditions; in other words, for such initial data, the conclusion of Theorem~\ref{thm.I} holds unconditionally, see Remark~\ref{rmk.3} below.

	\end{rmk}

	\begin{figure}\label{fig:local}
		
		\begin{center}
			
			\includegraphics[width=64 mm, height=40 mm]{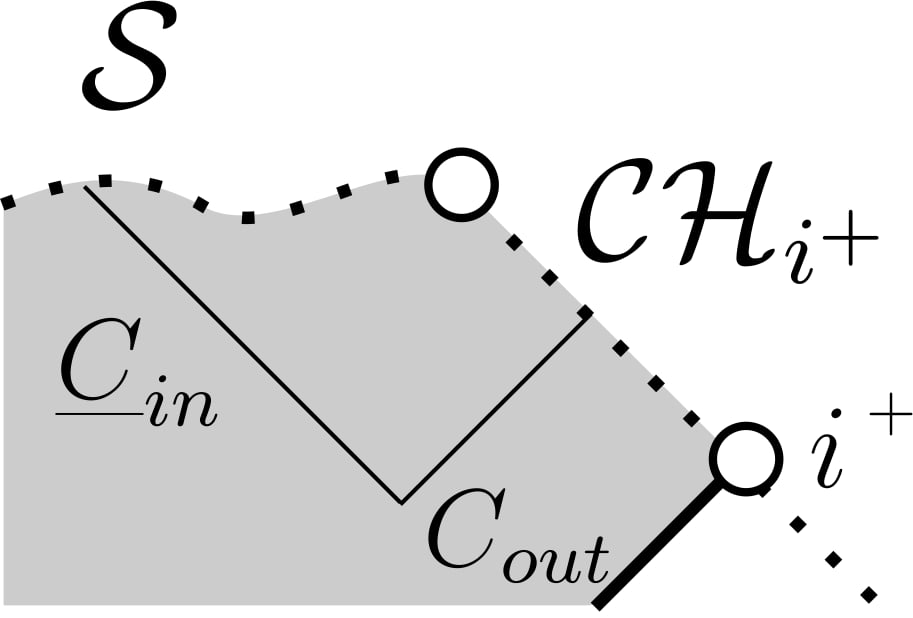}
			
		\end{center}
		\caption{Penrose diagram of a solution obtained in Theorem~\ref{thm.I} with bifurcate initial data $\Cin \cup C_{out}$.}

	\end{figure}
	
	\begin{rmk}\label{rmk.3}
		To  prove Theorem~\ref{thm.I}, we establish roughly three main statements, which are all of independent of interest (see Theorem~\ref{main.thm} in Section~\ref{thm.section} which can be viewed as a more precise formulation of Theorem~\ref{thm.I}): \begin{enumerate}[i.]
			\item  \label{I.intro} (Theorem~\ref{main.thm}, statement~\ref{main.thm.I}). Unconditional estimates under the Cauchy horizon $\CH$ \emph{which do not require that $\CH$ breaks down}. We show that for $v_0$ sufficiently large, $[u_0,\uch] \times [v_0,+\infty)$  is in the trapped  region  and there exists $s>1$ such that for all  $(u,v)\in[u_0,\uch] \times [v_0,+\infty)$, the area-radius $r(u,v)$ satisfies   $$ r(u,v) \gtrsim v^{\frac{1}{2} -s},$$	
			where $\uch$ is defined so that $\{\uch\}\times [v_0,+\infty)$ is the outgoing light cone terminating at the endpoint of the Cauchy horizon $\CH$.	As a corollary, we show that there can never exist a null ingoing boundary component $\mathcal{S}_{i^+}$ on which $r$ extends to $0$   (see Figure~\ref{fig:one-ended} for a depiction of the case we exclude).
			\item \label{II.intro} (Theorem~\ref{main.thm}, statement~\ref{main.thm.II}). Now \emph{assuming that the Cauchy horizon breaks down} locally, i.e.,  $ \CH \underset{\neq}{\subset} \mathcal{B}  $, we show that $r$ tends to $0$ towards the endpoint of $\CH$ denoted $(\uch,+\infty)$, i.e., $$ \lim_{ (u,v) \rightarrow (\uch,+\infty)} r(u,v) = 0.$$
			
			Moreover, we show that there exists a spacelike component of the boundary $\mathcal{S}=\{r=0\}$ connected to $\CH$ and described by a Kasner metric of  exponents  $(1-2p(u,v),p(u,v),p(u,v))$ with $p(u,v) \approx\frac{1}{v}$  as $v \rightarrow +\infty$. 
			\item \label{III.intro}  (Theorem~\ref{main.thm}, statement~\ref{main.thm.III}) We provide (a large class of) \emph{local initial data} leading to a Cauchy horizon breakdown ($ \CH \underset{\neq}{\subset} \mathcal{B}  $). Thus, the conclusion of Statement~\ref{II.intro} applies unconditionally for such  initial data.
		\end{enumerate}
	\end{rmk}

	To explain the terminology, we recall the Kasner metric \cite{Kasner}, a solution of \eqref{1.1}-\eqref{5.1} with $F\equiv 0$, $m^2=0$:
	\begin{equation} \label{Kasner.exact}
		g_{Kas}= -d{\tau}^2 + {\tau}^{2 p_{1}}  dx_1^2 +    {\tau}^{2 p_{2}} dx_2^2+ {\tau}^{2 p_{3}}dx_3^2, \hskip 5 mm \phi({\tau})= p_{\phi} \cdot \log({\tau}^{-1}); 
	\end{equation} \begin{equation} \label{p.Kasner}
		p_{1}+ p_2+ p_3 =1, \hskip 10 mm p_{1}^2+ p_2^2+ p_3^2+ 2 p_{\phi}^2 =1,
	\end{equation} where $(p_{1},p_2,p_3)$ are \underline{constants} called the Kasner exponents. If these exponents depend on $(x_1,x_2,x_3)$ instead, \eqref{Kasner.exact} do not solve \eqref{1.1}-\eqref{5.1}, but it may approximate a solution of \eqref{1.1}-\eqref{5.1} as $\tau \rightarrow 0$. Such solutions describe a large amount of spacelike singularities \cite{FournodavlosLuk}. 
	The case  $(p_{1},p_2,p_3)=(1,0,0)$ is degenerate: it corresponds to the Minkowski metric and has no singularity at $\tau=0$. Despite many Kasner-with-variable-exponents having been constructed \cite{FournodavlosLuk}, it appears that the spacetimes of Theorem~\ref{thm.I} provide the first examples of \emph{Kasner exponents which converge to the degenerate values} $(1,0,0)$. It is striking that this phenomenon  conjecturally arises in \emph{any} generic black hole (even outside of spherical symmetry), see Section~\ref{Kasner.section}.

	As we explained, our main result Theorem~\ref{thm.I} is a local statement, that concerns null/spacelike transition, and their quantitative behavior, regardless of the global topology of the spacetime. In our companion paper \cite{bif2}, we return to the global setting of gravitational collapse and show that the estimates of Theorem~\ref{thm.I} apply in this context, modulo obstructions related to locally naked singularities emanating from the center of the collapsing star $\Gamma$ (Theorem~\ref{thm.II}, Statement~\ref{II.A}). Moreover, we construct a large class of one-ended asymptotically flat black holes with a center $\Gamma$  satisfying \eqref{decay.intro} in which these obstructions are not present and thus to which we can apply Theorem~\ref{thm.I} (Theorem~\ref{thm.II}, Statement~\ref{II.B}). In \cite{bif2}, we also derive a  simpler and more constructive proof of the breakdown of weak null singularities first proven in \cite{breakdown}, see Section~\ref{breakdown.section}. We note   that Theorem~\ref{thm.II} provides the first-time construction of a one-ended black hole with a spacelike singularity $\mathcal{S}$ that has a  non-trivial charge. 
	
	\begin{theo} \label{thm.II}[Gravitational collapse black hole with coexisting singularities, our companion paper \cite{bif2}].
		\begin{enumerate}
			\item\label{II.A} (Theorem IV, Statement A in \cite{bif2}). Consider a spherically-symmetric one-ended solution of \eqref{1.1}--\eqref{5.1} with $q_0\neq 0$ that has a weakly singular null Cauchy horizon $\CH  \neq \emptyset$ towards which \eqref{decay.intro} is satisfied, and no locally-naked singularity emanating from the center $\Gamma$. Then,  the rest of the spacetime boundary is a crushing singularity $\mathcal{S}=\{r=0\}$, which is spacelike in a neighborhood of $b_{\Gamma}$, and in a neighborhood of $\CH\cap \mathcal{S}$, where it obeys the Kasner asymptotics of Theorem~\ref{thm.I}.
			
			\item \label{II.B} (Theorem II  in \cite{bif2}). There exists a large class of  spherically-symmetric one-ended asymptotically flat black hole solutions of \eqref{1.1}--\eqref{5.1} with $q_0\neq0$, $m^2=0$, with  a regular center $\Gamma \neq \emptyset$,   a weakly singular null Cauchy horizon $\CH  \neq~ \emptyset$, and a crushing singularity $\mathcal{S}=\{r=0\}$, which is spacelike in a neighborhood of $\CH\cap \mathcal{S}$ and obeys the Kasner asymptotics of Theorem~\ref{thm.I}. 
		\end{enumerate}
		
	\end{theo}
	
	\begin{figure}[H] \begin{center}\includegraphics[width=79 mm, height=45 mm]{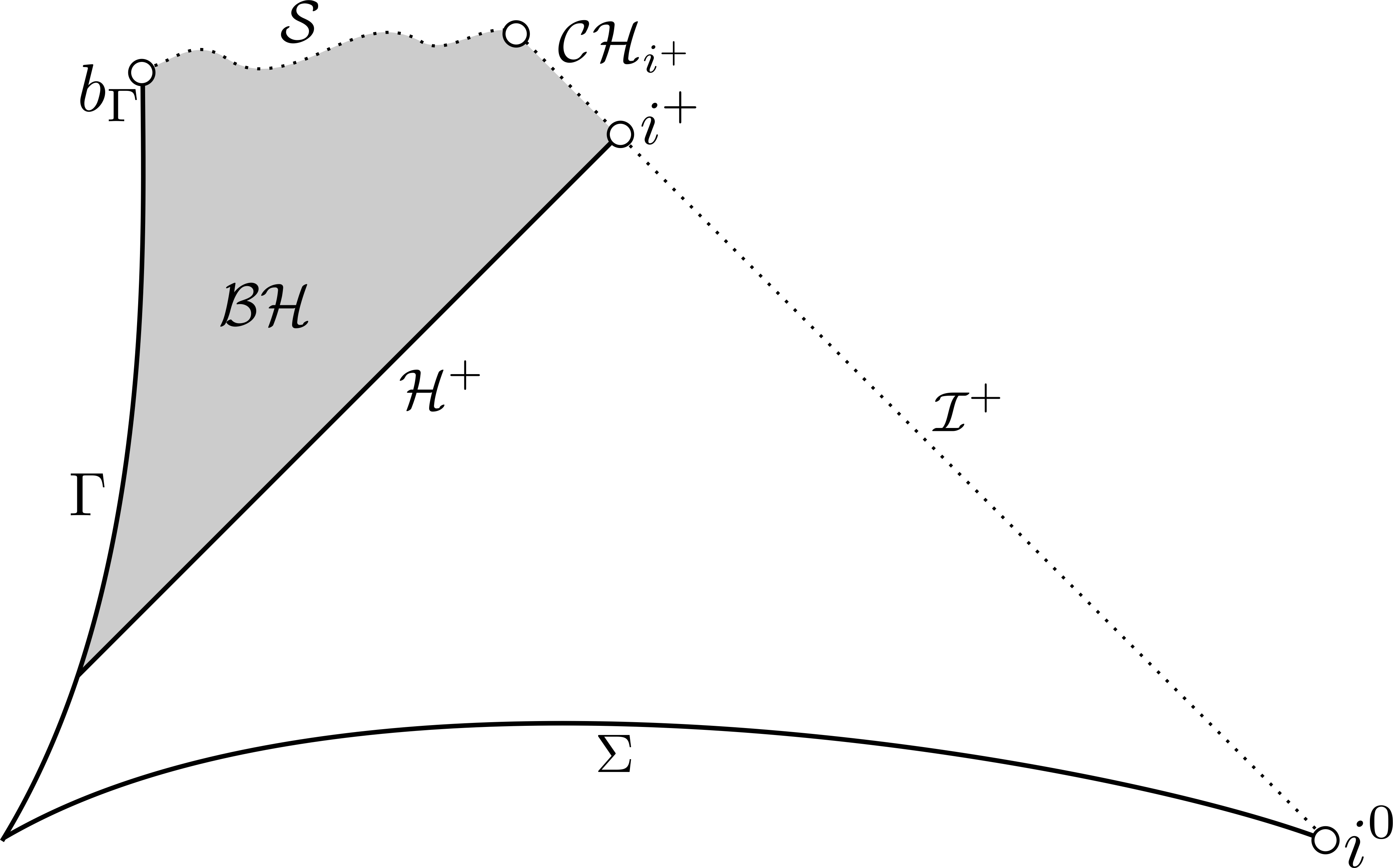}				
			\caption{Penrose diagram of the gravitational collapse (=one-ended) spacetimes obtained in Theorem~\ref{thm.II}. \\ $\CH$ is a weak null singularity and $\mathcal{S}$ a strong spacelike singularity.}
			\label{fig:spacelikeconj}	\end{center}\end{figure}
	
	\begin{rmk}\label{rmk.Oppen}
		We emphasize that the model of gravitational collapse offered by Theorem~\ref{thm.II} differs from the Oppenheimer--Snyder and Christodoulou's model  in the following crucial ways: \begin{itemize}
			\item There exists a non-trivial null boundary component emanating from $i^+$:   the Cauchy horizon $\CH$. The spacetime curvature is infinite at $\CH$,  but in-falling observers  experience finite tidal deformations. 
			\item The singularity $\mathcal{S}$ is spacelike and the infinite tidal deformations are compressive in all directions, corresponding to positive Kasner exponents
			(at least in a neighborhood of $\CH$).
		\end{itemize}
		In Sections~\ref{CH.section}--\ref{breakdown.section} and Section~\ref{Kasner.section}, we will provide heuristic arguments indicating that a generic black hole solving  the Einstein-scalar-field equations (without symmetry) possesses the same qualitative features.
	\end{rmk}
	
	We refer the reader to Theorem 3.4 in \cite{bif2} for the precise result corresponding to Theorem~\ref{thm.II}. We also highlight that we obtain an analogue\footnote{In the two-ended case, locally-naked singularities do not exist; however, one must assume that $\mathcal{S}\neq \emptyset$, instead of proving it.} of Theorem~\ref{thm.II} for two-ended spacetimes, see \cite{bif2} and Theorem~\ref{thm.III} in Section~\ref{twoended.section}.


	We finally conclude this preamble by highlighting the importance of  estimates in the fashion of Theorem~\ref{thm.I} for the analogue of Open Problem~\ref{spacelike.conj} \emph{outside of spherical symmetry}. It is indeed unlikely that a contradiction argument in the style of \cite{breakdown} can be employed in this case, and thus, the quantitative analysis of the spacetime metric ought to be instrumental in an eventual resolution of Strong Cosmic Censorship in generic non-spherical gravitational collapse, see Section~\ref{Kasner.section} for further discussions. We hope to return to this problem in future works.

	
	\subsection{Connections to the Strong Cosmic Censorship Conjecture}  \label{SCC.section}
	
	\paragraph{The mathematical statement of Strong Cosmic Censorship}	The statement of Strong Cosmic Censorship concerns the solutions of the Einstein equations in the presence of a reasonable matter model, and comes back to Penrose \cite{penrose1974gravitational}. We provide its modern formulation (see \cite{nospacelike,ChristoCQG}) making use of the Maximal Globally Hyperbolic Development (MGHD) of Choquet-Bruhat and Choquet-Bruhat--Geroch \cite{MGHD,GHD}:
	\begin{conjecture}[Strong Cosmic Censorship] \label{SCC.conj}
		The Maximal Globally Hyperbolic Development of generic, one-ended asymptotically
		flat complete initial data  is
		inextendible as a solution to the Einstein equations.
	\end{conjecture} We remark that Conjecture~\ref{SCC.conj} is formulated in the context of gravitational collapse (one-ended initial data with trivial topology, see Section~\ref{WCC.section}) following the original spirit of Penrose \cite{PenroseSCC,penrose1974gravitational,Penrose1979}.
	
	\paragraph{Connection to black holes}		As it turns out, Strong Cosmic Censorship, although its statement does not directly refers to black holes, is very much related to black hole dynamics. The simplest black hole, namely Schwarzschild's solution \eqref{Schwarz}  of the Einstein equations, is non-rotating and terminates at a spacelike singularity $\mathcal{S}=\{r=0\}$, where tidal deformations are infinite as we discussed earlier. It is known \cite{JanC0} that the Schwarzschild MGHD is inextendible as a $C^0$ Lorentzian manifold, and thus satisfies  the statement of Conjecture~\ref{SCC.conj}.	Penrose, however, noted that the well-known Kerr black hole \cite{MR0156674} (the rotating analogue of \eqref{Schwarz}) seemingly does not respect Strong Cosmic Censorship \cite{Penroseblue,penrose1974gravitational}: this is because the MGHD of Kerr initial data, in contrast, terminates at a null boundary -- the Cauchy horizon $\CH$ -- which is smoothly extendible. The original motivation of  Penrose for Conjecture~\ref{SCC.conj} came from an argument that dynamical perturbations of the Kerr black hole suffer from a blue-shift instability \cite{Penroseblue,Penrose} which restores Strong Cosmic Censorship at least for  generic solutions of the Einstein equations. In essence, a proof of Conjecture~\ref{SCC.conj} will require to characterize the terminal boundary of a generic black hole and show that it is sufficiently singular to prevent any extension as a solution to the Einstein equations in the appropriate sense (see the review \cite{review} for an extended discussion).

	\paragraph{Perturbations of the Kerr black hole for the Einstein equations in vacuum}	While a full proof of Conjecture~\ref{SCC.conj}, even in vacuum, remains elusive, it is known by the work of Dafermos--Luk \cite{KerrStab} that small perturbations of the Kerr black hole admit a continuously extendible Cauchy horizon $\CH$ (despite the potential occurence of a blue-shift instability). As a corollary, the so-called $C^0$-version of Strong Cosmic Censorship \cite{MihalisICM,JonathanICM,KerrStab}, a version of Conjecture~\ref{SCC.conj} precluding any continuous extension of a generic MGHD, is false! In view of the blue-shift instability mechanism, it is nonetheless  likely that the Cauchy horizon of generic MGHDs is $C^2$ or $H^1$ inextendibible, see \cite{ChristoSCC,KerrStab,review}.  The proof of such a statement remains open at present; see, however, the following linear  instability results \cite{SbierskiTeukolsky,KerrInstab,blueyakovM} on the Kerr black hole interior.

	\paragraph{Spherically-symmetric version of Strong Cosmic Censorship}	Finally, we note that, as a simplified problem for Conjecture~\ref{SCC.conj}, it is possible to examine the inextendibility of spherically-symmetric spacetimes solutions of the Einstein equations coupled with Maxwell--Klein--Gordon, i.e., \eqref{1.1}--\eqref{5.1}, namely a spherically-symmetric analogue of Strong Cosmic Censorship. In gravitational collapse, this problem is open, even assuming spherical symmetry (see Section~\ref{CH.section} and Section~\ref{breakdown.section} for a discussion of partial results). For  two-ended black holes, however, a $C^2$ version of the spherically-symmetric Strong Cosmic Censorship was obtained by Luk--Oh \cite{JonathanStab,JonathanStabExt} for \eqref{1.1}--\eqref{5.1} assuming $q_0=m^2=0$ (uncharged massless scalar field model), building up on previous works by Dafermos \cite{Mihalis1,MihalisPHD} and Dafermos--Rodnianski \cite{PriceLaw}.

	\paragraph{Inextendibily  across the singularities obtained in our  result}	We conclude this section  discussing the (future)-inextendibility of the spacetimes obtained in Theorem~\ref{thm.I}. It is proven in \cite{Moi,MoiChristoph} that the Cauchy horizon $\CH$ is continuously extendible (under natural\footnote{\label{fMk}these assumptions include oscillation conditions on the scalar field  conjectured to hold generically, see \cite{MoiChristoph,review} for a discussion.} assumptions on the event horizon). However, $\CH$ is $C^2$-inextendible as proven in \cite{Moi4}. The techniques of \cite{JanC1} moreover ensure that $\CH$ is $C^1$-inextendible, even though a stronger $H^1$-inextendibility statement might be speculated, see \cite{review} for a discussion. As for the spacelike component $\mathcal{S}=\{r=0\}$ of the terminal boundary, it is manifestly $C^2$-inextendible due to the blow-up of the Kretschmann scalar (see e.g. \cite{Kommemi}). However, in view of \cite{JanC0} and the blow-up of tidal deformations at $\mathcal{S}$, we may conjecture that the metric  is $C^0$-inextendible across $\mathcal{S}$. For both boundary components, Theorem~\ref{thm.I} thus supports the conclusion of a spherically-symmetric version of Strong Cosmic Censorship, at least at the $C^2$-level.\\

	We refer the reader to the review \cite{review} for further details on the history of Strong Cosmic Censorship and how it articulates with modern mathematical results in black hole dynamics.
	
	\subsection{Local aspects of the black hole interior terminal boundary near $i^+$}  \label{CH.section}
	The only stationary solution of \eqref{1.1}--\eqref{5.1} in spherical symmetry is the Reissner--Nordstr\"{o}m metric (see \cite{Hawking}) \begin{equation}\label{RN}
		g_{RN} = -(1-\frac{2M}{r}+ \frac{e^2}{r^2}) dt^2+ (1-\frac{2M}{r}+  \frac{e^2}{r^2})^{-1} dr^2+ r^2 ( d\theta^2+ \sin^2(\theta) d\varphi^2),
	\end{equation} corresponding to $\phi \equiv 0 $ and $F= \frac{e}{r^2} dt \wedge dr$. In the sub-extremal case $0<|e|<M$, \eqref{RN} is a two-ended black hole solution whose MGHD terminates a smooth Cauchy horizon $\CH=\{r= M-\sqrt{M^2-e^2}\}$.
	
	The first step in the understanding of black interior dynamics for \eqref{1.1}--\eqref{5.1} is to consider the terminal boundary of the MGHD in the vicinity of $i^+$, interpreted as timelike infinity. To study this problem, one  poses characteristic data of a future affine-complete outgoing cone $\mathcal{H}^+$ interpreted as the black hole event horizon and a regular ingoing cone $\Cin$ penetrating $\mathcal{H}^+$ as depicted in Figure~\ref{fig:interior} and assume that $\mathcal{H}^+$ asymptotically converges to a sub-extremal Reissner--Nordstr\"{o}m  black hole. For \eqref{1.1}--\eqref{5.1} in spherical symmetry, it is conjectured  (based on heuristic/numerical arguments \cite{HodPiran1,HodPiran2,BurkoKhanna,KoyamaTomimatsu,KoyamaTomimatsu2,KoyamaTomimatsu3}, see  \cite{Moi,MoiChristoph,review} for further discussion) that for generic regular  solutions in the black hole exterior, this convergence occurs an the inverse-polynomial rate such that	\begin{equation}\label{decay}
		[1+|v|^{s}]  \phi_{|\mathcal{H}^+}(v) \text{ and its derivatives are bounded in an unweighted-in-}v\text{ norm.} 
	\end{equation}in an Eddington--Finkelstein advanced-time coordinate $v$ on $\mathcal{H}^+$ as $v\rightarrow +\infty$,  for some $s\geq \frac{5}{6}$; such inverse-polynomial decay statements are sometimes called generalized Price's law, after Price's original paper \cite{Pricepaper}  (see \cite{Moi2,MoiDejan,KGSchw1,JonathanStabExt,twotails,PriceLaw,AAG1,AAG2,AAG3,Hintz,Tataru,Tataru2} for related mathematical results or Section 6 in the review \cite{review}).

	Dafermos first studied this problem in spherical symmetry for solutions of \eqref{1.1}--\eqref{5.1} with $q_0=m^2=0$ (uncharged massless scalar field model) converging to a sub-extremal Reissner--Nordstr\"{o}m solution \eqref{RN} at rates consistent with \eqref{decay} and proved \cite{Mihalis1,MihalisPHD} that the terminal boundary admits a null component -- the Cauchy horizon -- which is moreover weakly singular  (consistently with the mass blow-up scenario, see \cite{review}, Section 4.5 and references therein). Here, and later in this article, we use the word ``Cauchy horizon'' for a null component of the MGHD terminal boundary foliated by topological spheres  of non-zero area (see \cite{Kommemi,KerrStab,MihalisSStrapped}). Detailed estimates together with a complete study of the black hole exterior were later obtained in Luk--Oh \cite{JonathanStab,JonathanStabExt}, in the context of their proof of  the $C^2$ version of Strong Cosmic Censorship  discussed in Section~\ref{SCC.section}.

	For the charged scalar field model, i.e., \eqref{1.1}--\eqref{5.1} with $q_0\neq 0$, the existence a Cauchy horizon was obtained in the author's work \cite{Moi}, and its weakly singular character was proven in \cite{Moi,Moi4}.

	\begin{thm}[\cite{Moi}]\label{CH.thm.SS} 
		
		
		Consider spherically symmetric characteristic initial data  for \eqref{1.1}--\eqref{5.1} 		on the event horizon $\mathcal{H}^+$ converging to a sub-extremal  Reissner--Nordstr\"{o}m black hole and on a $C^1$-regular ingoing cone $\Cin$.

		\begin{enumerate}[I.]			\item (\cite{Moi}, Theorem 3.2) Assume \eqref{decay}		holds as an upper bound on $\mathcal{H}^{+}=[v_0,+\infty)$ for some  decay rate $s>\frac{1}{2}$.					Then the spacetime is bound to the future by an ingoing null boundary $\CH \neq \emptyset$  (the Cauchy horizon) foliated by spheres of  positive radius and emanating from $i^+$, and the Penrose diagram is given by the dark gray region in Figure~\ref{fig:interior}. Moreover, if $s>1$, then $\phi$ is uniformly bounded and $g$ is continuously-extendible.			\item (\cite{Moi}, Theorem 3.3 \& \cite{Moi4}) If, moreover,  \eqref{decay} holds as a $L^2$ lower bound  on $\mathcal{H}^{+}$, then $\CH$ is a weak null singularity. In particular, the curvature of $g$ is infinite at $\CH$.		\end{enumerate}
	\end{thm}
	
	\begin{figure}[H]	\begin{center}\includegraphics[width=0.3\linewidth]{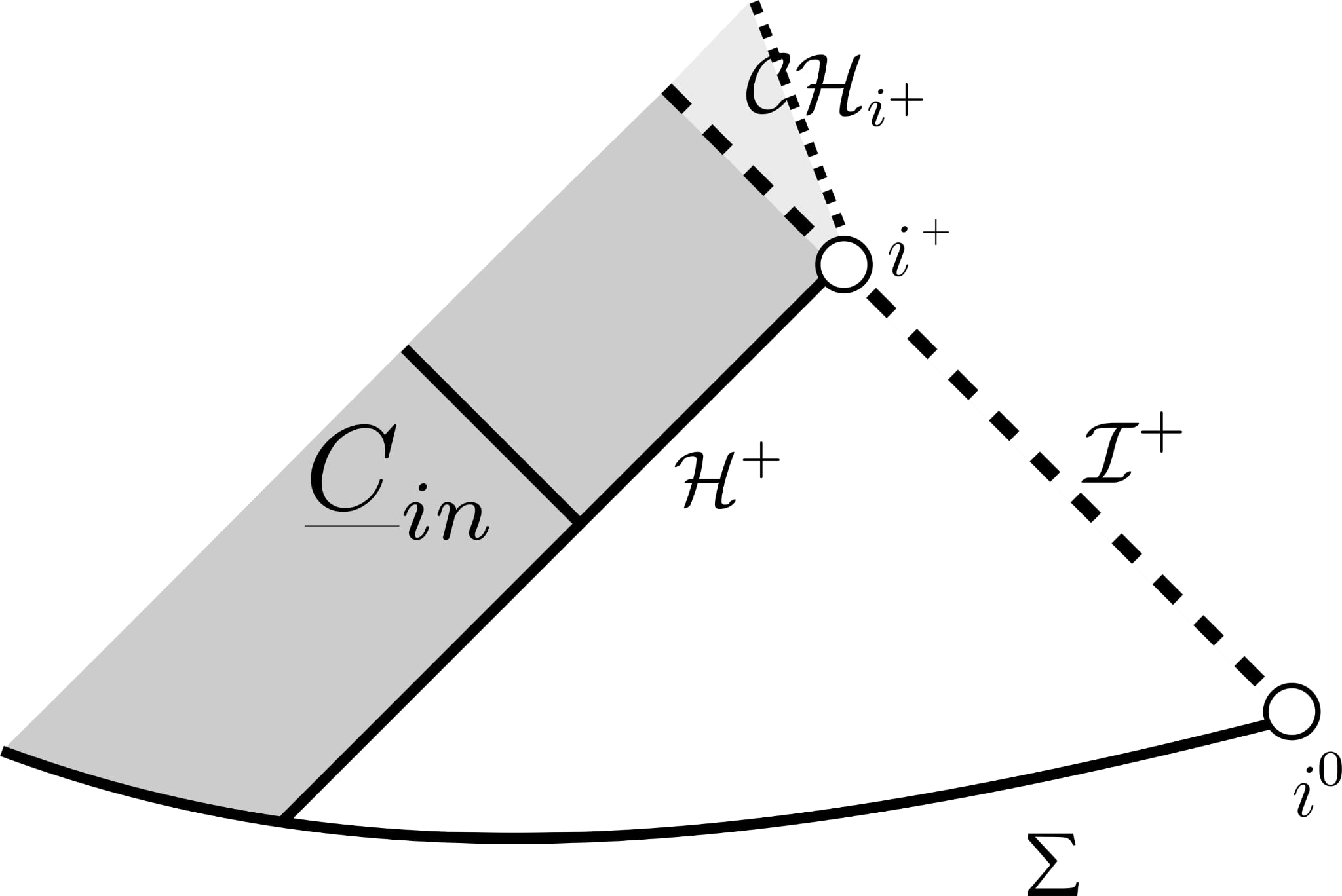}			\caption{\small Local structure of the black hole terminal boundary near $i^+$ for a spherically-symmetric solution of \eqref{1.1}--\eqref{5.1}.}		\label{fig:interior}\end{center}	\end{figure}

	The assumptions \eqref{decay.intro} made in Theorem~\ref{thm.I} (see Theorem~\ref{main.thm} for more specifics) are all consistent with the estimates obtained in Theorem~\ref{CH.thm.SS}. While, as we discussed, \eqref{decay} is conjecturally satisfied for some $s\geq \frac{5}{6}$, this $s$-value is different depending on the black hole parameters (we refer the reader again to the section 6 of \cite{review}): \begin{itemize}
		\item If $m^2\neq 0$, it is conjectured that $s=\frac{5}{6}$.
		
		\item If $m^2= 0$ and $|q_0 e| \geq \frac{1}{2}$, where $e$ is the black hole asymptotic charge, it is conjectured that $s=1$.
		
		\item If $m^2= 0$ and $|q_0 e| < \frac{1}{2}$, where $e$ is the black hole asymptotic charge, it is conjectured that $s>1$.
	\end{itemize}	
	If $s>1$, Theorem~\ref{CH.thm.SS} provides comprehensive estimates, the boundedness of $\phi$ and the metric coefficients. However, the black hole dynamics near $i^+$ are more delicate in Theorem~\ref{CH.thm.SS}  if one only assumes \eqref{decay} for $\frac{1}{2}<s\leq 1$. Indeed, there exist examples of initial data satisfying  the assumptions of Theorem~\ref{CH.thm.SS} but such as  $\phi$ is unbounded as showed in \cite{MoiChristoph}. However, the boundedness problem was settled in \cite{MoiChristoph} which proved more precise estimates under additional assumptions that are conjecturally satisfied in the exterior  (see footnote \ref{fMk}). 
	
	The assumptions \eqref{decay.intro} in Theorem~\ref{thm.I} are consistent with \eqref{decay} for\footnote{However, even without assuming $s>1$, we still obtain a priori estimates (\eqref{quant1} in Theorem~\ref{main.thm})  sufficient to rule out collapsed  ingoing cones  $\mathcal{S}_{i^+}$ and  reprove the breakdown of weak null singularities, see Proposition~\ref{apriori.prop2} and the related discussion in    \cite{bif2}.} $s>1$ (see Theorem~\ref{main.thm}) to avoid dealing with the subtleties of the    $\frac{1}{2}<s\leq 1$ case addressed in \cite{MoiChristoph}. We hope to cover this additional case in future works. We furthermore note that the Dafermos--Luk result in vacuum \cite{KerrStab}, while not assuming spherical symmetry (contrary to Theorem~\ref{CH.thm.SS}), proceeds under an assumption of fast decay on the event horizon analogous to \eqref{decay} for $s>1$. Therefore, our assumption that $s>1$ in Theorem~\ref{thm.I} will not be too restrictive in eventually generalizing the analysis 
	of Theorem~\ref{thm.I} to the  Einstein equations in vacuum outside of spherical symmetry.

	\subsection{Gravitational collapse and local aspects of the black hole interior terminal boundary near $\Gamma$} \label{WCC.section}

	\paragraph{Mathematical setting} Gravitational collapse  (see \cite{review}, Section 5) is modeled by the MGHD of asymptotically flat initial data $(\Sigma,g)$ with one-end, meaning that $\Sigma$ is diffeomorphic to $\mathbb{R}^3$ and thus has a center $\Gamma$ corresponding to the origin of  $\mathbb{R}^3$. Gravitational collapse spacetimes and their global aspects will be described further in Section~\ref{breakdown.section}. In contrast, the two-ended black holes discussed in   Section~\ref{twoended.section} are not models of gravitational collapse. 
	
	\paragraph{Locally-naked singularities}	The presence of a center $\Gamma$ allows for the existence of locally-naked singularities (see $\mathcal{CH}_{\Gamma}$, $\mathcal{S}_{\Gamma}^1$ and $\mathcal{S}_{\Gamma}^2$ in Figure~\ref{fig:one-ended}) emanating from $\Gamma$, which have been constructed by Christodoulou in \cite{Christo.existence} for the system \eqref{1.1}--\eqref{5.1} in spherical symmetry. The null component $\mathcal{CH}_{\Gamma}$ (see \cite{Kommemi} and Section 3 in \cite{bif2}) is potentially smoothly extendible, which is a threat to Strong Cosmic Censorship (see Section~\ref{SCC.section}).
	\begin{figure}\begin{center}	\includegraphics[width=105 mm, height=60 mm]{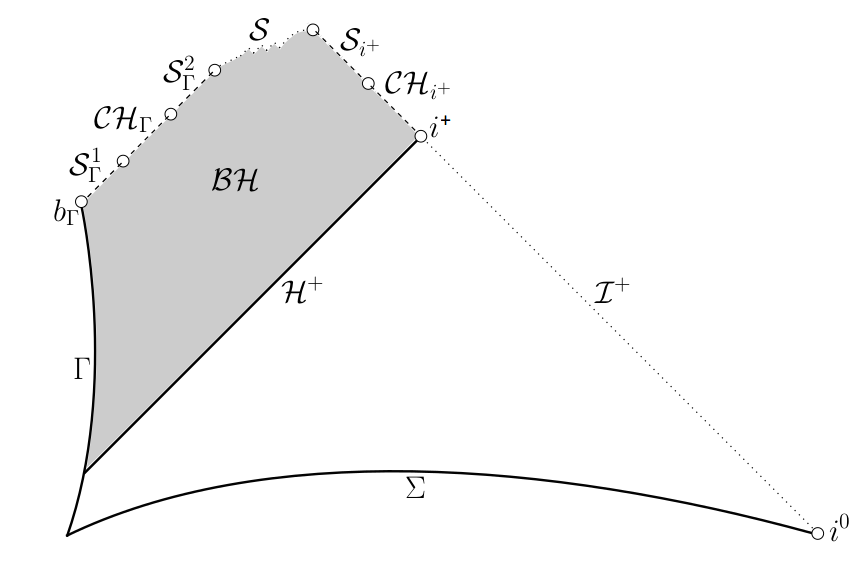}\caption{General Penrose diagram of a one-ended spherically-symmetric black hole solution after \cite{Kommemi}.}\label{fig:one-ended}\end{center}\end{figure}
	
	\paragraph{Weak Cosmic Censorship}	We also note that such locally-naked singularities can also arise outside of the black hole region: there are then \emph{naked singularities}. The celebrated statement of Weak Cosmic Censorship \cite{PenroseSCC,GerochHorowitz1979,ChristoCQG} precludes the existence of such naked singularities for generic gravitational collapse spacetimes.
	
	\paragraph{Trapped surface conjecture} In his monumental works on Weak Cosmic Censorship \cite{Christo1,Christo2,Christo3}, Christodoulou proved the instability of any locally-naked singularity with respect to the system \eqref{1.1}--\eqref{5.1}  with $F\equiv 0$ (namely, the Einstein-scalar-field model) in spherical symmetry. To do so, he obtained a proof of the  trapped surface conjecture, involving $b_{\Gamma}$ the endpoint of the center (see Figure~\ref{fig:one-ended}), whose simplified version is given below:	
	\begin{conjecture}[Spherical trapped surface conjecture, \cite{Kommemi,JonathanICM}] \label{trappedsurfaceconj}
		In generic gravitational collapse, a black hole spacetime features  a sequence of trapped surfaces asymptoting to $b_{\Gamma}$, the endpoint of the center $\Gamma$. As a consequence, the spacetime has no locally-naked singularity emanating from $\Gamma$.
	\end{conjecture} 
	A successful resolution of Conjecture~\ref{trappedsurfaceconj} will turn Theorem~\ref{thm.II} into an unconditional statement (assuming  the decay estimates \eqref{decay.intro} are also satisfied for generic solutions) and eliminate our assumptions regarding the absence of locally-naked singularities in gravitational collapse. 
	While  Conjecture~\ref{trappedsurfaceconj} (without assuming   $F\equiv 0$) remains open at present, see \cite{AnZhang} for progress towards its proof for \eqref{1.1}--\eqref{5.1} in spherical symmetry.

	\subsection{Global aspects of the black hole interior  in gravitational collapse} \label{breakdown.section}
	
	\paragraph{Breakdown of the weak null singularity}	The weak singularity obtained in Theorem~\ref{CH.thm.SS} is instrumental in proving the breakdown of the Cauchy horizon in gravitational collapse in \cite{breakdown}. We note that, combining Conjecture~\ref{trappedsurfaceconj} with the theorem on the breakdown of weak null singularities shows that the terminal boundary of the spacetime only consists of $\CH$ and a not-entirely-null component $\mathcal{S}=\{r=0\}$, where $r$ is the area-radius of the metric. In other words, the Penrose diagram is  then given by Figure~\ref{fig:spacelikeconj}, where, however, it is not known whether $\mathcal{S}$ is spacelike, and even whether $\mathcal{S}$ contains any spacelike portion (see \cite{Kommemi} for a  definition of $\mathcal{S}$).
	
	\begin{figure}[H]	\begin{center}
			\includegraphics[width=62.5 mm, height=50 mm]{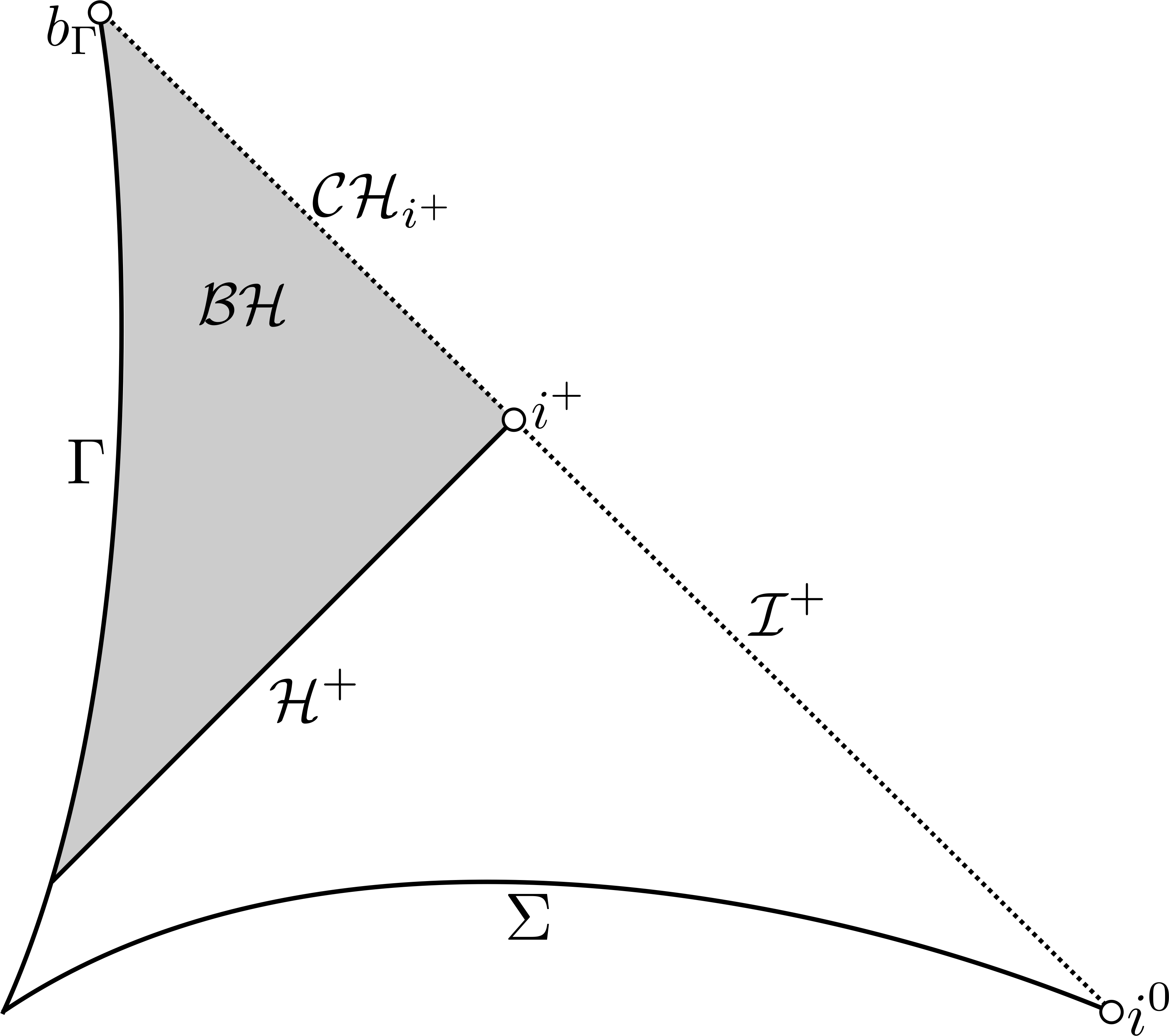}
		\end{center}
		\caption{The impossible Penrose diagram if $\CH$ is weakly singular, as a consequence of the result of \cite{breakdown}.}\label{fig:disproof}\end{figure}

	\paragraph{Towards a resolution of Conjecture~\ref{spacelike.conj}}	The new result, Theorem~\ref{thm.I}, applied to the global one-ended setting, however, shows that $\mathcal{S}$ is spacelike in the proximity of $\CH\cap\mathcal{S} $ and provides sharp and detailed quantitative estimates, as already stated in Theorem~\ref{thm.II}. It still not known, however, whether $\mathcal{S}$ is entirely spacelike up to $\mathcal{S}\cap \Gamma$, see Figure~\ref{fig:spacelikeconj}. To summarize, under the assumptions of Theorem~\ref{thm.II}, it is proven that the terminal boundary indeed contains a null singularity $\CH$ and a spacelike singularity $\mathcal{S}'\subset\mathcal{S}$ that coexist as claimed in Conjecture~\ref{spacelike.conj}. To complete the proof  of Conjecture~\ref{spacelike.conj}, one  additionally needs to study the black hole exterior and show that for generic, localized  asymptotically flat initial data, the decay assumptions \eqref{decay.intro} are satisfied on some outgoing cone $C_{out}$ in the black hole interior, see \cite{Moi2,MoiDejan} for progress in that direction. Note, however, that Theorem~\ref{thm.II} (Statement~\ref{II.B}) already provides an unconditional construction of  one-ended asymptotically flat black hole solutions for which \eqref{decay.intro} holds and the Penrose diagram features a spacelike singularity  $\mathcal{S}$ and a null Cauchy horizon $\CH$ as depicted in  Figure~\ref{fig:spacelikeconj}.
	
	\paragraph{Proof of the breakdown of weak null singularities: old and new} The logic of the breakdown of weak null singularities in \cite{breakdown} follows a contradiction argument.
	In our companion paper \cite{bif2}, we offer a simpler, constructive proof, although it requires slightly stronger assumptions than the result in \cite{breakdown}. It relies on a novel a priori estimate (see Section~\ref{section.proof}) which guarantees the existence of a trapped causal diamond right under the Cauchy horizon $\CH$, and thus precluding spacetimes with the Penrose diagram of Figure~\ref{fig:disproof}. These  a priori estimates also play a role in the  analysis of the present paper and are stated  in Proposition~\ref{apriori.prop1} in Section~\ref{apriori.est.section}.

	\subsection{Global aspects of the black hole interior  in the two-ended case} \label{twoended.section}
	
	Two-ended black holes have initial data with the topology $\mathbb{R}\times \mathbb{S}^2$ (see Section 3 in \cite{bif2} or the review \cite{review}) which is incompatible with the $\mathbb{R}^3$ topology of gravitational collapse. For this reason, they do not possess a center, i.e., $\Gamma=\emptyset$, and thus spherically symmetric locally-naked singularities cannot exist in such spacetimes \cite{MihalisSStrapped,nospacelike}.
	
	However, it is \emph{not true} that weak null singularities break down for two-ended black holes. Indeed, a result of Dafermos \cite{nospacelike} shows that two-ended small spherically-symmetric perturbations of the Reissner--Nordstr\"{o}m black hole for \eqref{1.1}--\eqref{5.1} do not have any spacelike singularity, i.e., $\mathcal{S}=\emptyset$. This result moreover generalizes to small perturbations of the Kerr black hole for the Einstein equations in vacuum \cite{KerrStab}. The  Penrose diagram is then the left-most in  Figure~\ref{fig:nospacelike} with $\CH$ as the only terminal boundary component.

	Our local Theorem~\ref{thm.I} also has applications to two-ended black holes, albeit only the ones  such that $\mathcal{S}\neq \emptyset$ depicted on the rightmost Penrose diagram of Figure~\ref{fig:nospacelike}. The following theorem, whose proof is contained in our companion paper \cite{bif2}, states that for such two-ended black holes, the quantitative estimates of Theorem~\ref{thm.I} apply (conditional result), and  also provides examples of such black holes if $q_0=m^2=0$ (unconditional result).

	\begin{theo} [Two-ended black holes with coexisting singularities, our companion paper \cite{bif2}]\label{thm.III}~ \begin{enumerate}
			\item (Theorem IV, Statement B in \cite{bif2}). Consider a spherically-symmetric two-ended black hole solution of \eqref{1.1}--\eqref{5.1} with the rightmost Penrose diagram of Figure~\ref{fig:nospacelike} and assume that \eqref{decay.intro} holds on an outgoing cone $C_{out}$ under the Cauchy horizon $\CH$. Then $\mathcal{S}$ is spacelike in a neighborhood of $\CH\cap \mathcal{S}$, where it obeys the Kasner asymptotitcs of Theorem~\ref{thm.I}. 
			\item  (Theorem III in \cite{bif2}). There exists a large class of spherically-symmetric two-ended black hole solution of \eqref{1.1}--\eqref{5.1} for $q_0=m^2=0$ with the rightmost Penrose diagram of Figure~\ref{fig:nospacelike}   satisfying the above assumptions  and conclusions.
		\end{enumerate}
	\end{theo}

	In Section~\ref{SCC.section}, we have already discussed the results of Luk--Oh \cite{JonathanStab,JonathanStabExt} achieving a proof of spherically-symmetric Strong Cosmic Censorship for \eqref{1.1}--\eqref{5.1} with $q_0=m^2=0$ . Interestingly, their argument does not operate any distinction between spacetimes with the leftmost diagram of Figure~\ref{fig:nospacelike}, and those with   rightmost diagram of Figure~\ref{fig:nospacelike}. This is because $C^2$-inextendibility though $\mathcal{S}=\{r=0\}$ is guaranteed due to the blow-up of Kretschmann scalar, and thus, no quantitative estimates on $\mathcal{S}$ are necessary for their proof of Strong Cosmic Censorship. Therefore, even in the $q_0=m^2=0$ case, Theorem~\ref{thm.III} provides the first  non-trivial quantitative estimates on the singularity $\mathcal{S}$ in two-ended spherically-symmetric black hole solutions of \eqref{1.1}--\eqref{5.1}.

	\begin{figure}[H]	\begin{center}
			\includegraphics[width=150 mm, height=50 mm]{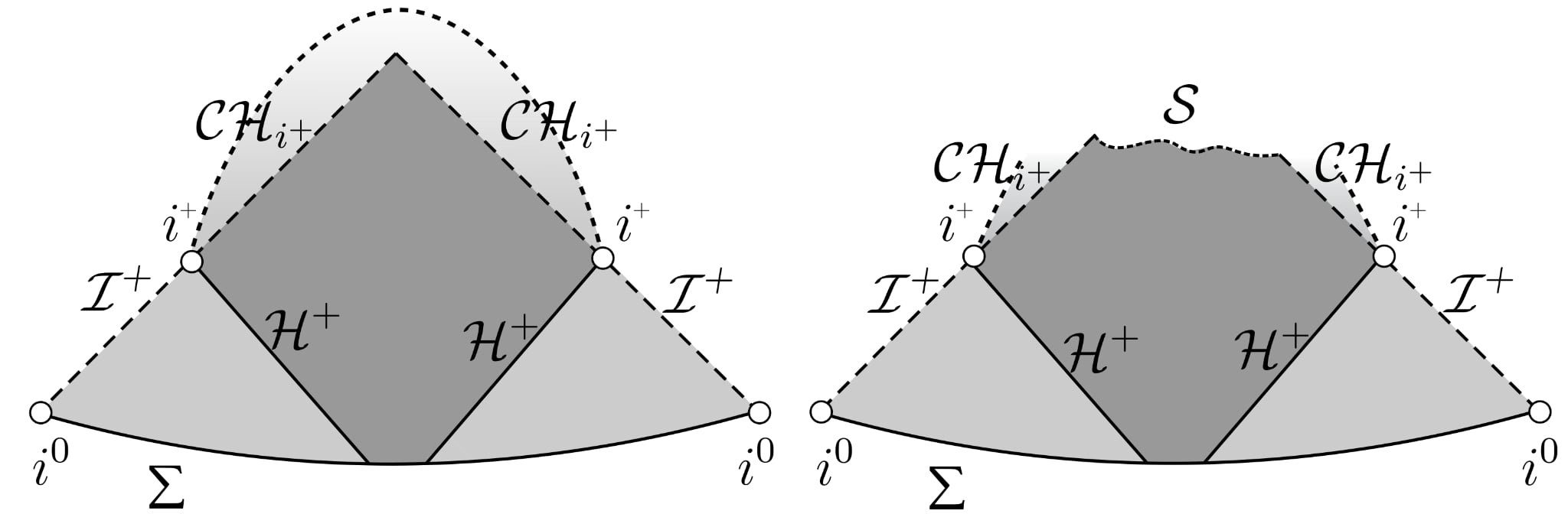}
		\end{center}
		\caption{Left: Two-ended black hole with a Cauchy horizon $\CH$ and no spacelike singularity. \\Right: Two-ended black hole with coexisting Cauchy horizon $\CH$ and   singularity $\mathcal{S}=\{r=0\}$.}\label{fig:nospacelike}\end{figure}
	
	\subsection{Dynamics of Kasner metrics and Kasner metrics with variable exponents} \label{Kasner.section}
	
	\paragraph{Properties of the Kasner solution} We recall the Kasner spacetime \eqref{Kasner.exact}, where $\mathcal{S}=\{\tau=0\}$ is a spacelike singularity if $(p_1,p_2,p_3) \notin \{(0,0,1), (1,0,0), (0,1,0)\}$, at which tidal deformations are infinite. The case  $(p_1,p_2,p_3) \in \{(0,0,1), (1,0,0), (0,1,0)\}$ is degenerate, however, and we recall that in this instance the Kasner spacetime is locally isometric to Minkowski spacetime, i.e., it is flat. The sign of $p_i$ moreover indicates the type of deformation: $p_i>0$  corresponds to a contraction in the $x_i$ direction, and  $p_i<0$ a dilation. If $\phi \equiv 0$ in \eqref{1.1}--\eqref{5.1}, the Kasner spacetime \eqref{Kasner.exact} is a solution of the Einstein vacuum equations and \eqref{p.Kasner} (where $p_{\phi}=0$) mandates that at least one of the $p_i$'s is negative. Note that this constraint disappears in the presence of a scalar field (i.e., with $p_{\phi}\neq0$),  or under $\mathbb{U}(1)$-symmetry, or in $(d+1)$ spacetime-dimension for $d\geq 10$, and thus there exists Kasner solutions where all $p_i$'s are positive in these cases, see e.g.\ the introduction of \cite{FournodavlosRodnianskiSpeck}.
	\paragraph{Comparison between the Schwarzschild  and the Kasner  spacelike singularities}	 The Schwarzschild metric \eqref{Schwarz} can be expressed in the unit proper-time gauge to compare $\mathcal{S}=\{r=0\}$ to the Kasner spacelike singularity $\{\tau=0\}$. This gauge change shows that asymptotically as $\tau \rightarrow 0$,  the Schwarzschild metric converges to a Kasner solution \eqref{Kasner.exact} with $(p_1,p_2,p_3)= (-\frac{1}{3},\frac{2}{3},\frac{2}{3})$, which explains  why we previously mentioned an infinite dilation in the radial direction and infinite contraction in the orthoradial directions at $\mathcal{S}=\{r=0\}$ -- the so-called ``spaghettification'' \cite{gravitation,Hawking}. We moreover note that, in view of the fact that the Oppenheimer--Snyder spacetime coincides with Schwarzschild \eqref{Schwarz} outside a spatially-compact region, the Kasner exponents associated to its spacelike singularity are also $(p_1,p_2,p_3)= (-\frac{1}{3},\frac{2}{3},\frac{2}{3})$, at least near timelike infinity $i^+$.

	\paragraph{Kasner-like metrics with variable exponents} As a generalization of the exact Kasner spacetime \eqref{Kasner.exact}, one can also  consider generalized Kasner solutions with variable exponents $(p_1(x),p_2(x),p_3(x))$ of the  form: \begin{equation}\label{Kasner.approx}\begin{split}
			&	g= -d\tau^2 + \sum_{1\leq i,j \leq 3} a_{i j}(\tau,x) \tau^{2 p_{\max\{i,j\}}(x)} dx^{i} dx^{j},\\ & \lim_{\tau \rightarrow 0}  a_{i j}(\tau,x) = c_{i j}(x).\end{split}
	\end{equation} where $(p_1(x),p_2(x),p_3(x))$ still satisfy the Kasner relation \eqref{p.Kasner}. While typically \eqref{Kasner.approx} is expressed for $(x_1,x_2,x_3) \in \mathbb{T}^3$, there also exist spherically-symmetric analogues of \eqref{Kasner.approx} where $p_2(x)=p_3(x)$. It was shown in \cite{Warren1}  (see also \cite{AnZhang}) that spacelike singularities occurring in spherically-symmetric solutions of \eqref{1.1}--\eqref{5.1} with $q_0=0$ (uncharged scalar field) are locally described by a Kasner-like metric with variable exponents of the above form.\\
	Moreover, a large class of solutions of the Einstein vacuum equations of the form \eqref{Kasner.approx} where $c_{i j}(x)$, $p_i(x)$ are merely assumed to be smooth functions, was constructed in the work of Fournodavlos--Luk \cite{FournodavlosLuk} assuming\footnote{\label{f}More precisely, it is enough to assume that none of the $p_i$'s vanish (which is assumed in \cite{FournodavlosLuk})  and use the compactness of $\mathbb{T}^3$ in the setting of \cite{FournodavlosLuk} to obtain \eqref{nondeg}. This non-degeneracy condition \eqref{nondeg} moreover plays an important role in the analysis of \cite{FournodavlosLuk}.} a non-degeneracy  condition of the following form on the Kasner exponents: there exists $\epsilon>0$ such that for all $x$ \begin{equation}\begin{split}\label{nondeg} 
			\min\{|p_1|(x),|p_2|(x),|p_3|(x)\} > \epsilon.
		\end{split}
	\end{equation}The results in \cite{Warren1,AnZhang}  are similarly restricted to Kasner-like metric obeying the non-degeneracy assumption \eqref{nondeg} (relatedly, see  footnote~\ref{f}). \\Finally, we  discuss  spherically-symmetric solutions of \eqref{1.1}--\eqref{5.1} in spherical symmetric with $F\equiv 0 $ (Christodoulou's model from \cite{Christo1,Christo2,Christo3}). Due to the absence of electromagnetism, the terminal boundary $\mathcal{S}$ is spacelike (in contrast to the general case, compare with Theorem~\ref{CH.thm.SS}) with no Cauchy horizon emanating from $i^+$.  The work of An--Gajic \cite{DejanAn} considers the asymptotic behavior of $\mathcal{S}$ towards $i^+$ and finds that it converges to the Schwarzschild singularity of \eqref{Schwarz}. Thus, $\mathcal{S}$ is a Kasner-like metric with non-degenerate exponents $(p_1(v),p_2(v),p_3(v))$ indexed by an Eddington--Finkelstein coordinate $v$ as in Theorem~\ref{thm.I}, where $i^+$ corresponds to $v=+\infty$ and $$ \lim_{v \rightarrow +\infty} (p_1(v),p_2(v),p_3(v)) = (-\frac{1}{3},\frac{2}{3},\frac{2}{3}),$$ consistently with the earlier discussion on the Kasner-like character of the Schwarzschild singularity $\mathcal{S}=\{r=0\}$.
	
	\paragraph{Kasner singularities with degenerate exponents} In Theorem~\ref{thm.I} and subsequent results, we construct \emph{the first non-trivial examples of Kasner-like metrics with degenerating exponents}. Such singularities were previously conjectured to occur at the junction between the spacelike singularity  and the Cauchy horizon in view of numerical results in  \cite{Chesler21}. We show precisely that, as $\tau \rightarrow0$, the metric takes the following form near $\mathcal{S}=\{\tau=0\}$ \begin{equation}\label{kasner.deg.intro}\begin{split}
			&		g\approx -d\tau^2 +  \tau^{2(1-2p(x))} dx^2+  \tau^{2p(x)} (d\theta^2+\sin^2(\theta) d\varphi^2),\\ & p(x) \approx x^{\frac{1}{2s-1}} \text{ as } x \rightarrow 0,  \end{split}
	\end{equation} where $\{\tau=x=0\}$ corresponds to $\CH\cap \mathcal{S}$. Note that, since $s>\frac{1}{2}$,  $(1-2p(x),p(x),p(x)) \rightarrow (1,0,0)$ as $x\rightarrow 0 $, therefore \eqref{nondeg} is indeed violated for the spacetimes of Theorem~\ref{thm.I}. It is instructive to express \eqref{kasner.deg.intro}  and the scalar-field asymptotics near $\mathcal{S}$ in the Eddington--Finkelstein coordinate $v$ already introduced in Theorem~\ref{thm.I} 
	
	\begin{equation}\label{kasner.deg.intro2}\begin{split}
			&	 p(u,v) \approx v^{-1},\\ & \phi(u,v) \approx \sqrt{v}\ \log(r^{-1}(u,v))\approx \frac{1}{\sqrt{v}}\ \log(\tau^{-1}) , \end{split}
	\end{equation}
	where we have introduced the area-radius $r(u,v)$. We remark on the following fundamental features of \eqref{kasner.deg.intro2}: \begin{itemize}
		\item The exponents  $(1-2p(v),p(v),p(v))$ are all positive, and thus in-falling observer experience infinite tidal contractions in all directions, in contrast to the Schwarzschild \eqref{Schwarz} and the Oppenheimer--Snyder models.
		\item $p(u,v)$ converges to $0$ at an inverse-polynomial rate: it is expected  due to the estimates of Theorem~\ref{CH.thm.SS}.
		\item The $v$-dependence of $p(u,v)$ and $\phi(u,v)$, however,  is surprisingly independent of $s$ introduced in \eqref{decay.intro}.
		
		\item It is unclear whether  it is possible to construct Kasner-like metrics of the form \eqref{kasner.deg.intro} with different inverse-polynomial (or even exponential) rates than those prescribed by \eqref{kasner.deg.intro2}.
	\end{itemize}
	We will return to the $s$-independence of \eqref{kasner.deg.intro2} when discussing the proof  in Section~\ref{section.proof}.

	\paragraph{Stability of Kasner singularities} In the breakthrough work of Fournodavlos--Rodnianski--Speck \cite{FournodavlosRodnianskiSpeck}, it is proven that Kasner spacetimes \eqref{Kasner.exact} with all positive $p_i$'s are stable with respect to small perturbations for the Einstein vacuum or scalar-field equations. However, it is conjectured that \eqref{Kasner.exact} with one of the $p_i$'s negative is unstable, since the celebrated BKL heuristics \cite{BKL1,BKL2}, see \cite{fluctuating,ringstrombianchi,Warreninv,Warreninv2} for related mathematical works.\\
	We remark that the $v$-dependent Kasner exponents in \eqref{kasner.deg.intro}, \eqref{kasner.deg.intro2} are  positive, which may lead to conjecturing the stability of $\mathcal{S}$ obtained in Theorem~\ref{thm.I} against \emph{non-spherical perturbations} for the Einstein-scalar-field system. \begin{rmks} As we will see in Section~\ref{section.proof}, the presence of a non-trivial scalar field $\phi$ is crucial in the proof of Theorem~\ref{thm.I}, beyond the mere fact that it enables the existence of Kasner metric of the form \eqref{Kasner.exact}  with all $p_i$ positive solving \eqref{1.1}--\eqref{5.1}.
		Recall, however, that there does not exist  Kasner  metrics \eqref{Kasner.exact}  with all $p_i$ positive  solving the Einstein equations in vacuum in $(3+1)$ spacetime dimensions. Because of this, the presence and dynamics of a spacelike singularity in vacuum, as well as the formulation of a conjecture analogous to Conjecture~\ref{spacelike.conj}, are unclear. Celebrated BKL heuristics  from the 70's \cite{BKL1,BKL2,mixmaster} describe oscillatory dynamics which could involve chaotic behavior. How this  scenario articulates with a presence of a Cauchy horizon in the dynamics of the vacuum Einstein equations is not answered in the current work and remain  open at present.
	\end{rmks}\noindent The proof of \cite{FournodavlosRodnianskiSpeck}, however, takes advantage of a non-degeneracy condition of the form \eqref{nondeg} (see footnote~\ref{f}) and thus does not apply to the spacelike singularities constructed in our Theorem~\ref{thm.I}. This is because the estimates employed in \cite{FournodavlosRodnianskiSpeck} lose an arbitrarily amount of derivatives when $\epsilon\rightarrow 0 $ in \eqref{nondeg}, similarly to the scheme employed in the backwards construction of solutions in \cite{FournodavlosLuk}, see  \cite{Warrendeg} for a related  general scattering result in this context. No such loss of derivatives is encountered, however, in the proof of Theorem~\ref{thm.I} (see Section~\ref{section.proof}).

	\subsection{Key ingredients in the proof of Theorem~\ref{thm.I}}\label{section.proof}
	We will use the following notation for the spherically-symmetric metric $g$ and electromagnetic field $F$.	\begin{equation} 
		g= -\Omega^2(u,v) du dv+ r^{2}(u,v)   (d\theta^2 + \sin(\theta)^2 d\varphi^2),\ F = \frac{Q(u,v)}{r^2(u,v)}\ \Omega^2(u,v) du \wedge dv, \label{SS.intro}
	\end{equation} 
	
	\subsubsection{General proof strategy}\label{proof.section1}
	
	We  prove that the Cauchy development of our initial data on $C_{out} \cup \Cin$ comprises a spacelike singularity $\mathcal{S}$ and a null Cauchy horizon $\CH$ as depicted in Figure~\ref{fig:local}. At the endpoint of the Cauchy horizon, denoted $(\uch,v=+\infty)$ (see Section~\ref{geometricframework}),  we prove that the area-radius $r$ is zero  (recalling the notations of \eqref{SS.intro}), i.e.,
	
	\begin{equation} \lim_{v \rightarrow +\infty} r(\uch,v)= \lim_{u \rightarrow \uch}  r_{|\CH}(u)=0, \end{equation}
	
	It is well-known that the system \eqref{1.1}--\eqref{5.1} in spherical-symmetry is energy-supercritical when the area-radius $r$ tends to $0$. Therefore, the precise $r$-weights and degenerations of $(g,\phi,F)$ will be crucial in our analysis.	The proof will be divided into two parts, which will involve distinct regimes for the Einstein equations: \begin{itemize}
		\item Below the Cauchy horizon, i.e., for $u\leq \uch$, where $r(u,v)$ tends to $0$ only as $(u,v) \rightarrow (\uch,+\infty)$.
		\item Below the spacelike singularity $\mathcal{S}$,  i.e., for $ \uch < u \leq \uch+ \ep $, where for every $u$, there exists $v_{\mathcal{S}}(u)<+\infty$ such that $(u,v_{\mathcal{S}}(u))\in \mathcal{S}$, hence $r(u,v_{\mathcal{S}}(u))=0$. Here, $\ep>0$ is a small constant.
	\end{itemize}
	The highly nonlinear nature of \eqref{1.1}--\eqref{5.1} typically requires a bootstrap method  to obtain quantitative estimates and the presence of a smallness parameter is necessary to close the argument. In the region below the spacelike singularity $\mathcal{S}$, this smallness parameter is $\ep$ and thus, we only describe $\mathcal{S}$ in a neighborhood of $\CH \cap \mathcal{S}$. 
	
	\subsubsection{A priori estimates under the Cauchy horizon}\label{proof.section2}

	Under the Cauchy horizon, there is, however, no analogous smallness parameters, which makes it difficult to close a bootstrap argument. We rely instead on novel a priori estimates of the following form (recalling  \eqref{SS.intro})	\begin{equation}\label{apriori.intro}
		\int_{u_0}^{u} \Omega^2(u',v) |\phi|^p(u',v) du' \lesssim_p e^{2K_- v},
	\end{equation}
	for any $p\geq 0$ and	with the usual Eddington--Finkelstein $v$-coordinate  choice (see Section~\ref{geometricframework}). The $e^{2K_- v}$ term, where $K_-<0$, already appears in the corresponding term at the Reissner--Nordstr\"{o}m Cauchy horizon: $$ \Omega^2_{RN}(u,v) \approx e^{2K_-(u+v)} \text{ as } v \rightarrow +\infty.$$
	
	The key feature of 	 \eqref{apriori.intro} is a control of  $\Omega^2$ by an exponential scalar-field flux, following from a careful application of the Raychaudhuri equation, that can absorb polynomial scalar-field powers like $|\phi|^p$ (see  Lemma~\ref{lemma.hardy}).

	\eqref{apriori.intro}  gives a priori control of many geometric quantities \emph{without any quantitative knowledge of the scalar field}  (see Proposition~\ref{apriori.prop1}). It is at the heart of our new proof of Cauchy horizon breakdown discussed in Section~\ref{breakdown.section} and proved in \cite{bif2}, and also constitutes a crucial ingredient in the nonlinear analysis of Theorem~\ref{thm.I}.
	
	One of the key estimates provides the rate at which $r$ tends to $0$ in terms of $v$, schematically given by \begin{equation}\label{r.intro}\begin{split} &  r^2(u,v) \approx (\uch-u) + O(v^{1-2s}), \\ & -r\rd_v r \approx v^{-2s},\  -r\rd_u r \approx  1.
		\end{split}
	\end{equation} In particular, we show that under the Cauchy horizon (i.e., for $u\leq \uch$), $r$ degenerates to $0$ at worst at an inverse polynomial rate in $v$. We remark that \eqref{r.intro} follows from the propagation equations on $\rd_u(r\rd_v r)=...$  and  $\rd_v(r\rd_u r)=...$  whose right-hand-side $...$ is a negligible error term.

	\subsubsection{Scalar field estimates under the Cauchy horizon} \label{proof.section3} A priori estimates on $(r,\Omega^2)$ are already given  in  Section~\ref{proof.section2}, so addressing the region under the Cauchy horizon reduces to propagating linear estimates on $\phi$ for a metric given by $(r,\Omega^2)$. The  quantitative behavior of $\phi$ originates from the propagation of estimates on $rD_v \phi$ in the ingoing direction, as follows:
	\begin{equation}\label{phi.intro}\begin{split} & |\phi|(u,v) \lesssim r^{-\frac{1}{2s-1}}(u,v), \\ & r|D_v \phi| \lesssim v^{-s},\  r|D_u\phi|(u,v) \lesssim r^{-\frac{s}{s-\frac{1}{2}}}(u,v),
		\end{split}
	\end{equation} where we recall that $s>\frac{1}{2}$. Note that \eqref{phi.intro} is consistent with the following blow-up behavior of the scalar field: \begin{equation}\label{phi.intro2}
		\phi(\uch,v) \approx \sqrt{v}.
	\end{equation}
	
	\subsubsection{Key geometric estimates near the spacelike singularity} \label{proof.section4} In the region $\{ u \geq \uch\}$, we prove that a spacelike singularity $\mathcal{S}=\{r=0\}$ develops. Solutions to the linear wave equation  on Schwarzschild $\Box_{g_s} \phi=0$ (recall \eqref{Schwarz}) blow-up at a logarithmic rate in $r$ as $r\rightarrow 0$, see e.g.~\cite{JanGreg} $$ \phi(t,r,\theta,\varphi) \sim A(t)\log(r) \text{ as } r \rightarrow 0.$$
	Consistently with \eqref{phi.intro2}, we prove an upper-bound estimate of the following form, for $\uch \leq u\leq \uch+\epsilon$: \begin{equation}\label{phi.intro3}
		|\phi|(u,v) \lesssim \sqrt{v}\ [ 1+ \log(\frac{r_0(v)}{r(u,v)})],
	\end{equation} where $r_0(v)= r(\uch,v) \approx v^{\frac{1}{2}-s}$ as given by \eqref{r.intro}.	Assuming  the metric coefficients $(r,\Omega^2)$ take a Kasner form \eqref{Kasner.approx}, \eqref{phi.intro3} is consistent with our following key bootstrap assumption, for some small $\delta>0$:  \begin{equation}\label{bootstrap.intro}
		\Omega^2(u,v) \lesssim e^{2K_- v}\cdot [r(u,v)]^{\delta v }.
	\end{equation}
	Such estimates were already used  in our previous work \cite{violent} in the spatially-homogeneous case, where a large constant $\epsilon^{-2}$ appeared in lieu of $v$. Following the strategy of \cite{violent}, our main idea is  then to exploit  \eqref{bootstrap.intro} to close all remaining estimates on the scalar field, such as \eqref{phi.intro3}, together with estimates on $r$ of the form \begin{equation}\label{lim.intro}
		\lim_{u \rightarrow u_{\mathcal{S}}(v)} -r\rd_v r(u,v) \text{ exists  and }   -r\rd_v r(u,v)  \approx v^{-2s},
	\end{equation} and its ingoing equivalent. Estimates such as \eqref{lim.intro} are crucial in obtaining a stable ``quiescent'' spacelike singularity, see for instance \cite{Warren1}, where \eqref{lim.intro} must be assumed a priori before proving Kasner asymptotics. We finally remark on the necessity to track both $r$ and $v$ weights, contrary to the local studies of Kasner singularities in which only $r$ weights need to be tracked, see for instance \cite{FournodavlosLuk,Warren1,Warreninv,Warreninv2,Ringstrom23}.

	\subsubsection{Precise scalar-field control near the spacelike singularity} \label{proof.section5} The important remaining task is to (im)prove \eqref{bootstrap.intro}, a task for which upper bounds such as \eqref{phi.intro3} are insufficient: we must also \emph{establish scalar field lower bounds}. In \cite{violent} the same issue occurs and it is resolved invoking linear scattering arguments and taking advantage to the proximity of the Reissner--Nordstr\"{o}m background, but no analogous method is available  in the spacetime region covered by Theorem~\ref{thm.I}. The key idea is to exploit an approximate monotonicity\footnote{This monotonicity property was first exploited by Dafermos  in his foundational  paper \cite{MihalisPHD}; however, for the uncharged scalar field case he considered, the monotonicity is exact, without any error terms, contrary to the estimate obtained in \eqref{monotonicity}.}  property satisfied by $r\rd_v \phi$ and $r\rd_u \phi$. More precisely, under \eqref{bootstrap.intro}, \eqref{r.intro} we have \begin{equation}\begin{split}\label{monotonicity}
			&	\rd_u (r\rd_v \phi) \approx -v^{-2s} \rd_u \phi+ O(e^{2K_- v}),\\ &  	\rd_v (r\rd_u \phi) \approx -\rd_v \phi+ O(e^{2K_- v}),
		\end{split}
	\end{equation} which implies, if the initial data $\rd_v \phi(u_0,v) \gtrsim v^{-s}$ is positive that, up to $r=0$, we have, for some $C>0$						
	\begin{equation}\label{lower.intro}r\rd_v \phi(u,v)\geq  C v^{-s}.			 \end{equation}

	While \eqref{lower.intro} is our key estimate and our only  scalar-field lower bound, it is not sharp, and insufficient  to prove \eqref{bootstrap.intro}! The sharp lower bound (see \cite{DejanAn} in the uncharged case)  relies on the point-wise propagation of $r^2 \rd_v \phi$ in the region $\{r\leq r_0(v)\}$, instead. However, for $r\geq \ep r_0(v) \approx \ep v^{\frac{1}{2}-s}$, \eqref{lower.intro} still gives   a sub-optimal bound \begin{equation}\label{lower.intro2}r^2\rd_v \phi(u,v)\geq  C\cdot \ep \cdot  v^{\frac{1}{2}-2s}.			 \end{equation}
	
	
	The next step in our argument is to show \eqref{lower.intro2} still holds in the region $\{r\leq \ep r_0(v)\}$, which does not follow from \eqref{lower.intro}. To this effect, we prove Asymptotically Velocity Term Dominated (AVTD)  estimates which are familiar in the context of spacelike singularities \cite{BKL1,BKL2}. As an example, we obtain an estimate of the following form: \begin{equation}\label{X.intro}
		|X\phi|(u,v) \lesssim v^{2s-\frac{1}{2}} ( 1+ \log^2( \frac{r_0(v)}{r(u,v)})), \text{ where } X= \frac{1}{-r\rd_v r } \rd_v - \frac{1}{-r\rd_v u } \rd_u.
	\end{equation} Up to the logarithmic terms, \eqref{X.intro} shows that $X\phi$ obeys a much better estimate than $\frac{D_v \phi}{-r \rd_v r}$ when $r\ll \ep r_0(v)$, comparing with \eqref{lower.intro2}. To obtain \eqref{X.intro}, we commute \eqref{1.1}--\eqref{5.1} with $X$, taking advantage of $X(r)=0$.

	The vector field $X$, together with this commutation strategy has been previously used in \cite{Warren1} in the context of non-degenerating Kasner singularities. A key difference, however,  is that in our context, the ingoing and outgoing directions are not interchangeable, and moreover $v$-weights need to be tracked  throughout the estimates.
	
	Finally, we use \eqref{X.intro}, taking advantage of the following estimate, a consequence of the wave equation:  \begin{equation}
		|\rd_u (r^2 \rd_v \phi)|\lesssim |X\phi|+O(e^{2K_- v}),
	\end{equation} which gives \begin{equation}\label{lower.intro3}
		\bigl| r^2 \rd_v \phi-r^2 \rd_v \phi(r=\ep r_0(v)) \bigr|\lesssim \ep^2 \log^2(\ep) \cdot v^{\frac{1}{2}-2s}.
	\end{equation}
	
	To conclude the proof of \eqref{lower.intro2} when $r\leq \ep r_0(v)$, we combine \eqref{lower.intro2} with \eqref{lower.intro3}, using crucially the bound $\ep^2 \log^2(\ep) \ll \ep$,   which improves \eqref{bootstrap.intro} and closes the bootstrap argument. We also show that $X$-derivatives of the metric $g$ also obey better estimates than regular derivatives near $\mathcal{S}$, consistently with the AVTD behavior.

	\subsubsection{Converting into Kasner-type estimates on the metric} \label{proof.section6} 
	To  cast the metric into the Kasner-type form \eqref{kasner.deg.intro}, we resort to a change of variable $(u,v) \rightarrow (\tau,x)$ such that				
	\begin{equation*}
		\tau(u,v)  \approx r^2(u,v) \Omega(u,v) v^{s-1},\  x(u,v) \approx v^{2-2s} \text{ as } r\rightarrow  0\text{ and } v\rightarrow+\infty, 
	\end{equation*}
	see Theorem~\ref{main.thm}, Statement~\ref{main.thm.II} for  details. We then define the third (and smallest) Kasner exponent $p(u,v)$ as $$ [\tau(u,v)]^{2p(u,v)}= r^2(u,v).$$
	
	Metric estimates such as \eqref{bootstrap.intro} and their improvements then show  the degeneracy of Kasner exponents: \begin{equation}
		p(u,v) \approx v^{-1} \text{ as } v\rightarrow+\infty.
	\end{equation}

	\subsubsection{Construction of initial data} \label{proof.section7}
	We first mentioned in  Remark~\ref{rmk.3} that we construct a large class of initial data on $C_{out} \cup \Cin$ such that the conclusion of Theorem~\ref{thm.I} applies to the resulting Cauchy development. The key condition is to arrange for $$ \Cin \text{ to be trapped and } r\rightarrow 0 \text{ towards the endpoint of } \Cin.$$ 
	
	We provide several constructions achieving this, including some with a large scalar field, see Remark~\ref{data.construction.remark}.
	
	\section*{Outline of the paper}
	
	In Section~\ref{geometricframework}, we introduce the necessary geometric preliminaries, together with relevant gauge choices and the Einstein--Maxwell--Klein--Gordon equations \eqref{1.1}--\eqref{5.1} in $(u,v)$-coordinates. In Section~\ref{thm.section}, we provide precise statements of our main results, together with material from previous works we will be using. In Section~\ref{apriori.section}, we provide a priori estimates in the region situated under the Cauchy horizon that are crucial to both our new results and a novel proof of Cauchy horizon breakdown. In Section~\ref{section.CH}, we prove quantitative estimates under the Cauchy horizon. In Section~\ref{quant.spacelike.section}, we obtain quantitative estimates near the spacelike singularity and close our main bootstrap assumptions. In Section~\ref{kasner.section}, we use the results of previous sections  to cast the metric in Kasner form and obtained more refined estimates which conclude the proof of Theorem~\ref{thm.I}. In Section~\ref{section.data}, we provide the construction of a large of initial data on which our quantitative estimates apply. 
	
	\section*{Acknowledgement} The author warmly thanks Amos Ori for inspiring conversations on black hole singularities. We also gratefully acknowledge the support from the NSF Grant DMS-2247376.
	
	\section{Geometric preliminaries} \label{geometricframework}

	The purpose of this section is to provide the precise setup, together with the definition of various geometric quantities, the coordinates and the equations that we will use throughout the paper.

	\subsection{Spherically symmetric solutions} \label{preliminary}
	We consider $(M,g,\phi,F)$, a regular solution of the system \eqref{1.1}, \eqref{2.1}, \eqref{3.1}, \eqref{4.1}, \eqref{5.1}, where $(M,g)$ is a Lorentzian manifold of dimension $3+1$, $\phi$ is a complex-valued function on $M$ and $F$ is a real-valued 2-form on $M$. 	$(M,g,\phi,F)$ is related to a quadruplet of scalar functions $ \{\Omega^2(u,v), r(u,v), \phi(u,v), Q(u,v)\}$, with $(u,v) \in \mathcal{Q}^+ \subset \RR^{1+1}$ by \begin{equation} \label{gdef}
		g= g_{\mathcal{Q}^+}+ r^{2} \cdot  (d\theta^2 + \sin(\theta)^2 d\varphi^2)= -\Omega^2(u,v) du dv+ r^{2}(u,v) \cdot  (d\theta^2 + \sin(\theta)^2 d\varphi^2),
	\end{equation} 
	$$ F(u,v)= \frac{Q(u,v)}{2r^{2}(u,v)}  \Omega^{2}(u,v) du \wedge dv.$$ 
	One can now formulate the Einstein equations \eqref{1.1}, \eqref{2.1}, \eqref{3.1}, \eqref{4.1}, \eqref{5.1} as a system of non-linear PDEs on $\Omega^2$, $r$, $\phi$ and $Q$ expressed in the double null coordinate system $(u,v) \in \mathcal{Q}^+  $: \begin{equation}\label{Omega}
		\partial_{u}\partial_{v} \log(\Omega^2)=-2\Re(D_{u} \phi \overline{D_{v}\phi})+\frac{ \Omega^{2}}{2r^{2}}+\frac{2\partial_{u}r\partial_{v}r}{r^{2}}- \frac{\Omega^{2}}{r^{4}} Q^2,
	\end{equation} \begin{equation}\label{Radius}\partial_{u}\partial_{v}r =\frac{- \Omega^{2}}{4r}-\frac{\partial_{u}r\partial_{v}r}{r}
		+\frac{ \Omega^{2}}{4r^{3}} Q^2 +  \frac{m^{2}r }{4} \Omega^2 |\phi|^{2}, \end{equation} 	\begin{equation}\label{Field}
		D_{u} D_{v} \phi =-\frac{ \partial_{v}r \cdot D_{u}\phi}{r}-\frac{\partial_{u}r \cdot  D_{v}\phi}{r} +\frac{ iq_{0} Q \Omega^{2}}{4r^{2}} \phi
		-\frac{ m^{2}\Omega^{2}}{4}\phi,\end{equation} 	\begin{equation} \label{chargeUEinstein}
		\partial_u Q = -q_0 r^2 \Im( \phi \overline{ D_u \phi}),
	\end{equation}	\begin{equation} \label{ChargeVEinstein}
		\partial_v Q = q_0 r^2 \Im( \phi \overline{D_v \phi}),
	\end{equation}
	\begin{equation} \label{RaychU}\partial_{u}(\frac {\partial_{u}r}{\Omega^{2}})=\frac {-r}{\Omega^{2}}|  D_{u} \color{black}\phi|^{2}, \end{equation} 
	\begin{equation} \label{RaychV}\partial_{v}(\frac {\partial_{v}r}{\Omega^{2}})=\frac {-r}{\Omega^{2}}|D_{v}\color{black}\phi|^{2},\end{equation} where the gauge derivative is defined by $D_{\mu}:= \partial_\mu+iq_0 A_{\mu}$, and the electromagnetic potential $A_{\mu}=A_u du + A_v dv$ satisfies 	\begin{equation}\label{Maxwell} \partial_u A_v - \partial_v A_u = \frac{Q \Omega^2}{2r^2}.\end{equation}	 Note that, under our electromagnetic gauge choice $A_v \equiv 0$ (see \eqref{A.gauge}), \eqref{Field} can also be written as
	\begin{equation}\label{Field4}
		\partial_{u}\partial_{v} \phi =-\frac{\partial_{u}\phi\partial_{v}r}{r}-\frac{\partial_{u}r \partial_{v}\phi}{r} +\frac{ q_{0}i \Omega^{2}}{4r^{2}}Q \phi
		-\frac{ m^{2}\Omega^{2}}{4}\phi- i q_{0} A_{u}\frac{\phi \partial_{v}r}{r}-i q_0 A_{u}\partial_{v}\phi.\end{equation}

	Subsequently, we define the Lorentzian gradient of $r$, and introduce the mass ratio $\mu$ by the formula $$ 1-\mu:=g_{\mathcal{Q}^+}(\nabla r,\nabla r),$$ where we recall that $g_{\mathcal{Q}^+}$ was the spherically symmetric part of $g$ defined in \eqref{gdef}. We can also define the Hawking mass:	$$ \rho := \frac{\mu \cdot r}{2} =\frac{r}{2} \cdot(1- g_{\mathcal{Q}^+} (\nabla r, \nabla r )).$$	
	
	Notice that the $(u,v)$ coordinate system, we have $g_{\mathcal{Q}^+} (\nabla r, \nabla r )= \frac{-4 \partial_u r \cdot \partial_v r}{\Omega^2}$. Now we introduce the modified mass $\varpi$ which involves the charge $Q$:
	
	\begin{equation} \label{electromass}
		\varpi := \rho + \frac{Q^2}{2r}= \frac{\mu r}{2} + \frac{Q^2}{2r} ,
	\end{equation}
	An elementary computation relates the previously quantities : 	\begin{equation} \label{murelation}
		1-\frac{2\rho}{r} = 1-\frac{2\varpi}{r}+\frac{Q^2}{r^2}=\frac{-4 \partial_u r \cdot \partial_v r}{\Omega^2}
	\end{equation}
	On the sub-extremal Reissner--Nordstr\"{o}m black hole \eqref{RN} of mass $\varpi \equiv M>0$, charge $Q \equiv e$ with $0<|e|<M$, we denote $r_+=M+\sqrt{M^2-e^2}$, the radius of the event horizon, its surface gravity $2K_+:= \frac{2}{r^2_+}(M- \frac{e^2}{r_+})>0$, and $r_-=M-\sqrt{M^2-e^2}$, the radius of the Cauchy horizon, its surface gravity $2K_-:= \frac{2}{r^2_-}(M- \frac{e^2}{r_-})<0$.

	Now we can reformulate our former equations to put them in a form that is more convenient to use. For instance, the Klein-Gordon wave equation \eqref{Field} can be expressed in different ways, using the commutation relation $[D_u,D_v]=\frac{ iq_{0} Q \Omega^{2}}{2r^{2}}$ 
	and  under our electromagnetic gauge choice $A_v \equiv 0$ (see \eqref{A.gauge}): 
	\begin{equation}\label{Field2}
		\rd_u \theta =  -\frac{\rd_v r}{r} \cdot\xi +\frac{ \Omega^{2} \cdot \phi}{4r} \cdot ( i q_{0} Q-m^2 r^2)- i q_0 A_u  r\rd_v \phi -i q_0 \partial_{v}r \cdot A_u \phi,
	\end{equation}  \begin{equation}\label{Field3}
		\rd_v \xi =  -\frac{\partial_{u}r }{r}\cdot \theta +\frac{ \Omega^{2} \cdot \phi}{4r} \cdot ( i q_{0} Q-m^2 r^2)- i q_0 A_u  r\rd_v \phi -i q_0 \partial_{v}r \cdot A_u \phi,
	\end{equation} where we introduced the notations $\theta=r\rd_v \phi$ and $\xi=r\rd_u \phi$.		Next, taking the real and imaginary parts of \eqref{Field2} and \eqref{Field3} we obtain, recalling that $\Re f$ and $\Im f$ denote the real and imaginary part of $f$ respectively:  \begin{equation}\label{Field2'}\begin{split}
			&		\rd_u \rtheta =  -\frac{\rd_v r}{r} \cdot\rxi -\frac{ \Omega^{2}}{4r} \cdot (  q_{0} Q \iphi+m^2 r^2\rphi)+ q_0 A_u \itheta + q_0 \partial_{v}r \cdot A_u \iphi,\\ & 	\rd_u \itheta =  -\frac{\rd_v r}{r} \cdot\ixi +\frac{ \Omega^{2}  }{4r} \cdot (  q_{0} Q \rphi-m^2 r^2 \iphi)-  q_0 A_u  \rtheta - q_0 \partial_{v}r \cdot A_u \rphi, \end{split}
	\end{equation}  \begin{equation}\label{Field3'}\begin{split}
			& 		\rd_v \rxi =  -\frac{\partial_{u}r }{r}\cdot \rtheta -\frac{ \Omega^{2}}{4r} \cdot (  q_{0} Q \iphi+m^2 r^2\rphi)+ q_0 A_u \itheta + q_0 \partial_{v}r \cdot A_u \iphi,\\ & 		\rd_v \ixi =  -\frac{\partial_{u}r }{r}\cdot \itheta +\frac{ \Omega^{2}  }{4r} \cdot (  q_{0} Q \rphi-m^2 r^2 \iphi)-  q_0 A_u  \rtheta - q_0 \partial_{v}r \cdot A_u \rphi.\end{split}\end{equation}

	We can also re-write  
	\eqref{Radius}, introducing the notation $\lambda=\rd_v r$ and $\nu=\rd_u r$: 	
	\begin{equation} \label{Radius3}
		-\partial_u  \partial_v (\frac{r^2}{2})	=\partial_u (-r \lambda) =	-\partial_v (r\nu ) = \frac{\Omega^2}{4}\cdot (1- \frac{Q^2}{r^2}-m^2 r^2 |\phi|^2).
	\end{equation}

	\subsection{Double null coordinate choice}

	We will work on a spacetime as in Figure~\ref{fig:local} which is the Cauchy development of bifurcate null hypersurfaces $C_{out}\cup \Cin$, on which we must specify gauge conditions for the $(u,v)$ coordinates. 
	
	On $C_{out}=\{u_0\}\times [v_0,+\infty)$, 	we fix the gauge condition on $v$ to be
	\begin{equation} \label{gauge1}
		\Omega^2(u_0,v)  = \Omega^2_{RN}(u_0,v) \sim e^{2K_-(M,e)(u_0+v)} \text{ as } v\rightarrow+\infty.
	\end{equation}
	We will impose $	r\rd_u r \equiv -1$ on  the future terminal boundary  of the spacetime,	more precisely, for fixed $u$, we denote $v_{\mathcal{B}^+}(u)$ to be the supremum of $v$ such that $(u,v) \in \mathcal{Q}^+$, and we impose 
	
	\begin{equation} \label{gauge2}
		r\rd_u r(u,v_{\mathcal{B}^+}(u)) \equiv -1.
	\end{equation} Note that imposing gauge \eqref{gauge2} requires to first prove that 	$r\rd_u r(u,v_{\mathcal{B}^+}(u)) <0$ in any regular gauge choice  on $\Cin=[u_0,u_F) \times \{v_0\}$, which is easy to show under the following assumption on $C_{out}$  $$ \lim_{v\rightarrow+\infty} \rd_u r(u_0,v) <0$$ that we will always require.
	
	In particular, in the past of the Cauchy horizon $\CH$, we impose $$ \lim_{v \rightarrow+\infty} r\rd_u r(u,v)=-1,$$ and to the future of the Cauchy horizon $\CH$, we impose $$ \lim_{v \rightarrow v_{\mathcal{S}}(u)}r\rd_u r(u,v)=-1,$$ where $\mathcal{S}=\{(u,v_{\mathcal{S}(u)}),\ u\geq \uch\}$, where $\uch$ is the $u$-coordinate of $\CH \cap \mathcal{S}$, the endpoint of the Cauchy horizon $\CH$ (see Section~\ref{thm.section}). We still have a translation gauge freedom $u_0 \rightarrow u_0 +a$, that we use to fix \begin{equation}
		\uch = 0,
	\end{equation} and therefore $u_0<0$.

	\subsection{Electromagnetic gauge choice, and gauge invariant estimates}
	
	The system of equations \eqref{1.1}, \eqref{2.1}, \eqref{3.1}, \eqref{4.1}, \eqref{5.1} is invariant under the gauge transformation : $$ \phi \rightarrow  e^{-i q_0 f } \phi ,$$
	$$ A \rightarrow  A+ d f. $$
	where $f$ is a smooth real-valued function.  As it well-known,  $|\phi|$ and $|D_{\mu}\phi|$ are gauge invariant. Throughout the paper, we make the following electromagnetic gauge choice \begin{equation}\label{A.gauge}
		A_v(u,v)\equiv 0.
	\end{equation}
	
	\eqref{A.gauge} still leaves the gauge freedom to fix $A_u$ on a hypersurface, in view of the fact that \eqref{Maxwell} with \eqref{A.gauge} gives $\rd_v A_u =-\frac{Q\Omega^2}{2r^2} $. We always will fix $A_u=0$ on the future terminal boundary on the spacetime under consideration, namely, we impose (in the same fashion as for \eqref{gauge2}) \begin{equation}\label{A.gauge2}
		\lim_{v \rightarrow v_{\mathcal{B}^+}(u)}	A_u(u,v)\equiv 0.
	\end{equation}

	In particular, in the past of the Cauchy horizon $\CH$, we impose $$ \lim_{v \rightarrow+\infty} A_u(u,v)=0,$$ and to the future of the Cauchy horizon $\CH$, we impose $$ \lim_{v \rightarrow v_{\mathcal{S}(u)}} A_u(u,v)=0,$$ where $\mathcal{S}=\{(u,v_{\mathcal{S}}(u)),\ u\geq \uch\}$.

	\subsection{Trapped region and apparent horizon}

	By our assumptions on the initial data (see Theorem~\ref{main.thm}), we  derive that for $v_0$  large enough, we have on $C_{out}$ \begin{equation}
		\rd_u r(u_0,v)<0 \text{ for all } v\geq v_0.
	\end{equation}
	Therefore, by the Raychaudhuri equation \eqref{RaychU},  there is no anti-trapped surface in the development on $C_{out}\cup \Cin$, namely  $\partial_u r(u,v)<0$ for all $(u,v)$.

	We define the trapped region $\mathcal{T}$, the regular region $\R$ and the apparent horizon $\A$ as 
	\begin{enumerate}
		\item  \label{characttrapped}$(u,v) \in \T$ if and only if $\partial_v r(u,v)<0$ if and only if $1-\frac{2\rho(u,v)}{r(u,v)}<0$,
		\item $(u,v) \in \R$ if and only if $\partial_v r(u,v)>0$ if and only if $1-\frac{2\rho(u,v)}{r(u,v)}>0$,
		\item $(u,v) \in \A$ if and only if $\partial_v r(u,v)=0$ if and only if $1-\frac{2\rho(u,v)}{r(u,v)}=0$.
	\end{enumerate}	
	
	\subsection{Notation}
	
	We will write $A\ls B$ if there exists a constant $C>0$ such that $A \leq B$, and $A\approx B$ if there exists two constants $C_{\pm}>0$ such that $C_- A \leq B \leq C_+ A $.

	\section{Precise statements of the main theorem}\label{thm.section}

	The following theorem is the main result in our paper, and formulated for local initial data on bifurcate hypersurfaces $\Cin \cup C_{out} = [u_0,u_F)\times \{v_0\} \cup \{u_0\}\times [v_0,+\infty)$ as depicted in Figure~\ref{fig:local}.
	
	Our assumptions on $C_{out}$ correspond to the behavior on an outgoing  cone inside a black hole whose event horizon   converges to Reissner--Nordstr\"{o}m at the rate predicted by Price's law and were obtained in \cite{Moi,Moi4}. These assumptions were already discussed in the introduction and will not lead to further discussions in this paragraph. On the other hand, it is very difficult to find \emph{all} possible initial conditions on $\Cin$ leading to a breakdown of the Cauchy horizon. Instead, we thus formulate three statements of our main result. In short: \begin{enumerate}[i.]
		\item The first statement is independent of whether a breakdown of the Cauchy horizon occurs or not.
		
		\item The second statement is conditional and assumes that a breakdown of the Cauchy horizon has taken place within the domain of dependence of $\Cin \cup C_{out}$, and describes it quantitatively (this second statement is our most general result).
		
		\item The third statement is unconditional and provides a large class of initial data on $\Cin$ leading to a breakdown of the Cauchy horizon, which is then also described quantitatively applying the second statement.
	\end{enumerate}
	
	For the benefit of the reader, we now describe these three statements in more detail. \begin{enumerate}[i.]
		\item Statement~\ref{main.thm.I} consist of estimates in the past of the Cauchy horizon $\mathcal{CH}_{i^+}$. It is itself divided into two parts: the first consists of \emph{a priori estimates} that do not require precise control on the scalar field under the weak  assumptions \eqref{hyp1}, which is surprising at first but follows from a new estimate described in Section~\ref{proof.section2}. 		 These a priori estimates are essential to the analysis and also offer a new, simpler proof of the breakdown of the weak null singularity  of \cite{breakdown} obtained in our companion paper \cite{bif2}. The second part of Statement~\ref{main.thm.I} additionally requires the decay of the scalar field on $C_{out}$ (assumption~\ref{hyp2}) and provides sharper scalar-field upper bounds up to the future endpoint of the Cauchy horizon towards which $r$ potentially tends to $0$. Strictly speaking, since it is a local result, there is no guarantee that the breakdown of the Cauchy horizon already occurs in the causal rectangle $[u_0,u_F] \times [v_0,+\infty)$. 
		Finally, we will additionally show as part of the proof as Statement~\ref{main.thm.I}, that, under our assumptions, there cannot exist any ingoing null boundary component $\mathcal{S}_{i^+}$ on which $r$ extends to $0$.
		\item In Statement~\ref{main.thm.II}, we assume, in contrast, that the local initial data has been sufficiently extended so that a breakdown of the Cauchy horizon  occurs in the domain of dependence of the initial data $[u_0,u_F] \times [v_0,+\infty)$. Under the additional quantitatives assumptions \eqref{hyp4}-\eqref{hyp3} on the initial data, we then show the existence of a spacelike singularity $\mathcal{S}=\{r=0\}$ and provide sharp estimates \eqref{K1}--\eqref{K4} on the solution $(g,\phi)$ in its vicinity, proving that it is well-approximated by a Kasner metric with variable exponents $(1-2p(u,v),p(u,v),p(u,v))$ degenerating to $(1,0,0)$ with $p(u,v) \approx v^{-1}$ as $v\rightarrow +\infty$.
		\item  Statement~\ref{main.thm.II} is our most general statement and it is conditional on the occurrence of a Cauchy horizon breakdown. Statement~\ref{main.thm.III}, on the other hand, is an unconditional result in which we construct a large class of initial data on  $\Cin$ so that a breakdown of the Cauchy horizon described in Statement~\ref{main.thm.II} takes place for any $C_{out}$ satisfying the standard assumptions \eqref{hyp1}, \eqref{hyp2}, \eqref{hyp4}-\eqref{hyp3}. We thus obtain the same Kasner asymptotics and conclusion as in Statement~\ref{main.thm.II}. The construction is arranged so that $\Cin$ is trapped and $r\rightarrow 0$ towards the future endpoint of $\Cin$ (see Remark~\ref{data.construction.remark} for further details on the construction).
		
	\end{enumerate}

	\begin{thm}\label{main.thm}
		Let $u_0 < u_F$ and $v_0 \geq 1$.  Let $\uch= \sup\{u \in [u_0,u_F],\ \underset{v\rightarrow +\infty}\lim r(u,v) >0 \} $. We denote $\underset{v\rightarrow +\infty}\lim r(u,v)= r_{CH}(u)$ for all $u\in [u_0,\uch]$. Assume  the following estimates hold on $C_{out}=\{u_0\}\times [v_0,+\infty)$:
		
		\begin{equation}	\label{hyp1}\begin{split}  &r_{CH}(u_0) >0,\ \lim_{v \rightarrow +\infty} \rd_u r(u_0,v) <0, \\
				&  				L_-\ v^{-2s}\leq -r\partial_v r(u_0,v) \leq L_+\ v^{-2s}, \\ & |\phi|^2(u_0,v),\ |Q|(u_0,v) \leq D,
			\end{split} 
		\end{equation}

		where $s>\frac{1}{2}$,  $L_{\pm}>0$. Then, assuming $v_0$ is large enough with respect to  the constants involved in \eqref{hyp1}, \begin{enumerate}[i.]
			\item \label{main.thm.I} $\uch \in (u_0,u_F]$ and $[u_0,\uch] \times [v_0,+\infty) \subset \T$. Moreover, there exists $D_->0$, $D_+>0$, such that for all  $u \in [u_0,\uch]$, $v\geq v_0$ \begin{equation} \label{quant1}\begin{split} &   D_-\ v^{-2s}\leq -r\rd_v r(u,v) \leq D_+\ v^{-2s},\\ &
					r^2_{CH}(u) + \frac{2D_- }{2s-1}\cdot v^{1-2s} \leq	r^2(u,v)\leq   r^2_{CH}(u) +  \frac{2D_+ }{2s-1}\cdot v^{1-2s}.\end{split}
			\end{equation} 
			If, moreover, there exists $s>1$ such that the initial data satisfy \eqref{hyp1} and the additional estimate: \begin{equation}\label{hyp2}
				\begin{split}
					|D_v \phi|(u_0,v) \leq \tilde{D} \cdot v^{-s}
				\end{split}
			\end{equation} for some $\tilde{D}>0$, then the following spacetime estimates are satisfied: for all $u \in [u_0,\uch]$, $v\geq v_0$: 
			\begin{equation}\label{quant2}
				\begin{split}&	
					|\phi|(u,v) \lesssim r^{-\frac{1}{2s-1}}(u,v),\\ 
					&
					r |D_u\phi|(u,v) \lesssim r^{-\frac{s}{s-\frac{1}{2}}}(u,v),
					\\ &
					r| D_v \phi|(u,v)\lesssim v^{-s},\\ & |Q|(u,v) \leq \check{D},
				\end{split}
			\end{equation} for some $\check{D}>0$.
			
			\item \label{main.thm.II}	If  $\uch< u_F$ (breakdown assumption), then for all $ \uch <u \leq u_F$, there exists $v_{\mathcal{S}}(u)<+\infty$ such that  \begin{equation}
				\lim_{v \rightarrow v_{\mathcal{S}}(u)} r(u,v) =0 \text{, } \lim_{v\rightarrow +\infty} r(\uch,v)=0 \text{ and } \lim_{u \rightarrow \uch} r_{CH}(u)=0.\end{equation}
			
			If, moreover, there exists $s>1$ such that the initial data satisfy \eqref{hyp1}, \eqref{hyp2} and the following additional estimate holds: there exists $D_L>0$, $D_C>0$, $\delta>0$ such that for all $v\geq v_0$:   \begin{equation}\label{hyp4}
				\begin{split}& 	|D_{v v}^2 \phi|(u_0,v) \leq D_{C} \cdot v^{-s-1},
			\end{split}	\end{equation} \begin{equation}\label{hyp5}
				|\Im(\bar{\phi} D_v \phi)|(u_0,v) \leq D_{C} \cdot|\phi|(u_0,v) \cdot v^{-s-\delta},
			\end{equation}		 \begin{equation}\label{hyp3}
				\begin{split}&  |D_v \phi|(u_0,v) \geq  D_L \cdot  v^{-s}.
				\end{split}
			\end{equation}

			Then, defining  $\mathcal{S}=\{(u,v_{\mathcal{S}}(u)),\ u\in (\uch,u_F)\}$, there exists $0<\ep< u_F -\uch$  such that \\$\mathcal{S} \cap (\uch,\uch+\ep)$ is spacelike with the following estimate for all $\uch<u \leq \uch +\ep$: \begin{equation}\label{K1}
				v_{\mathcal{S}}(u) \approx \left( u- \uch\right)^{-\frac{1}{2s-1}},\ v'_{\mathcal{S}}(u) \approx -\left( u- \uch\right)^{-\frac{2s}{2s-1}}.
			\end{equation}
			Moreover, the metric takes the following approximate Kasner form:  there exists coordinates $(x,\tau)$ so that $\mathcal{S}=\{\tau=0\}$, $\mathcal{S}\cap \CH=\{\tau=0,\ x=0\}$  and   $x_0\geq 0$ small enough so that for all 
			$0 \leq x \leq x_0$, $\tau\geq0$: 
			\begin{equation}\label{K2}\begin{split}
					&	g= - (1+ \mathcal{E}_T(\tau,x)) d \tau^2 + \tau^{2 (1-2p(\tau,x))}(1+ \mathcal{E}_X(\tau,x)) dx^2 + \tau^{2p(\tau,x)} ( d\theta^2+\sin^2(\theta)d\varphi^2 ) ,\\   &  \bigl| p(\tau,x)-p(0,x)\bigr|\ls \frac{|\log|(x)}{|\log|(\tau)},\ p(0,x) \approx x^{\frac{1}{2(s-1)}},
					\\ & \phi(\tau,x) = p_{\phi}(x) \left( \log(\frac{x^{\frac{2s-1}{2(s-1)}}}{\tau})+ \mathcal{E}_\phi(\tau,x)\right)+\varXi_{\mathcal{S}}(x),\\ & |p_{\phi}|(x) \approx x^{-\frac{1}{4(s-1)}},\ |\varXi_{\mathcal{S}}|(x)  \ls x^{-\frac{1}{4(s-1)}},  \\ & |\mathcal{E}_T|(\tau,x),\  |\mathcal{E}_X|(\tau,x),\ |\mathcal{E}_\phi|(\tau,x) \lesssim \frac{\tau^{2p(\tau,x)}}{x^{\frac{2s-1}{s-1}}}\left[1+\log^2(\frac{\tau^{2p(\tau,x)}}{x^{\frac{2s-1}{s-1}}})\right], \end{split}
			\end{equation}  where $p_\phi(x) \in \mathbb{C}$ satisfies the usual Kasner relations \begin{equation}\begin{split}\label{K3}
					&  p_1^2(\tau,x)+ p_2^2(\tau,x)+ p_3^2(\tau,x)+2 |p_\phi|^2(x)=1,\\ & p_1(\tau,x) = 1-2p(\tau,x),\ p_2(\tau,x) = p_3(\tau,x) = p(\tau,x), \text{ i.e., }  p_1(\tau,x)+ p_2(\tau,x)+ p_3(\tau,x)=1. 
				\end{split}	
			\end{equation}
			in the coordinate system $(\tau,x,\theta,\varphi)$, which relates to  $(u,v,\theta,\varphi)$ in the following way, with $(u,v)=(\uch,+\infty)$ corresponding to $(\tau,x)=(0,0)$ and $\mathcal{S}=\{r=0\}=\{\tau=0\}$ and defining $ x_{\mathcal{S}}(v):=  \underset{u \rightarrow u_{\mathcal{S}}(v)}{\lim}x(u,v)$: \begin{equation}\begin{split}\label{K4}
					&\tau(u,v) =[ r(u,v)]^{p^{-1}(u,v)},\\ & \bigl| x(u,v) - x_{\mathcal{S}}(v)\big| \lesssim \frac{r^2(u,v)}{r_0^2(v)}\left[1+\log^2(\frac{r^2(u,v)}{r_0^2(v)})\right],\  x_{\mathcal{S}}(v) \approx v^{2(1-s)}.
				\end{split}
			\end{equation}
			
			In particular, the Kasner exponents and scalar field  obey the following estimates in $(u,v)$ coordinates: \begin{equation}\begin{split}\label{K5}
					&	p(u,v) \approx v^{-1},\ |p_{\phi}|(v) \approx v^{\frac{1}{2}},\ \phi(u,v) = p_{\phi}(v) \log(\frac{r_0(v)}{r(u,v)})+ \tilde{\varXi}_{\mathcal{S}}(v),\  |\tilde{\varXi}_{\mathcal{S}}|(v)\ls v^{\frac{1}{2}},\\ & |\mathcal{E}_T|(u,v),\  |\mathcal{E}_X|(u,v),\ |\mathcal{E}_\phi|(u,v) \lesssim \frac{r^2(u,v)}{r_0^2(v)}\left[1+\log^2(\frac{r^2(u,v)}{r_0^2(v)})\right], \end{split}
			\end{equation} where $r_0(v):= r(\uch,v) \approx v^{\frac{1}{2}-s}$ as a consequence of \eqref{quant1}.

			\item  \label{main.thm.III}If we assume that $[u_0,u_F) \times\{v_0\} \subset \T\cup \A$ and $\underset{ u \rightarrow u_F}{\lim} r(u,v_0)=0$, then $\uch < u_F$, so Statement~\ref{main.thm.II} holds.
			
			Moreover, there exist (a large class of) initial data $\underline{C}_{in} \cup C_{out}=[u_0,u_F) \times [v_0,+\infty)$ such that the assumptions \eqref{hyp1}, \eqref{hyp2}, \eqref{hyp4}-\eqref{hyp3} hold on   $C_{out}$ and $\Cin \subset \T$ with $\underset{ u \rightarrow u_F}{\lim} r(u,v_0)=0$.    So, for such initial data,  the conclusion of Statement~\ref{main.thm.II} holds. These initial data are constructed as such: starting from the gauge choice    \begin{equation}\begin{split}
					& A_u(\cdot,v_0)=0,\\ & -r\partial_u r(\cdot,v_0) \equiv-1\ \& \lim_{u\rightarrow u_F} r(u,v_0)=0, \text{ or equivalently } u_F - u = \frac{r^2(u,v_0)}{2} ,\label{Kasner.data.construction1}\end{split}
			\end{equation}   we then assume Kasner-like scalar field asymptotics, in  that there exists a constant $|\psi_0|>1$ such that 
			\begin{equation} \label{Kasner.data.construction2}\begin{split}
					& |\phi|(u,v_0) \lesssim \log(r^{-1})(u,v_0),\ \frac{ |D_u \phi|}{|\rd_u r|}(u,v_0)  \ls r^{-1}(u,v_0),\\ & \liminf_{u\rightarrow u_F}\frac{  r|D_u \phi|(u,v_0)}{|\rd_u r|(u,v_0)} \geq |\psi_0|.\end{split}
			\end{equation}
			
			Under these assumptions, the following  quantities are well-defined: \begin{equation}\begin{split}\label{norms.def}
					&I(\phi)=  \sup_{u_0 \leq u < u_F } \left( r^2(u,v_0)[m^2 r^2 |\phi|^2(u,v_0)-1] +  [Q(u,v_0)+q_0\int_{u_0}^{u} r^2 \Im(\bar{\phi}D_u\phi)(u',v_0) du' ]^2\right), \\ & N_0(\phi)= \int_{u_0}^{u_F} r^{-2}(u,v_0) 
					\exp(-\mathcal{F}(u)) |\rd_u r|(u,v_0)  du,
			\end{split}	\end{equation}  
			where we have introduced the notation $\mathcal{F}(u)=\int_{u_0}^{u} \frac{r|D_u \phi|^2(u',v_0)}{|\rd_u r|(u',v_0)}  du'$. Then, we make the additional quantitative assumption that 
	
	\begin{equation}\label{construction.main.assumption}
		2 \rho(u_0,v_0) -r(u_0,v_0)  >   I(\phi)N_0(\phi),
	\end{equation}
	which is, in particular,  satisfied if $2\rho(u_0,v_0) -r(u_0,v_0)$ is large with respect to  $|Q|(u_0,v_0)$ and $\phi(\cdot,v_0)$. Moreover, \eqref{construction.main.assumption} is also satisfied for a  class of large scalar field initial data, specifically with Kasner asymptotics of the form \begin{equation}\label{kasner.data.intro}
		\phi(u,v_0) = \Psi_0 \log(r^{-1}(u_0,v)) + \tilde{\phi}(u,v_0),
	\end{equation} where $\tilde{\phi}(u,v_0)$  is bounded and $|\Psi_0|$ is sufficiently large.

\end{enumerate}

\end{thm}

\begin{rmk}\label{data.construction.remark}  The construction of initial data sketched in Statement~\ref{main.thm.III} will be expanded in more details in Proposition~\ref{prop.data} in Section~\ref{section.data}, in which multiple examples satisfying \eqref{construction.main.assumption} will be constructed. There are, in particular, five different constructions we carry out, two of which (\ref{IV.data.intro} and \ref{V.data.intro} below) surprisingly allow for \emph{large scalar field initial data} (see  Section~\ref{section.data} for details). These five possible constructions are outlined below: \begin{enumerate}[A.]
	\item \label{I.data.intro} Large initial Hawking mass $\rho(u_0,v_0)$. 
	\item  \label{II.data.intro} Small  coupling parameters $(q_0,m^2)$ and initial charge $|Q|(u_0,v_0)$.
	\item  \label{III.data.intro} Small scalar field $|\phi|(\cdot,v_0)$  and initial charge $|Q|(u_0,v_0)$, assuming the Kasner-asymptotics ansatz \eqref{kasner.data.intro}.
	\item  \label{IV.data.intro} Large Kasner exponent $|\Psi_0|$, assuming the Kasner-asymptotics ansatz \eqref{kasner.data.intro}.
	\item \label{V.data.intro} Perturbation of exact Kasner asymptotics of the form \eqref{kasner.data.intro} with  fixed $\Psi_0$, small $\tilde{\phi}$, small initial charge $|Q|(u_0,v_0)$ and small Klein--Gordon mass $m^2$.

\end{enumerate} 

The key property in the construction of initial data in Statement~\ref{main.thm.III} above is to ensure that $\Cin$ is trapped, which cannot be directly imposed (however, the entire outgoing cone $C_{out}$  is  trapped by assumption).		

Note that the first two situations do not require the Kasner-asymptotics ansatz \eqref{kasner.data.intro}, however in all constructions the weaker  \eqref{Kasner.data.construction2} is essential: it is the condition of  \emph{stable Kasner asymptotics}, see Section~\ref{Kasner.section}.

Construction~\ref{I.data.intro} follows  immediately from \eqref{construction.main.assumption}; however, genuine  small data constructions are more subtle.

\ref{II.data.intro} is a small data construction which can be viewed as a perturbation of the uncharged massive scalar field case of Christodoulou, but assumes small parameters $(q_0,m^2)$, which is quite restrictive.  

The smallness of the charge $|Q|(\cdot,v_0) $ is essential in constructing a trapped ingoing cone $\Cin$, and obtained requiring the smallness of both the scalar field $|\phi|(\cdot,v_0) $ and the initial charge $|Q|(u_0,v_0) $ in Construction~\ref{III.data.intro}. An important difficulty in Construction~\ref{III.data.intro} is  to obtain the smallness of $\phi$ while retaining Kasner asymptotics of the form \eqref{Kasner.data.construction2} which are essential to secure the finiteness of the norms in \eqref{norms.def} (in particular the lower bound on $|D_u \phi|$): to do this, we must assume the more precise Kasner form \eqref{kasner.data.intro}. The norm in which to measure the smallness of $\phi$ is then quite subtle, as it needs to be small for profiles of the form \eqref{kasner.data.intro}, even for non-small $\Psi_0$ (recall $|\Psi_0|>1$). We end up requiring the smallness of the following  scale-critical norm in Construction~\ref{III.data.intro}: $$ \| \frac{r D_u \phi}{\rd_u r} \|_{crit}= \sqrt{ \int_{u_0}^{u_F} r  [1+ \log(\frac{r(u_0,v_0)}{r(u,v_0)})]^{-4} \frac{|D_u \phi|^2(u,v_0)}{|\rd_u r |(u,v_0)} du}.$$

Construction~\ref{IV.data.intro} is a large scalar-field data construction, in view of the fact that $|\Psi_0|$ can be chosen arbitrarily large. It is due to the fact that $N_0(\phi)$ defined \eqref{norms.def} is made arbitrarily small when $|\Psi_0|$ is large and thus \eqref{construction.main.assumption} is satisfied, regardless of the smallness of the initial charge $|Q|(u_0,v_0)$ (which might seem unexpected).

Construction~\ref{V.data.intro} is obtained perturbing Kasner asymptotics of the form \eqref{kasner.data.intro} with a small $\tphi \approx 0$ and also allows for a large scalar-field initial data.  However, it is more restrictive in the sense that the Klein--Gordon mass $m^2$ must be small enough; interestingly, the charge coupling constant $q_0$ can be chosen arbitrarily. 


\end{rmk}


\begin{rmk}
The assumption \eqref{hyp5} can be replaced by another (electromagnetic) gauge-invariant condition: there exists $\alpha_{\infty}(u_0) \in \RR$ such that  \begin{equation}\label{hyp5'}
	|\Im( e^{i q_0 \int_{v_0}^{v}  A_v(u_0,v') dv'} e^{-i\alpha_\infty(u_0)} D_v \phi(u_0,v))| \leq D_{C} \cdot v^{-s-\delta}.
\end{equation} In fact, we prove in Section~\ref{section.faster} that \eqref{hyp5} implies \eqref{hyp5'} and the rest of the proof just makes use of \eqref{hyp5'}. This slight variation of Theorem~\ref{main.thm} will be important in our companion paper \cite{bif2}.
\end{rmk}

\begin{rmk}
It is easy to see that the assumption $ L_-\ v^{-2s} \leq -r\rd_v r(u_0,v) \leq L_+\ v^{-2s}$  in \eqref{hyp1} is superseded by \eqref{hyp4}-\eqref{hyp3}, integrating the Raychaudhuri equation \eqref{RaychV} in $v$ on the initial outgoing cone $\{u=u_0\} \times [v_0,+\infty)$.
\end{rmk}

\begin{rmk}
It is also possible to rephrase the Kasner asymptotics of Statement~\ref{main.thm.II} of Theorem~\ref{main.thm} using the Weingarten formalism, see Remark~\ref{rmk.Weingarten} in Section~\ref{kasner.section}.
\end{rmk}

\

\section{A priori estimates and qualitative aspects of Theorem~\ref{main.thm}}\label{apriori.section}


\subsection{A priori estimates}\label{apriori.est.section}

The next lemma provides  very general  a priori estimates and does not require any assumption.  We denote $\mathcal{F}(u,v):= \int_{u_0}^u  \frac{r|D_u \phi|^2(u',v)}{|\rd_u r|(u',v)} du'$, a gauge-invariant quantity.

\begin{lemma}\label{lemma.hardy} We denote $F_0(v) = F(u_0,v)$ for any function $F(u,v)$. Then, there exists a numerical constant $C>0$ such that
\begin{equation}\label{Hardy.phi}
	|\phi|(u,v) \leq |\phi_0|(v)  + \mathcal{F}^{1/2}(u,v)\log^{1/2}( \frac{r_0(v)}{r(u,v)}),
\end{equation}\begin{equation}\label{Hardy.Q}
	\bigl| Q(u,v)- Q_0(v)\bigr|  \leq |q_0| r^2_0(v) \mathcal{F}^{1/2}(u,v)\left(   |\phi_0|(v)  + \mathcal{F}^{1/2}(u,v)\right),
\end{equation}
\begin{equation}\label{trapped.est}
	\bigl| -r\partial_v r(u,v)   	+(r\partial_v r)_0(v) \bigr| \leq \frac{C r_0(v)  }{r(u,v)}\ \frac{ \Omega^2_0(v)  }{  (-r\rd_u r)_0(v) }  \left( r^2_0(v)+ Q^2_0(v) + q_0^2 r^4_0(v) (1+ |\phi_0|^4(v))+ m^2 r_0^4(v) \big[ 1+  |\phi_0|^2(v)\big] \right),\end{equation}
\begin{equation}\label{Omega.est0}
	\Omega^2(u,v)  \leq   \Omega^2_0(v)  \cdot  \frac{ (-\partial_u r)(u,v)}{(-\partial_u r)_0(v)} ,
\end{equation}

\begin{equation}\label{Omega.est}
	\int_{u_0}^{u} \Omega^2(u',v) du' \leq \frac{r^2_0\Omega^2_0(v)  }{(-r\rd_u r)_0(v)}.
\end{equation}

\end{lemma}

\begin{proof} First, note that \eqref{Hardy.phi} follows by an elementary use of the Cauchy-Schwarz inequality (see \cite{breakdown}[Lemma 4.3] for details). Note that  we have by \eqref{RaychU} \begin{equation}\label{omega.est.proof}
	\Omega^2(u,v) = \frac{|\rd_u r|(u,v)}{|\rd_u r|_0(v)} \Omega^2_0(v) \exp(-\mathcal{F}(u,v)).
\end{equation}  Since $\mathcal{F}\geq 0$, \eqref{Omega.est0} and \eqref{Omega.est} follow immediately from \eqref{omega.est.proof}. Then, by \eqref{chargeUEinstein} we get \eqref{Hardy.Q}, more precisely \begin{equation*}
	\begin{split}
		&| Q(u,v)- Q_0(v)| \leq |q_0| \mathcal{F}^{1/2}(u,v)  (\int_{u_0}^{u} |\rd_u r|(u',v)  r^3|\phi|^2(u',v) du')^{1/2}\leq |q_0| r^2_0(v)\left(  \mathcal{F}^{1/2}(u,v) |\phi_0|(v)  + \mathcal{F}(u,v)\right)\\ & \leq 2 |q_0|\ r^2_0(v) \left( |\phi_0|^2(v) + \mathcal{F}(u,v) \right),
	\end{split} 
\end{equation*}   where we have used the computation $\int_0^{1} x^3 \log(x^{-1}) dx= \frac{1}{16}$ to obtain the one-before-last inequality and we have also used the estimate \eqref{Hardy.phi}. \eqref{trapped.est} follows similarly from \eqref{Radius3}, using crucially \eqref{omega.est.proof} and inequalities of the form $\exp(-\mathcal{F}) \mathcal{F}^p \lesssim_p 1$, we omit the details.

\end{proof}

We now place ourselves in the setting of Theorem~\ref{main.thm}, Statement~\ref{main.thm.I} and consider the Cauchy development of data on $C_{out} = \{u_0\} \times [v_0,+\infty)$ and $\Cin =[u_0,u_F) \times  \{v_0\}$.

\begin{prop}\label{apriori.prop1} We  assume that the estimates \eqref{hyp1} on $C_{out}$ are true. 	 Then $u_0 \leq\uch \leq u_F$,  and there exists $v_0(D)>0$ large enough so that $[u_0, \uch]\times [v_0,+\infty) \subset \T$, and for all $ (u,v) \in [u_0, \uch]\times [v_0,+\infty)$: \begin{align}
	&0.9 L_-\ v^{-2s}\leq -r\rd_v r(u,v) \leq  1.1 L_+\ v^{-2s}	\label{lambda.est.prelim},\\  &  \label{r.est.prelim}  r^2_{CH}(u) + \frac{0.9 L_-}{2s-1}\ v^{1-2s}\leq r^2(u,v) \leq r^2_{CH}(u) + \frac{1.1 L_+}{2s-1}\ v^{1-2s},\\ &  \bigl| r|\rd_u r|(u,v) -r|\rd_u r|_{CH}(u) \bigr| \leq E(D)\ e^{1.9 K_- v} \label{nu.est.prelim}.
\end{align}

\end{prop}

\begin{proof}
We make the following bootstrap assumption \begin{equation}\label{B.r}
	r^2(u,v) \geq \ v^{-p},
\end{equation}
for some large $p>0$ to be determined later. Integrating $(-r\rd_v r)(u_0,v)$ and using \eqref{hyp1}, it is clear that \eqref{B.r} is satisfied for $u=u_0$ if $p>2s-1$ and $v_0(D)$ large enough. Thus, by \eqref{trapped.est}, there exists $E(D)>0$ such that \begin{equation*}
	| -r\rd_v r(u,v) + r\rd_v r(u_0,v) | \leq E(D)\ e^{1.9 K_- v} v^p.	\end{equation*} Thus, in view of the lower bound on $| r\rd_v r(u_0,v) |$ from \eqref{hyp1}, it is clear that for $v_0(D)>0$ large enough, $(u,v) \in \T$  and moreover \eqref{lambda.est.prelim} hold. Note that, by definition, for any $u_0 \leq u \leq \uch$, $\{u\}\times[v_0,+\infty) \subset \mathcal{Q}$ therefore one can integrate \eqref{lambda.est.prelim} and get  \begin{equation*}
	r^2(u,v) \geq r^2_{CH}(u) + \frac{0.9 L_-}{2s-1} v^{-2s} \geq \frac{0.9 L_-}{2s-1} v^{-2s}.
\end{equation*}

In particular, \eqref{B.r} is retrieved for any $u \leq \uch$, and \eqref{r.est.prelim} follow. \eqref{nu.est.prelim} also immediately follows integrating \eqref{Radius} in the form $\rd_v(-r\rd_u r)=...$ in the $v$-direction and using \eqref{Omega.est}, \eqref{r.est.prelim}.

\end{proof}
We note that Proposition~\ref{apriori.prop1} provides a proof of the first part of Statement~\ref{main.thm.I} of Theorem~\ref{main.thm}, namely of the quantitative estimates \eqref{quant1}. The full proof of Statement~\ref{main.thm.I} of Theorem~\ref{main.thm} (namely \eqref{quant2}) will be completed in Section~\ref{section.CH}. The global breakdown of the Cauchy horizon in gravitational collapse first proven in \cite{breakdown} and re-proven in a simpler fashion in \cite{bif2} also follows immediately from Proposition~\ref{apriori.prop1}.

From Proposition~\ref{apriori.prop1}, it is clear that the gauge \eqref{gauge2} can be imposed, which we will do in the rest of the proof. We will also fix the gauge freedom so that $\uch=0$ (see Section~\ref{geometricframework}). Therefore, we have for all $u_0 \leq u \leq 0$: \begin{equation*}\begin{split}
	&(r\rd_u r)_{CH}(u) \equiv -1,\\ &  r_{CH}^2(u) =2|u|. 
\end{split}
\end{equation*}

\subsection{A priori characterization of the spacetime boundary}

In what follows, we prove a result of independent interest, which is that the Cauchy horizon $\CH$ (which obeys estimates given by \eqref{hyp1} that are essentially equivalent to mass inflation) cannot be followed by an ingoing light cone $\mathcal{S}_{i^+}$ on which $r$ extends to $0$. We also show that, if the Cauchy horizon breaks down, then $r$ tends to $0$ towards its endpoint.  

\begin{prop}\label{apriori.prop2}
We  assume that the estimates \eqref{hyp1} on $C_{out}$ are true, and that $\uch<u_F$. Then \begin{equation}\begin{split}
		\lim_{v\rightarrow+\infty} r(\uch,v) = \lim_{u \rightarrow \uch,\ u< \uch} r_{CH}(u)=0.		 	\end{split}
\end{equation} Moreover, $\mathcal{S}_{i^+}=\{ \uch< u \leq u_F, \underset{v\rightarrow+\infty}{\lim} r(u,v) = 0 \}=\emptyset$, and for all $\uch<u \leq u_F$, there exists $v_{\mathcal{S}}(u)<+\infty$ such that \begin{equation*}
	\lim_{ v \rightarrow v_{\mathcal{S}}(u)} r(u,v) =0.
\end{equation*}
\end{prop}

\begin{proof} Note that by Proposition~\ref{apriori.prop1}, $(\uch,v_0) \in \T$. Since $\uch< u_F$, and $\T$ is open, there exists $\epsilon>0$ such that $[\uch,\uch+\ep] \times \{v_0\} \subset \T$ and by the monotonicity of \eqref{RaychV} we deduce that  $[\uch,\uch+\ep] \times [v_0,+\infty)\subset \T$.

Since $\uch< u_F$ by definition, then it means that for every $\uch < u< u_F$, there exists $v_{\mathcal{B}}(u) \in \RR \cup \{+\infty\}$ such that  $$ \lim_{v\rightarrow v_{\mathcal{B}}(u)} r(u,v)=0.$$ 

We then integrate $-\rd_u \rd_v (r^2)$ on $\{ u\in [\uch-\ep,\uch+\ep], V \leq v\leq v_{\mathcal{B}}(u) \} $ (adopting the convention that  $v_{\mathcal{B}}(u)=+\infty$ for $u\leq \uch$) using \eqref{Radius3} and obtain, exploiting the monotonicity of \eqref{RaychV} and the fact that $ \{ u\in [\uch-\ep,\uch+\ep], V \leq v\leq v_{\mathcal{B}}(u) \} \subset \T $ (namely we use  $\int_{v_1}^{v_2} \Omega^2(u,v) dv \leq \frac{\Omega^2}{|\rd_v r|}(u,v_1) r(u,v_1)$):    \begin{equation*}
	r^2(\uch+\ep,V)- r^2(\uch-\ep,V) + r^2_{CH}(\uch-\ep) \leq \int_{\uch-\ep}^{\uch+\ep} \int_{V}^{v_{\mathcal{B}}(u)} \Omega^2(u,v) dv du  \lesssim \ep.\end{equation*}
Taking $\ep \rightarrow 0$ gives (using the continuity of $u \rightarrow r^2(u,V)$)  $$ \lim_{u \rightarrow \uch,\ u< \uch} r_{CH}(u)=0. $$  Then, using \eqref{r.est.prelim} we obtain $$\lim_{v\rightarrow+\infty} r(\uch,v)=0.$$

Let us prove that $\mathcal{S}_{i^+}=\{\uch< u \leq u_F, \underset{v\rightarrow+\infty}{\lim} r(u,v) = 0 \}=(\uch,\usi] = \emptyset$ by contradiction. Assuming that $\mathcal{S}_{i^+} \neq \emptyset$, we revisit the proof of Proposition~\ref{apriori.prop1} to show \eqref{lambda.est.prelim}, \eqref{r.est.prelim}, \eqref{nu.est.prelim} are still  valid on the larger rectangle $ (u,v) \in [u_0, \usi]\times [v_0,+\infty) \subset \T$, and we have $r_{CH}(u)=0$ for all $\uch \leq u \leq \usi$. But we note that integrating \eqref{nu.est.prelim} under the gauge \eqref{gauge2} gives $$ \frac{1}{2} (u-\uch) \leq \frac{r^2(\uch,v)- r^2(u,v)}{2}+ E(D) e^{1.99 K_- v} \lesssim v^{1-2s},$$ which, as $v\rightarrow +\infty$, gives $u-\uch=0$, which is obviously a contradiction.
\end{proof}

\subsection{Sufficient condition for  a local breakdown of the Cauchy horizon} The following proposition is important in the initial data construction relevant to Statement~\ref{main.thm.III} of Theorem~\ref{main.thm}, see Section~\ref{section.data}.
\begin{prop}\label{localbreak.prop} Assume that the estimates \eqref{hyp1} on $C_{out}$ are satisfied. Moreover, assume that $[u_0,u_F) \times\{v_0\} \subset \T$ and $\underset{ u \rightarrow u_F}{\lim} r(u,v_0)=0$. Then $\uch < u_F$ and for all $\uch < u \leq u_F$, there exists $v_{\mathcal{S}}(u)<+\infty$ such that \begin{equation*}
	\lim_{ v \rightarrow v_{\mathcal{S}}(u)} r(u,v) =0.
\end{equation*}
\end{prop}
\begin{proof} 	$\uch < u_F$ follows immediately from the monotonicity of \eqref{RaychV}: indeed for all $u_0 \leq u < u_F$: $$ r(u,v) < r(u,v_0),$$ which is clearly violating \eqref{r.est.prelim} as $u \rightarrow u_F$ if it was the case that $\uch = u_F$.	\end{proof}

\subsection{A calculus lemma}
The following calculus lemma will be crucial in the estimates under the Cauchy horizon proven in Section~\ref{section.CH}.

\begin{lemma}\label{v^s/r^2.lem}Let $1<s\leq p < 2s$. Suppose that for all $u_0 \leq u \leq 0$, $v\geq v_0$: $$  2|u| +\frac{2D_-}{2s-1}\ v^{1-2s}\leq r^2(u,v)   \leq 2|u| +\frac{2D_+}{2s-1}\ v^{1-2s}.$$   
Then \begin{equation}\label{int1}
	\int_{v_0}^{v} \frac{(v')^{-p}}{r^2(u,v')} dv' \lesssim r^{\frac{p-2s}{s-\frac{1}{2}}}(u,v).
\end{equation} If moreover, $s\leq p < s+\frac{1}{2}$, then
\begin{equation}\label{int2}
	\int_{v_0}^{v} \frac{(v')^{-p}}{r(u,v')} dv' \lesssim r^{-\frac{1-2p+2s}{2s-1}}(u,v).
\end{equation}
\end{lemma}

\begin{proof}
First, by bounding crudely  $r^2(u,v) \gtrsim  v^{1-2s}$ we obtain \begin{align} \label{int1.t}
	& \int_{v_0}^{v} \frac{(v')^{-p}}{r^2(u,v')} dv \lesssim v^{2s-p}, \\ &  \int_{v_0}^{v} \frac{(v')^{-p}}{r(u,v')} dv \lesssim v^{\frac{1}{2} -p+s}. \label{int2.t}
\end{align}

Now note that in the region $|u| \leq v^{1-2s}$, we have $r^2(u,v) \sim v^{1-2s}$, so \eqref{int1.t}, \eqref{int2.t} give \eqref{int1} and \eqref{int2} already. In the rest of the proof we focus on the   region $|u| \geq  v^{1-2s}$. In  this region we write $r^2(u,v) \sim |u|  $. We will split the integrals as such, for $\alpha \in \{1,2\}$: denote $v_{\gamma}(u)= |u|^{-\frac{1}{2s-1}}$ \begin{equation*}
	\int_{v_0}^{v} \frac{(v')^{-p}}{r^{\alpha}(u,v')} dv' =  \underbrace{\int_{v_0}^{v_{\gamma}(u)} \frac{(v')^{-p}}{r^{\alpha}(u,v')} dv'}_{\lesssim |u|^{-\frac{\alpha}{2}+\frac{p-1}{2s-1}}}+\int_{v_{\gamma}(u)}^{v} \frac{(v')^{-p}}{r^{\alpha}(u,v')} dv',
\end{equation*} where we have obtained the inequality in the first integral thanks to \eqref{int1.t}, \eqref{int2.t} applied to $v=v_{\gamma}(u)$. Note that in our region $|u| \geq v^{1-2s}$: \begin{align} & |u|^{-\frac{\alpha}{2}+\frac{p-1}{2s-1}} \approx r^{\frac{p-2s}{2s-1}}(u,v) \text{ if } \alpha =2,\label{u.eq1} \\ &  |u|^{-\frac{\alpha}{2}+\frac{p-1}{2s-1}} \approx   r^{-\frac{1-2p+2s}{2s-1}}(u,v)  \text{ if } \alpha =1/\label{u.eq2}
\end{align}
Note lastly that the second integral can be controlled as such, taking advantage of the fact that for $v_{\gamma}(u) \leq v' \leq v$, we have $r^{-\alpha}(u,v') \approx |u|^{-\frac{\alpha}{2}}$: \begin{equation}
	\int_{v_{\gamma}(u)}^{v} \frac{(v')^{-p}}{r^{\alpha}(u,v')} dv' \ls  |u|^{-\frac{\alpha}{2}} [v_{\gamma}(u)]^{1-p}
	\approx  |u|^{-\frac{\alpha}{2}+\frac{p-1}{2s-1}},
\end{equation} and then we  use \eqref{u.eq1}, \eqref{u.eq2} to conclude.
\end{proof}

\subsection{Linear propagation estimates on dynamical metrics}\label{propagation.section}
In this section, we provide linear estimates on the wave equation \eqref{5.1} for a dynamical metric $g$ obeying the preliminary estimates  we derived in Section~\ref{apriori.est.section} and up to a region where $r\lesssim \epsilon v^{1/2-s}$. We start with our main propagation lemma that relies on a Gr\"{o}nwall-type argument. For fixed $v\geq v_0$, $\ep \in (0,1]$, let us define $u_{\ep}(v) > 0$ such that $r(u_{\ep}(v),v) = \ep = \ep v^{1/2-s}$.

\begin{lemma}\label{lemma.propagation}  Assume that  for some $1<s \leq p <2s$,  \begin{equation}\label{ttheta.data}
	|\ttheta|(u_0,v) \lesssim v^{-p}.
\end{equation}
Let the functions $(\ttheta(u,v),\txi(u,v))$ satisfying the following system of equations
\begin{equation}\begin{split} \label{eq.lemma}
		& \rd_u  \ttheta= -\frac{\lambda}{r} \txi + F, \\ &  \rd_v  \txi= -\frac{\nu}{r} \ttheta+ G,
	\end{split}
\end{equation} 
with $(F(u,v),G(u,v))$ satisfying the following estimates, for some $0<\eta< 2s-p$: for all $v\geq v_0$, $u_0\leq u \leq u_{\ep}(v)$:
\begin{equation}\label{FG.eq}\begin{split}&	|F|(u,v) \lesssim r^{-2 } v^{-p- \eta},\\ & |G|(u,v) \lesssim v^{2s-1-p-\eta},\end{split}\end{equation}  and moreover $\lambda(u,v) = \rd_v r(u,v)$,  $\nu(u,v) = \rd_u r(u,v)$ satisfy the following estimates \begin{equation}\label{lambdanu.eq}\begin{split}
		&v^{-2s}\lesssim-r\lambda(u,v) \lesssim v^{-2s},\\ & 1\lesssim-r\nu(u,v) \lesssim 1.
	\end{split}
\end{equation}
Then,  for all $v\geq v_0$, $u_0\leq u \leq u_{\ep}(v)$, we have  \begin{equation}\label{ttheta.est}
	|\ttheta|(u,v) \lesssim_{\ep} v^{-p},
\end{equation}
\begin{equation}\label{txi.est}
	|\txi|(u,v) \lesssim_{\ep} r^{-\frac{2s-p}{s-\frac{1}{2}}}.
\end{equation}
\end{lemma}
\begin{proof}
We integrate \eqref{eq.lemma} in  $v$, taking advantage of \eqref{FG.eq} to write 
\begin{equation}\begin{split}
		&	\rd_u	 \ttheta(u,v) = \frac{|\lambda|}{r}(u,v) \txi(u,v_0)+ \frac{|\lambda|}{r}(u,v)  \int_{v_0}^{v} \frac{|\nu|}{r}(u,v')\ttheta(u,v') dv' +  F(u,v) + \frac{|\lambda|}{r}(u,v)  \int_{v_0}^{v}  G(u,v') dv'\\ &  =\frac{|\lambda|}{r}(u,v)  \int_{v_0}^{v} \frac{|\nu|}{r}(u,v')\ttheta(u,v') dv'+ O(r^{-2} v^{-p-\eta}).\end{split}
\end{equation} Then, we can integrate in $u$ and use \eqref{ttheta.data} to obtain \begin{equation}\label{Gronwall.arg1}
	v^{p}|\ttheta|(u,v) \lesssim 1 +  v^{p-2s}\int_{u_0}^{u} r^{-2}(u',v)  \int_{v_0}^{v}  r^{-2}(u',v')   |\ttheta|(u',v') dv' du'+ O( \log(v) v^{-\eta})
\end{equation} Now fix $v\geq v_0$, $u\geq u_0$, and define $ \Pi(u,v)\geq 0$ and $(u^{*}(u,v),v^{*}(u,v))\in [u_0,u]\times[v_0,v]$ as such $$ \Pi(u,v)= \sup_{ v_0 \leq v'\leq v, u_0 \leq  u' \leq u} [v']^{p}|\ttheta|(u',v') = [v^{*}(u,v)]^{p}|\ttheta|(u^{*}(u,v),v^{*}(u,v)) . $$

Now, evaluate \eqref{Gronwall.arg1} at $(u^{*},v^{*})$ and note that, using the definition of $\Pi$, \eqref{lambdanu.eq} and Lemma~\ref{v^s/r^2.lem},
\begin{equation}\label{Gronwall.arg2}\begin{split}
		&\Pi(u,v) \lesssim 1 +  [v^{*}]^{p-2s}\int_{u_0}^{u^{*}} r^{-2}(u',v^{*}) \Pi(u',v)  \int_{v_0}^{v^{*}}  r^{-2}(u',v') [v']^{-p} dv' du'\lesssim  1 +  [v^{*}]^{p-2s}\int_{u_0}^{u^{*}} r^{-2-\frac{2s-p}{s-\frac{1}{2}}}(u',v^{*}) \Pi(u',v)  du'. \end{split}
\end{equation} where we have used the fact that $(u',v') \in [u_0,u'] \times [v_0,v]$, hence $\Pi(u',v') \leq \Pi(u',v)$. Note also that \begin{equation}
	[v^{*}]^{p-2s}\int_{u_0}^{u^{*}} r^{-2-\frac{2s-p}{s-\frac{1}{2}}}(u',v^{*}) du' \lesssim [v^{*}]^{p-2s} r^{-\frac{2s-p}{s-\frac{1}{2}}}(u^{*},v^{*}) \lesssim  \ep^{-\frac{2s-p}{s-\frac{1}{2}}}(u^{*},v^{*}) .
\end{equation} Therefore, by the Gr\"{o}nwall lemma and \eqref{Gronwall.arg2}, we get \begin{equation}
	v^{p} |\ttheta|(u,v) \leq \Pi(u,v) \lesssim_{\ep} 1,
\end{equation} thus \eqref{ttheta.est} holds and \eqref{txi.est} follows integrating \eqref{eq.lemma} in $v$ and using Lemma~\ref{v^s/r^2.lem}, which concludes the proof.
\end{proof}

Then, we investigate the propagation of improved decay for the first derivative of the initial data.
\begin{lemma}\label{lemma.propagation.derivative}
Assume that $\ttheta(u_0,v)$ satisfies \eqref{ttheta.data} for some $1<s \leq p< 2s-1$ and, moreover, that  \begin{equation}\label{dvttheta.data}
	\bigl| \rd_v (\ttheta(u_0,v))\bigr|\lesssim v^{-p-1}.
\end{equation}  Assume the functions $(\ttheta(u,v),\txi(u,v))$ satisfy \eqref{eq.lemma}, $(\lambda,\nu)$ satisfy \eqref{lambdanu.eq}, $(F,G)$ satisfy \eqref{FG.eq} for some $1<\eta<2s-p$, and moreover for all $v\geq v_0$, $u_0\leq u \leq u_{\ep}(v)$ \begin{equation}\label{lambda.lemma.eq2}
	|\rd_v(r\lambda)|(u,v) \lesssim v^{-2s-1},
\end{equation} \begin{equation}\label{FG.eq2}
	|\rd_v F|(u,v) \lesssim r^{-2}(u,v) v^{-p -1-\eta}.
\end{equation} Then, for all $v\geq v_0$, $u_0\leq u \leq u_{\ep}(v)$ \begin{equation}\label{dttheta.est}
	\bigl| \rd_v (\ttheta(u_0,v))\bigr|(u,v) \lesssim_{\ep} v^{-p-1}.
\end{equation}
\end{lemma}
\begin{proof} Note that the assumptions of  Lemma~\ref{lemma.propagation} are satisfied, therefore \eqref{ttheta.est}, \eqref{txi.est} hold.
We differentiate \eqref{eq.lemma} in $v$ and obtain \begin{equation}
	\rd_u \rd_v \ttheta= -\rd_v(\frac{\lambda}{r}) \txi -\frac{\lambda}{r} \left(-\frac{\nu}{r} \ttheta+ G\right) + \rd_v F.
\end{equation}

Now, by \eqref{FG.eq}, \eqref{FG.eq2}, \eqref{ttheta.est}, \eqref{txi.est}, \eqref{lambdanu.eq} and \eqref{lambda.lemma.eq2} we have $$ |\rd_u \rd_v \ttheta| \lesssim  r^{-\frac{2s-p}{s-\frac{1}{2}}}(u,v)[ r^{-2}(u,v) v^{-2s-1} + r^{-4}(u,v)v^{-4s}]+ r^{-4}(u,v)v^{-p-2s}+ r^{-2}(u,v) v^{-p-1-\eta },$$ and then, integrating in $u$ using \eqref{lambdanu.eq} gives \eqref{dttheta.est}.

\end{proof}

Finally, we turn to the propagation of point-wise lower bounds using a monotonicity-type argument.
\begin{lemma}\label{lowerbound.nonosc.lemma} 
Let $\ttheta(u_0,v)$ satisfying \eqref{ttheta.data}. Assume moreover, that the following lower bound holds, for some $D_L>0$
\begin{equation}\label{real.lower.data}
	\Re \ttheta(u_0,v)  \geq D_L \cdot v^{-p}.
\end{equation} Then, assuming that  the functions $(\ttheta(u,v),\txi(u,v))$ satisfy \eqref{eq.lemma} with $F(u,v)$, $G(u,v)$ satisfying \eqref{FG.eq}, and \eqref{lambdanu.eq} is satisfied, then   for all $v\geq v_0$, $u_0\leq u \leq u_{\ep}(v)$\begin{equation}\label{real.lower}
	\Re \ttheta(u,v)  \geq \frac{D_L}{2}\cdot v^{-p}.
\end{equation}
\end{lemma}
\begin{proof} Taking the real part of \eqref{eq.lemma}, we obtain the following integral representation of 	$\Re \ttheta$ \begin{equation}\label{real.eq1}\begin{split}
		&	\rd_u	\Re \ttheta(u,v) = \frac{|\lambda|}{r}(u,v) \Re\txi(u,v_0)+ \frac{|\lambda|}{r}(u,v)  \int_{v_0}^{v} \frac{|\nu|}{r}(u,v')\Re\ttheta(u,v') dv' + \Re F(u,v) + \frac{|\lambda|}{r}(u,v)  \int_{v_0}^{v} \Re G(u,v') dv'\\ &  = \frac{|\lambda|}{r}(u,v)  \int_{v_0}^{v} \frac{|\nu|}{r}(u,v')\Re\ttheta(u,v') dv'+ O(r^{-2} v^{-p-\eta}),\end{split}
\end{equation} where in the second equality we used \eqref{FG.eq}. The rest of the proof is generalizing the monotonicity argument first introduced in \cite{MihalisPHD}, albeit now with additional (faster-decaying) error terms.	We introduce the bootstrap assumption  \begin{equation}\label{bootstrap.lower.theta}
	\Re \ttheta(u,v) \geq \frac{D_L}{10}\cdot v^{-p},
\end{equation}  which is satisfied if $u=u_0$ by \eqref{real.lower.data}. Now, integrating \eqref{real.eq1} in $u$, combining with \eqref{real.lower.data},  gives \begin{equation}\label{real.eq2}
	\Re \ttheta(u,v) \geq  \Re \ttheta(u_0,v) +O( \log(v) v^{-p-\eta})  \geq D_L \cdot v^{-p} +O( \log(v) v^{-p-\eta}) \geq  \frac{D_L}{2}\cdot v^{-p},
\end{equation} which improves the bootstrap assumption \eqref{bootstrap.lower.theta} and proves \eqref{real.lower}, thus concluding the proof of the lemma.
\end{proof}

\section{Quantitative estimates in Theorem~\ref{main.thm} under the  Cauchy horizon}\label{section.CH}
Under the Cauchy horizon, we recall that \eqref{gauge2} translates into the following $u$-gauge for all $u_0 \leq u \leq 0$:
\begin{equation}\label{gauge.UCH}
\lim_{v \rightarrow +\infty}	-r \partial_u r(u,v)= 1,
\end{equation} 
\begin{equation}\label{gauge.AUCH}\lim_{v \rightarrow +\infty}	A_u(u,v)=0,
\end{equation}

\subsection{The main upper bound estimates}

The goal of the following Proposition~\ref{main.prop.CH} is to prove \eqref{quant2} and thus complete the proof of Statement~\ref{main.thm.I} of Theorem~\ref{main.thm}. Recall that $\uch=0$ by our gauge choice (see Section~\ref{geometricframework}).

\begin{prop}\label{main.prop.CH}
Assuming \eqref{hyp1}, \eqref{hyp2}, the following estimates hold for all $u_0 \leq u\leq 0$: \begin{equation}\label{r.est}
	2|u| +\frac{2D_-}{2s-1}\ v^{1-2s}\leq r^2(u,v)   \leq 2|u| +\frac{2D_+}{2s-1}\ v^{1-2s},
\end{equation}
\begin{equation} \label{nu.est}
	\bigl| r|\nu|(u,v)-1 \bigr|  \lesssim  e^{1.99 K_- v},
\end{equation}
\begin{equation}\label{lambda.est2}
	D_-\cdot v^{-2s}\leq (-r\lambda)(u,v) \leq D_+\cdot v^{-2s},
\end{equation} \begin{equation} \label{Omega.est2}
	\Omega^2(u,v) \lesssim  v^{s-\frac{1}{2}} \cdot \Omega^2(u_0,v) \lesssim e^{1.99 K_- v},
\end{equation}
\begin{equation} \label{Au.est}
	|A_u|(u,v) \lesssim  v^{2s-1} \cdot \Omega^2(u_0,v) \lesssim e^{1.99 K_- v},
\end{equation}
\begin{equation}\label{phi.est}
	|\phi|(u,v) \lesssim r^{-\frac{1}{2s-1}},
\end{equation}
\begin{equation}\label{duphi.est}
	r |D_u\phi|(u,v) \lesssim r^{-\frac{s}{s-\frac{1}{2}}},
\end{equation}
\begin{equation}\label{Dvphi.est}
	r| D_v \phi|(u,v)\lesssim v^{-s},
\end{equation}
\begin{equation}\label{Q.est}
	|Q|(u,v) \lesssim 1,
\end{equation}
\begin{equation}\label{dv_Omega.est}
	|\rd_v \log(\Omega^2)- 2K_-| \lesssim   v^{s} r^{-\frac{s}{s-1/2}}(u,v).
\end{equation}
\end{prop}
\begin{proof}
\eqref{r.est}, \eqref{nu.est}, \eqref{lambda.est2} hold as an application  of Proposition~\ref{apriori.prop1}.
By \eqref{Omega.est0} and \eqref{r.est} this implies $$
\Omega^2(u,v) \lesssim \frac{\Omega^2(u_0,v)}{r(u,v)} \lesssim v^{s-\frac{1}{2}} \cdot \Omega^2(u_0,v) \lesssim e^{1.99 K_- v},$$ hence \eqref{Omega.est2} holds.	We then make the following  bootstrap assumptions:	
\begin{equation}\label{B3}
	|Q|(u,v) \leq \log^2(v),
\end{equation}	
\begin{equation}\label{B4}
	|\phi|^2(u,v) \leq  v \log^2(v).
\end{equation}
\eqref{Au.est} follows from \eqref{B3} and \eqref{Omega.est2}, \eqref{r.est} integrating  \eqref{Maxwell} in $v$, using \eqref{gauge.AUCH}.  Now we improve bootstrap \eqref{B4} by appealing to Lemma~\ref{lemma.propagation}. Note that defining $\ttheta=\theta$ and $\txi=\xi$, we have \eqref{eq.lemma} is satisfied by \eqref{Field4} defining $F= G=\frac{ \Omega^{2} \cdot \phi}{4r} \cdot ( i q_{0} Q-m^2 r^2)- i q_0 A_u  r\rd_v \phi -i q_0 \partial_{v}r \cdot A_u \phi$. By the above estimates, \eqref{FG.eq}, \eqref{lambdanu.eq} are satisfied and moreover \eqref{ttheta.data} is satisfied for $p=s$. As a consequence of Lemma~\ref{lemma.propagation}, we obtain \eqref{duphi.est}, \eqref{Dvphi.est}.  Integrating in $u$ gives \eqref{phi.est} and improves \eqref{B4}.   Now from \eqref{duphi.est}, and \eqref{B4} we get \begin{equation}
	|\partial_u Q | \lesssim  r^{-\frac{1}{s-1/2}} \lesssim r^{1-\frac{1}{s-1/2}}  |\rd_u r|,
\end{equation} which is integrable in $u$, since $s>1$. Thus $$ |Q(u,v) - Q(u,v_0) | \lesssim C. $$ Thus, \eqref{B3} is improved   choosing $v_0>1$ sufficiently large. Finally \eqref{dv_Omega.est} follows from the integration of \eqref{Omega}.
\end{proof}

\subsection{Higher order decay estimates}

In this section, we provide quantitative estimates on commuted quantities (both for the scalar field and the metric) which will  be crucial in establishing Kasner asymptotics in the next section, in order to prove Statement~\ref{main.thm.II} of Theorem~\ref{main.thm}.

\begin{prop}\label{prop.highorder.decay}
Assume that \eqref{hyp1}, \eqref{hyp2}, \eqref{hyp4}-\eqref{hyp3} are satisfied.		Then for all $ u_0 \leq u \leq 0$, $v\geq v_0$:	
\begin{equation}\label{Dvdvphi.est}
	r |\rd_{v v}^2 \phi|(u,v),\  |\rd_v ( r\rd_v\phi)|(u,v) \lesssim v^{-s-1},
\end{equation}
\begin{equation}\label{dvv.r.est}
	|	\rd_v(r \rd_v r)|(u,v) \lesssim v^{-2s-1}.
\end{equation}
\end{prop}

\begin{proof} First, \eqref{dvv.r.est} easily follows from \eqref{hyp3}, applying $\rd_v $ to \eqref{Radius3} and integrating in $u$, making use of \eqref{Omega.est2}. 
We have already seen in the proof of Proposition~\ref{main.prop.CH} that    $(\ttheta=\theta,\txi=\xi)$ satisfy \eqref{eq.lemma} with $F= G=\frac{ \Omega^{2} \cdot \phi}{4r} \cdot ( i q_{0} Q-m^2 r^2)- i q_0 A_u  r\rd_v \phi -i q_0 \partial_{v}r \cdot A_u \phi$ satisfying \eqref{FG.eq}. 
Note that \eqref{dvttheta.data} is satisfied by \eqref{hyp4} and by Proposition~\ref{main.prop.CH}, $F$ also satisfies \eqref{FG.eq2}, while \eqref{dvv.r.est} shows \eqref{lambda.lemma.eq2} is satisfied. Therefore, by Lemma~\ref{lemma.propagation.derivative}, \eqref{Dvdvphi.est} is satisfied.

\end{proof}

\section{Quantitative estimates in Theorem~\ref{main.thm}: the spacelike singularity}\label{quant.spacelike.section}
By the a priori estimates of Proposition~\ref{apriori.prop1}, we know that there exists $\ep(v_0)>0$ small enough so that $\mathcal{S}=\{ (u,v),\ r(u,v)=0,\ 0 < u \leq \ep \}$ is well-defined. We recall that \eqref{gauge2} translates into the following $u$-gauge:
\begin{equation}\label{gauge.C}
-r \partial_u r(u,v_{\mathcal{S}}(u))= 1,
\end{equation} 
\begin{equation}\label{gauge.A}	A_u(u,v_{\mathcal{S}}(u))=0,
\end{equation}  where we parametrize  $\mathcal{S}=\{ (u,v_{\mathcal{S}}(u)),\  0 < u \leq \ep\}$. We recall that we chose our gauge so that $\uch=0$ and will denote $C_0 =\{0\} \times [v_0,+\infty)$ the outgoing cone terminating at the Cauchy horizon endpoint. We also introduce the notation $F_0(v)=  F(0,v)$ for any function $F$.
\subsection{Statement of quantitative estimates to the future of $C_0$}	
Under the gauge choices provided by \eqref{gauge.C}, \eqref{gauge.A}, we now state the main result of this section, which consists of quantitative estimates up to the  singularity $\mathcal{S}$. We will also prove that $\mathcal{S}$ is $C^1$ spacelike, and can be also parametrized in $v$ as such  $\mathcal{S}=\{ (u_{\mathcal{S}}(v),v),\   v\geq v_0\}$. 
\begin{prop}\label{main.prop.spacelike} Assume \eqref{hyp1}, \eqref{hyp2}, \eqref{hyp4}-\eqref{hyp3} hold. Then, there exists $D>0$, $C>0$ such that the following quantitative estimates hold for any $0 \leq u \leq U_0$, $v_0 \leq  v \leq v_{\mathcal{S}}(u)$.
\begin{equation}\label{Omega.S}
	e^{2.1 K_- v} [\frac{r(u,v)}{r_0(v)}]^{2Dv}  \lesssim \Omega^2(u,v) \lesssim e^{1.9 K_- v} [\frac{r(u,v)}{r_0(v)}]^{Dv},
\end{equation}
\begin{equation}\label{Q.est.S}
	|Q|(u,v) \leq C,
\end{equation}		\begin{equation}\label{phi.est.S}
	|\phi|(u,v) \lesssim  \sqrt{v}  \cdot[1+\log( \frac{r_0(v)}{r(u,v)})],
\end{equation}		\begin{equation}\label{Dvphi}| D_v \phi|(u,v) \ls v^{\frac{1}{2}-2s} r^{-2}(u,v),
\end{equation}
\begin{equation}\label{Duphi}| D_u \phi|(u,v) \ls  v^{\frac{1}{2}}r^{-2}(u,v),
\end{equation}
\begin{equation}\label{DuDuphi}| D_u \left(rD_u \phi\right)|(u,v),\ | r D_{u u}^2 \phi|(u,v)  \ls  v^{\frac{1}{2}}r^{-3}(u,v),
\end{equation}
\begin{equation}\label{DvDvphi}|D_v(rD_v \phi)|(u,v),\ r| D_{vv}^2 \phi|(u,v) \ls 	v^{\frac{1}{2}-4s} r^{-3}(u,v),
\end{equation}
\begin{equation}\label{duA}
	\bigl| \rd_u A_u(u,v) \bigr| \lesssim \exd r^{95}(u,v),
\end{equation}
\begin{equation}\label{DuOmega}
	|\rd_u \log(\Omega^2)|(u,v) \ls v\cdot r^{-2}(u,v),
\end{equation}
\begin{equation}\label{DvOmega}
	|\rd_v\log(\Omega^2)|(u,v) \ls v^{1-2s}  r^{-2}(u,v),
\end{equation}
\begin{equation}\label{lambda.eq}
	\bigl| r|\lambda|(u,v) - r_0(v) |\lambda_0|(v) \bigr| \lesssim \exd,
\end{equation}\begin{equation}\label{nu.S.eq3}
	\bigl| r|\nu|(u,v) -1 \bigr| \lesssim e^{\frac{K_-}{4} v}r^{90},
\end{equation}\begin{equation}\label{A.S}
	\bigl| A_u(u,v) \bigr| \lesssim   e^{\frac{K_-}{4} v}r^{90},
\end{equation}
\begin{equation}\label{du.nu}
	|\rd_u(r\nu)|(u,v) \ls \exd r^{95}(u,v),
\end{equation} 
\begin{equation}\label{dv.lamba}
	|\rd_v(r\lambda)(u,v)- \rd_v(r_0\lambda_0)(v)| \ls \exd, \text{ where } |\rd_v(r_0\lambda_0)|(v)\ls v^{-2s-1},
\end{equation}\begin{equation}\label{logOmega}
	\bigl|\log(\Omega^2 r)\bigr| \ls v \cdot(1+ \log(\frac{r_0(v)}{r(u,v)})).
\end{equation}
Moreover, there exists a  positive function denoted $r|\lambda|_{\mathcal{S}}(v)$ such that 
\begin{equation}\label{rlambdaS.prop}
	\lim_{u \rightarrow u_{\mathcal{S}}(v)} r|\lambda|(u,v)  =r|\lambda|_{\mathcal{S}}(v),\  r|\lambda|_{\mathcal{S}}(v) \approx v^{-2s},\ \bigl| r|\lambda|(u,v)  -r|\lambda|_{\mathcal{S}}(v) \bigr| \ls r^{D v}.
\end{equation}

Finally, there exists $\ep>0$ sufficiently small such that in the region $\{ (u,v),\ u\geq 0,\  r(u,v) \leq \ep r_0(v),\ v\geq v_0\}$ \begin{equation}\label{dulogOmega}
	C_{-} \cdot \ep^2 \cdot v	 \leq -r^2  \rd_u \log(\Omega^2)(u,v) \leq C_{+} v,
\end{equation}
\begin{equation}\label{dvlogOmega}
	C_{-} \cdot \ep^2 \cdot  v^{1-2s}	 \leq -r^2  \rd_v \log(\Omega^2)(u,v) \leq C_{+}v^{1-2s},
\end{equation} where $C_{\pm}>0$ are $\ep$-independent constants.
\end{prop}

\subsection{The bootstrap assumptions and preliminary estimates}

We will make the following bootstrap assumptions through the region $\{ 0 \leq u,\ v\geq v_0\}$.
\begin{equation}\label{Omega.boot}
\Omega^2(u,v),\ 	|\rd_u\Omega^2 |(u,v),\ 
|\rd_v\Omega^2 |(u,v),\ |A_u|(u,v),\  |\rd_u A_u|(u,v)
\leq     e^{\frac{K_-}{2 } v}[r(u,v)]^{ 100},
\end{equation}\begin{equation}\label{nu.boot}
\frac{1}{10} \leq	r|\nu|\leq 10,
\end{equation}
\begin{equation}\label{Q.boot}
|Q| \leq C,
\end{equation}
\begin{equation}\label{phi.boot}
|\phi| \leq \frac{v^{10}}{r},
\end{equation}
where $C(M,e,q_0,m^2)>0$ is a constant to be fixed later.

\begin{prop}\label{S.curve.prop}
Assuming $v_0$ large enough: for every $v\geq v_0$, there exists $\uS>0$ such that $$ \lim_{ u \rightarrow \uS} r(u,v)=0.$$
Moreover, the curve $\mathcal{S}:=\{(\uS,v),\ v\geq v_0\}$ is $C^1$ and spacelike in the $(u,v)$-plane with \begin{equation}\label{uS.eq}
	\bigl| \uS - \frac{r_0^2(v)}{2} \bigr| \lesssim \exd,
\end{equation} \begin{equation}\label{duS.eq}
	\frac{d\uS}{dv}(v) = - [r|\lambda|]_{\mathcal{S}}(v),
\end{equation} and moreover \eqref{lambda.eq}, \eqref{nu.S.eq3}, \eqref{A.S}, \eqref{rlambdaS.prop} and the following estimate hold trues:   \begin{equation}\label{lambda.S.eq2}
	\bigl| r|\lambda|(u,v) -r|\lambda|_{|\mathcal{S}}(v) \bigr| \lesssim \exd r^{95},
\end{equation} 
\end{prop}

\begin{proof}


\eqref{lambda.eq} follows immediately from integrating \eqref{Radius3} in $u$, using the bootstrap assumptions \eqref{Omega.boot}--\eqref{phi.boot}.  Note also that it shows that the following limit exists \begin{equation}
	\lim_{ u \rightarrow  \uS}  r|\lambda|(u,v):=  r|\lambda|_{|\mathcal{S}}(v),
\end{equation} and \eqref{lambda.S.eq2} holds true. Similarly, one shows that the following limit exists, and the following estimate holds: \begin{equation}\label{nu.S.eq2} \lim_{ v \rightarrow  \vS}  r|\nu|(u,v):=  r|\nu|_{|\mathcal{S}}(u),\ 
	\bigl| r|\nu|(u,v) -r|\nu|_{|\mathcal{S}}(u) \bigr| \lesssim \exd r^{95}.
\end{equation} and \eqref{nu.S.eq3}. Recalling that we work in the gauge given by \eqref{gauge.C}, \eqref{nu.S.eq2}   immediately implies \eqref{nu.S.eq3}. 

Integrating \eqref{nu.S.eq3} on $[0,\uS]$ for fixed $v$ gives \eqref{uS.eq}.  Moreover, write the identity $$ \partial_v \left[  \int_0^{\uS} r|\nu|(u',v) du\right] = \frac{d\uS}{dv} +  \int_0^{\uS} \partial_v(r|\nu|)(u',v) du'= -r_0(v) |\lambda_0(v)|.$$ Note that the fact that the integrals exist and are continuous show that  $v \rightarrow \uS$ is $C^1$. Then by \eqref{Radius3} and the bootstrap assumptions we obtain \eqref{duS.eq}, which shows in view of \eqref{lambda.est2} that $\frac{d\uS}{dv}<0$ for $v$ large enough, in particular $\mathcal{S}$ is $C^1$ spacelike.

\end{proof}

\subsection{The red-shift region}\label{section.redshift} The following region $\mathcal{R} =\{ 0 \leq u \leq \uR(v),\ v \geq v_0 \}$  is named in analogy with the corresponding region near the event horizon of a black hole, see \cite{Mihalis1,MihalisPHD,Moi}. Here, we define $r(\uR(v),v) = (1-\delta) r_0(v)$ for some small $\delta>0$ to be determined.

\begin{prop}\label{prop.red} Assume that \eqref{hyp1}, \eqref{hyp2} are satisfied. Then for all $(u,v) \in \mathcal{R}$:
\begin{equation}\label{Dvphi.S.est}| D_v \phi|(u,v) \ls \frac{v^{\frac{1}{2}-2s}}{r^2(u,v)},
\end{equation}						\begin{equation}\label{Duphi.S.est}| D_u \phi|(u,v) \ls  \frac{v^{\frac{1}{2}}}{r^2(u,v)},
\end{equation}			\begin{equation}\label{DuDuphi.S.est} |  D_{u u}^2 \phi|(u,v)  \ls  \frac{v^{\frac{1}{2}}}{r^4(u,v)},
\end{equation}
\begin{equation}\label{DuOmega.S.est}
	|\rd_u \log(\Omega^2)|(u,v) \ls \frac{v}{ r^{2}(u,v)},
\end{equation}
\begin{equation}\label{DvOmega.S.est}
	|\rd_v\log(\Omega^2)|(u,v) \ls \frac{ v^{1-2s}}{r^{2}(u,v)},
\end{equation}
\begin{equation}\label{du.nu.S}
	|\rd_u(r\nu)|(u,v) \ls \exd r^{95}(u,v).
\end{equation} 

Now, assume that \eqref{hyp4} additionally holds. Then:		
\begin{equation}\label{DvDvphi.S.est}| \rd_v(r\rd_v \phi)|(u,v),\ r| \rd_{v v}^2 \phi|(u,v) \ls 	v^{\frac{1}{2}-4s} r^{-3}(u,v),
\end{equation}
\begin{equation}\label{dv.lamba.S} |\rd_v(r\lambda)|(u,v) \lesssim v^{-2s-1}.
\end{equation}
\end{prop}	
\begin{proof}
We start defining \begin{align*}
	& \Theta(u,v):= \frac{r\rd_v \phi(u,v)}{-r \lambda(u,v)}, \\ & \xi(u,v):= \frac{rD_u\phi(u,v)}{-r\nu(u,v)}, \\ & F(u,v):= \frac{\Omega^2(u,v)}{4r(u,v)} (iq_0 Q(u,v)-m^2 r^2(u,v)),\\  & \bar{F}(u,v):= \frac{\Omega^2(u,v)}{4r(u,v)} (iq_0 Q(u,v)+m^2 r^2(u,v)),\\  &  \tilde{F}(u,v):= \frac{\Omega^2(u,v)}{4} \left(1-\frac{Q^2(u,v)}{r^2(u,v)} -m^2 r^2 |\phi|^2(u,v) \right).
\end{align*}

We rewrite \eqref{Field2}, \eqref{Field3} as the following system of null transport equations. \begin{equation}\begin{split}\label{null.transport}
		& D_u \Theta(u,v) =   \frac{|\nu|(u,v)}{r(u,v)} \cdot  \xi(u,v)+ \frac{F(u,v)}{-r \lambda(u,v)} \cdot \phi(u,v) - \frac{\tilde{F}(u,v)}{-r\lambda(u,v)} \cdot \Theta(u,v),\\ & \rd_v \xi(u,v)=  \frac{|\lambda|(u,v)}{r(u,v)} \cdot  \Theta(u,v)-\frac{\bar{F}(u,v)}{-r\nu(u,v)} \cdot \phi(u,v) - \frac{\tilde{F}(u,v)}{-r\nu(u,v)} \cdot \xi(u,v).\end{split}
\end{equation}

We introduce the following additional bootstrap assumptions, for some $\Delta>1$ to be fixed later. \begin{equation}\label{BR.1}
	|\Theta|(u,v) \leq \Delta \cdot v^{s} \cdot \frac{r_0(v)}{r(u,v)},
\end{equation}
\begin{equation}\label{BR.2}
	|\xi|(u,v) \leq  v^{10s} \cdot \frac{r_0(v)}{r(u,v)}.
\end{equation}
We already remark on the fact that \eqref{BR.2} is a worse estimate than \eqref{Duphi.S.est} that we will ultimately prove. Now, integrating the equation for  $\rd_v \xi$ in $v$ gives (recalling \eqref{A.gauge}) using \eqref{Omega.boot}, \eqref{phi.boot} and \eqref{nu.S.eq3}, \eqref{rlambdaS.prop}: \begin{equation*}
	|\xi|(u,v) \leq |\xi|(u,v_0) +   \Delta v^{s} \frac{r_0(v)}{r(u,v)}+ e^{1.9 K_- v_0},
\end{equation*}  which already improves \eqref{BR.2} for $v_0$ large enough, and proves \eqref{Duphi.S.est}. Using the above estimate and \eqref{Omega.boot}, \eqref{phi.boot} and \eqref{nu.S.eq3}, \eqref{rlambdaS.prop} in the equation for $D_u \varTheta$ and integrating in $u$ gives, for some $C>0$ independent of $\Delta$.
\begin{equation*}\begin{split}
		&	|\varTheta|(u,v) \leq 	|\varTheta|(0,v)+  C \Delta\left( v^{s}\left[ \underbrace{\frac{r_0(v)}{r(u,v)}-1}_{\leq 2\delta}\right]+ e^{1.9 K_- v}\right).\end{split}
\end{equation*}
Thus, for $\Delta$ large enough and $\delta>0$ small enough, \eqref{BR.1}	is improved and  \eqref{Dvphi.S.est} are proved.

Now integrating \eqref{Omega} in $v$ using \eqref{Duphi.S.est}, \eqref{Dvphi.S.est}, \eqref{nu.S.eq3}, \eqref{rlambdaS.prop}  and the bootstrap assumptions \eqref{phi.boot}, \eqref{Q.boot} and \eqref{Omega.boot} give \eqref{DuOmega.S.est}.  \eqref{DvOmega.S.est} is obtained similarly,  using \eqref{dv_Omega.est}.  
\eqref{du.nu.S} follows applying $\rd_u$ to \eqref{Radius3} and integrating using \eqref{gauge.C} and \eqref{Omega.boot}. 

For \eqref{DuDuphi.S.est},  we can apply $\rd_u$ to  \eqref{Field3} we get, using \eqref{DuOmega.S.est},  \eqref{Omega.boot} \begin{equation*}
	\rd_v \left(  \rd_u(r\rd_u \phi)\right) =  \frac{-\rd_u(r\nu)}{r}  \rd_v \phi + 2 \frac{ |r\nu|^2}{r^3 }\rd_v \phi- \frac{\lambda \nu}{r} \rd_u \phi+ O (\exd r^{95}),
\end{equation*} from which \eqref{lambda.S.eq2}, \eqref{nu.S.eq2}, \eqref{dv.lamba.S}, \eqref{Duphi.S.est}, \eqref{Dvphi.S.est} give us \begin{equation*}
	|	\rd_v \left(  \rd_u(r\rd_u \phi)\right)| \ls r^{-5} v^{-2s+1/2} + \exd r^{94}\ls  |\lambda| r^{-4} v^{1/2} +  \exd r^{94},
\end{equation*} which gives  \eqref{DuDuphi.S.est} upon integrating.

In the rest of the proof,  we will assume the assumption \eqref{hyp4} holds. Then, \eqref{dv.lamba.S} follows, also using \eqref{dvv.r.est}. Finally we turn to the higher-order estimates: applying $\rd_v$ to  \eqref{Field2} we get, using \eqref{DvOmega.S.est}, \eqref{A.S} \begin{equation*}
	\rd_u \left(  \rd_v(r\rd_v \phi)\right) =  \frac{-\rd_v(r\lambda)}{r}  \rd_u \phi + 2 \frac{ |r\lambda|^2}{r^3 }\rd_u \phi- \frac{\lambda \nu}{r} \rd_v \phi+ O (\exd r^{95}),
\end{equation*} from which \eqref{lambda.S.eq2}, \eqref{nu.S.eq2}, \eqref{dv.lamba.S}, \eqref{Duphi.S.est}, \eqref{Dvphi.S.est} give us \begin{equation*}
	|	\rd_u \left(  \rd_v(r\rd_v \phi)\right)| \ls v^{-2s-1/2}  r^{-3} + v^{-4s+ 1/2} r^{-5}+\exd r^{94},
\end{equation*} which integrating gives for $v_0$ large enough and using that $r(u,v) \geq r_0(v) \approx v^{\frac{1}{2}-s}$ by \eqref{r.est}: \begin{equation*}
	|	\rd_v(r\rd_v \phi)|(u,v) \ls  |\rd_v(r\rd_v \phi)|(0,v) + v^{-4s+1/2} r^{-3}(u,v) + \exd \ls  v^{-s-1}  +v^{-4s+1/2} r^{-3}(u,v)  \ls  v^{-4s+1/2} r^{-3}(u,v)
\end{equation*} which gives \eqref{DvDvphi.S.est}, where we also used Proposition~\ref{prop.highorder.decay}.

\end{proof}

\subsection{The crushing region}\label{section.crushing}
We define the crushing region $	\mathcal{C}=\{ \uR(v) \leq u \leq \uS,\ v\geq v_0\}$. By Proposition~\ref{S.curve.prop},  it holds true in 	$\mathcal{C}$ that $$ u \approx v^{1-2s},\  \vR(u) \approx v.$$  Note that  also,  for all $(u,v) \in \mathcal{C}$\begin{equation} \inf_{J^{-}(u,v)\cap 	\mathcal{C}} v' = v_R(u) \approx v,\ 
\int_{\vR(u)}^{v} \frac{dv'}{v'} \lesssim 1,\   \min_{p\in J^{-}(u,v)\cap \mathcal{C}} r(p)=r(u,\vR(u)),
\end{equation} where $J^{-}(u,v)$ denotes the causal past of $(u,v)$. These estimates will be used repeatitively without a proof in this section. We will also assume that $\delta(D) \in(0,1)$ has been fixed in the previous section. The following proposition relies on a Gr\"{o}nwall argument that was previously developed in \cite{DejanAn} for \eqref{1.1}--\eqref{5.1} with $F\equiv 0$.

\begin{prop}\label{prop.crushing} Assume that \eqref{hyp1}, \eqref{hyp2} are satisfied. Then for all $(u,v) \in \mathcal{C}$:
\begin{equation}\label{Dvphi.S.est2}| D_v \phi|(u,v) \ls \frac{v^{\frac{1}{2}-2s}}{r^2(u,v)},
\end{equation}

\begin{equation}\label{Duphi.S.est2}| D_u \phi|(u,v) \ls  \frac{v^{\frac{1}{2}}}{r^2(u,v)},
\end{equation}
\begin{equation}\label{DuDuphi.S.est2}|  D_{u u}^2 \phi|(u,v)  \ls  \frac{v^{\frac{1}{2}}}{r^4(u,v)},
\end{equation}
\begin{equation}\label{DuOmega.S.est2}
	|\rd_u \log(\Omega^2)|(u,v) \ls \frac{v}{ r^{2}(u,v)},
\end{equation}
\begin{equation}\label{DvOmega.S.est2}
	|\rd_v\log(\Omega^2)|(u,v) \ls \frac{ v^{1-2s}}{r^{2}(u,v)},
\end{equation}
\begin{equation}\label{du.nu.S2}
	|\rd_u(r\nu)|(u,v) \ls\exd r^{95}(u,v),
\end{equation} 	\begin{equation}\label{dv.lamba.S2}
	|\rd_v(r\lambda)- \rd_v(r_0\lambda_0)|(u,v) \ls \exd.
\end{equation} 

Moreover, recalling the definition of  $r|\lambda|_{\mathcal{S}}(v)$ from the proof of Proposition~\ref{S.curve.prop}, we have for all $(u,v) \in \mathcal{C}$: \begin{equation}\label{rlambdaS}
	r|\lambda|(u,v) \approx r|\lambda|_{\mathcal{S}}(v) \approx v^{-2s},\ \bigl| r|\lambda|(u,v)  -r|\lambda|_{\mathcal{S}}(v) \bigr| \ls r^{90}.
\end{equation}

Moreover, assume that \eqref{hyp4} additionally holds. Then, for all $(u,v) \in \mathcal{C}$:		\begin{equation}\label{DvDvphi.S.est2}| D_v(rD_v \phi)|(u,v),\ r| D_{v v}^2 \phi|(u,v) \ls 	v^{\frac{1}{2}-4s} r^{-3}(u,v),
\end{equation}\begin{equation}\label{rlambdaS2}
	|\rd_v(r\lambda)|(u,v) \ls v^{-2s-1}.
\end{equation}
\end{prop}	

\begin{proof} As in the proof of Proposition~\ref{prop.red}, we make bootstrap assumptions of the following form: \begin{equation}\label{BC.1}
	|\Theta|(u,v),\ 	|\xi|(u,v)  \leq  v^{10s} \cdot \frac{r_0(v)}{r(u,v)}.
\end{equation}

Fix $(U,V) \in \mathcal{C}$. For any $0<R_0 \leq 1-\delta$, define $\Sigma_{R_0} =\{ (u,v) \in J^{-}(U,V),\ \frac{r(u,v)}{r_0(v)}= R_0\}$, and with the variable $R(u,v)=\frac{r(u,v)}{r_0(v)}$ $$ \Phi(R)= \sup_{(u,v) \in \Sigma_{R}} \max \{\ |\xi|(u,v),\ |\Theta|(u,v)\ \}$$

We will  prove $$ R \Phi(R) \leq (1-\delta) \Phi(1-\delta) \lesssim V^{s}$$ and evaluating at $(U,V) \in \Sigma_{R(U,V)}$ gives $$ |\xi|(U,V), |\Theta|(U,V) \lesssim \frac{V^{1/2}}{r(U,V)}.$$ We use \eqref{null.transport} which we integrate to get \begin{align*}
	&   |\Theta|(u,v) \leq C \cdot V^{s}+   \int_{ u_{\mathcal{R}}(v)}^{u} \frac{|\nu|(u',v)}{r(u',v)} \cdot  |\xi|(u',v) du' \leq  C \cdot V^{s}+   \int^{ 1-\delta}_{R(u,v)} \frac{  \Phi(R)}{R} dR ,\\ &  |\xi|(u,v) \leq C \cdot V^{s}+   \int_{ v_{\mathcal{R}}(u)}^{v}  \frac{|\lambda|(u,v')}{r(u,v')} \cdot  |\Theta|(u,v') dv' \leq  C \cdot V^{s}+  \int^{ 1-\delta}_{R(u,v)} \frac{  \Phi(R)}{R (1 - R^2 \frac{r_0 \lambda_0}{r\lambda})} dR,
\end{align*} where we used the following identity at constant $u$: $$ \frac{\lambda}{r} dv = \frac{dR}{R (1-R^2  - R^2[ \frac{r_0 \lambda_0}{r\lambda}-1])] }.$$ From the above and also \eqref{lambda.S.eq2} in Proposition~\ref{S.curve.prop}, we get for all $1-\delta \leq R \leq R(U,V)$: $$ \Phi(R) \leq C \cdot V^{s} +\int^{ 1-\delta}_{R} \frac{  \Phi(R')}{R'[1-(1+e^{\frac{K_-}{4} V})(R')^2]} dR	',$$	 thus, by Gr\"{o}nwall inequality we obtain, assuming $V\geq v_0$ large enough and \eqref{lambda.S.eq2} \begin{equation*}
	R\Phi(R)	 \leq C \cdot  V^{s} \left[\frac{1- (1+e^{\frac{K_-}{4} V})R}{1-(1+e^{\frac{K_-}{4} V})(1-\delta)} \right]^{\frac{1}{2}}\leq \frac{2C}{\sqrt{\delta}}\ V^{s},
\end{equation*}  where we have used the integral $\int \frac{dR}{R[1-a R^2]} = \log(\frac{R}{\sqrt{1-aR^2}}) $ for $a=1+e^{\frac{K_-}{4} V}$,
which concludes the proof of \eqref{Dvphi.S.est2}, \eqref{Duphi.S.est2} and closes the bootstrap assumptions \eqref{BC.1}. The remaining estimates follow easily from the integration of \eqref{Radius}, \eqref{Omega} and \eqref{Maxwell} as in the proof of Proposition~\ref{prop.red} (together with relevant commutations of the equation with $\rd_u$ and $\rd_v$).

\end{proof}

\subsection{Propagation of faster phase decay and pointwise lower bounds}\label{section.faster}

The goal of this section is to prove   point-wise lower bounds, as summarized in the  following proposition. \begin{prop}\label{lower.bound.prop1}
Assume \eqref{hyp1}, \eqref{hyp2}, \eqref{hyp4}--\eqref{hyp3} hold. Then, there exists $D_L>0$, $\alpha_{\infty} \in \RR$ such that for all $v\geq v_0$, $u_0\leq u\leq u_{\ep}(v)$, where we recall $u_{\ep}(v)$ is defined such that $r^2(u_{\ep}(v),v) = \ep v^{1-2s}$ and $\theta:= r\rd_v \phi$: \begin{equation}\label{Re.bound}
	\Re(\theta e^{-i\alpha_{\infty}} )(u,v)\geq \frac{D_L}{8} v^{-s},
\end{equation} and  \begin{equation}\label{Im.bound}
	|\Im(\theta e^{-i\alpha_{\infty}})|(u,v)\lesssim_{\ep} v^{-s-\delta}.
\end{equation}
In the rest of the paper, we take $\alpha_{\infty}=0$ with no loss of generality.
\end{prop}
\begin{proof}

First, we have to make use of \eqref{hyp3} and \eqref{hyp5} on the outgoing initial data cone $\{u=u_0\}$. We introduce the following modulus/phase decomposition of $\phi(u,v)$, where $P(u,v)\geq 0$ and $\alpha(u,v) \in \mathbb{R}$: \begin{equation}\label{modulus.phase}\phi(u,v) = P(u,v) e^{i \alpha(u,v)}.			\end{equation} In these notations, we have \begin{align} &\label{formula1} \rd_v \phi(u,v) = \left(\rd_v P(u,v) +  i P(u,v)\rd_v \alpha(u,v)\right) e^{i \alpha(u,v)},\\
	&\Im(\bar{\phi} \rd_v \phi)(u,v) = P^2(u,v) \rd_v \alpha(u,v), \\
	& |\rd_v \phi|(u,v) = \sqrt{(\rd_v P)^2(u,v)+ (P \rd_v \alpha)^2(u,v)}\label{modd_vphi.formula}. 
\end{align} Then, by \eqref{hyp2}, we have \begin{equation}\label{dvrhodecay}
	|\rd_v P|(u_0,v) \leq D_C v^{-s},
\end{equation} and since $s>1$, both $P(u_0,v)$ and  $\phi(u_0,v)$ admit a limit as $v\rightarrow +\infty$ denoted $P_{\infty}\geq0$ and $\phi_{\infty}\in \mathbb{C}$. Moreover, by \eqref{hyp5}, we have $$ P^2(u_0,v) |\rd_v \alpha|(u_0,v) \lesssim P(u_0,v) v^{-s-\delta}.$$
If $P(u_0,v)=0$, then since $|\rd_v \alpha|(u_0,v)$ is finite, $P^2(u_0,v) |\rd_v \alpha|(u_0,v)=0$. If   $P(u_0,v)\neq 0$, then this shows  $P(u_0,v) |\rd_v \alpha|(u_0,v) \lesssim  v^{-s-\delta}$. Either way, we have established that under \eqref{hyp2}, \eqref{hyp5} \begin{equation}\label{improved.phase.data}
	|\Im( \rd_v\phi \cdot e^{-i\alpha})|(u_0,v) =P(u_0,v) |\rd_v \alpha|(u_0,v)\lesssim  v^{-s-\delta}.
\end{equation}

Therefore, by \eqref{modd_vphi.formula}, \eqref{hyp3} implies that for $v_0$ large enough we have \begin{equation}\label{lower.dv.data}
	|\rd_v P|(u_0,v) \geq \frac{D_L}{2} v^{-s}.
\end{equation} This means that $\rd_v P(u_0,v)$ cannot have any zeroes (since $P$ is real-valued), and thus has a fixed signed. Therefore, we obtain integrating \eqref{lower.dv.data} on $[v,+\infty)$ that  
\begin{equation}\label{lower.dv.data2}
	\frac{D_L}{2(s-1)} v^{1-s}\leq  \bigl|P_{\infty} - P(u_0,v)\bigr|= \bigl||\phi_{\infty}| - |\phi|(u_0,v)\bigr|\leq  \frac{D_C}{s-1} v^{1-s}.
\end{equation}  If $P_{\infty}=0$, we have by \eqref{improved.phase.data} and \eqref{lower.dv.data2}: \begin{equation}\label{improved.phase.data2}
	|\rd_v \alpha|(u_0,v) \lesssim v^{-1-\delta}.
\end{equation} Note also that by \eqref{improved.phase.data}, \eqref{lower.dv.data2}, \eqref{improved.phase.data2} is still true if $P_{\infty}>0$ (since $s>1$). Either way,  \eqref{improved.phase.data2} holds true unconditionally. Hence,   $\alpha(u_0,v)$ admit a limit as $v\rightarrow +\infty$ denoted $\alpha_{\infty}\in \mathbb{R}$, $\phi_\infty=  P_{\infty} e^{i \alpha_{\infty}}$ and \begin{equation}\label{improved.phase.data3}|\alpha_{\infty} -\alpha(u_0,v)|\lesssim v^{-\delta}.
\end{equation} Therefore, by  \eqref{formula1}, \eqref{improved.phase.data}, \eqref{dvrhodecay}, we have \begin{equation}
	\bigl| \rd_v \phi(u_0,v) -  e^{i \alpha_{\infty}}  \rd_v P \bigr| \lesssim v^{-s-\delta}.
\end{equation}

With no loss of generality, we can choose $\alpha_{\infty}$ to assume any specific value (if not, consider $\tilde{\phi} = \phi e^{-i \alpha_{\infty}}$ and the same arguments go through). Recall that $\rd_v P(u_0,v)$ does not change sign: if   $\rd_v P(u_0,v)> 0$, we choose $\alpha_{\infty}=0$, but if  $\rd_v P(u_0,v)<  0$, we choose $\alpha_{\infty}=\pi$, so that either way we have,  for $v_0$ large and also using   \eqref{improved.phase.data}: \begin{align}
	\label{data.lower.lem} & \Re(\rd_v \phi)(u_0,v) \geq \frac{D_L}{4} v^{-s},  \\ & \label{data.im} |\Im(\rd_v \phi)|(u_0,v) \lesssim v^{-s-\delta}.
\end{align}
Using  \eqref{Field2'}, \eqref{Field3'} and combining Proposition~\ref{main.prop.CH}, Proposition~\ref{prop.red} and Proposition~\ref{prop.crushing} shows that \eqref{lambdanu.eq} is satisfied, and $\ttheta=\theta$, $\txi=\xi$ satisfy \eqref{eq.lemma} with $F=G$ satisfying \eqref{FG.eq} in the region $\{u_0 \leq u \leq u_{\ep}(v),\ v\geq v_0\}$. \eqref{data.lower.lem} shows that \eqref{real.lower.data} holds, therefore \eqref{Re.bound} follows from Lemma~\ref{lowerbound.nonosc.lemma}. Similarly, using  \eqref{Field2'}, \eqref{Field3'}, Proposition~\ref{main.prop.CH}, Proposition~\ref{prop.red} and Proposition~\ref{prop.crushing} allows to apply  Lemma~\ref{lemma.propagation} whose assumptions are satisfied for $\ttheta=\Im \theta$ and $\txi=\Im\xi$ and $p=s+\delta$, and deduce \eqref{Im.bound}.

\end{proof}

\subsection{Scalar field  commuted estimates and lower bounds}\label{section.commuted}

We recall that we have already proved some estimates of Proposition~\ref{main.prop.spacelike}, in particular \eqref{Q.est.S}, \eqref{phi.est.S}, \eqref{lambda.eq}, \eqref{nu.S.eq3}, \eqref{A.S}, \eqref{dv.lamba}, \eqref{Duphi}, \eqref{DuDuphi}, \eqref{Dvphi}, \eqref{DvDvphi} which closes the bootstrap assumptions  \eqref{nu.boot}, \eqref{Q.boot}, \eqref{phi.boot}. We, however, still need to close the (most important) bootstrap assumption \eqref{Omega.boot} which requires more refined commuted scalar field estimates. They involve the spacelike vector field $X$ such that $X(r)=0$ defined as 
\begin{equation*}\begin{split}
	&X=  \frac{1}{r|\rd_v r|(u,v)}\ \rd_v- \frac{1}{r|\rd_u r|(u,v)}\ \rd_u.
\end{split}
\end{equation*}
Re-writing \eqref{Field4} using the definition of $X$ gives the following system of equations:
\begin{equation}\begin{split}\label{FieldX}
	&\rd_u(r^2 \rd_v \phi) = \rd_v(r^2 \rd_u \phi)=- [ r|\nu|][r|\lambda|]X\phi  +\frac{ q_{0}i \Omega^{2}}{4}Q \phi
	-\frac{ m^{2}r^2\Omega^{2}}{4}\phi- i q_{0} A_{u}\phi [r\lambda]-i q_0 A_{u}r^2\partial_{v}\phi.
	\end{split}\end{equation}
	
	The logic of the proof in this section is inspired by the treatment of spacelike singularities with non-degenerating Kasner exponents provided in \cite{DejanAn,Warren1}, especially the work \cite{Warren1} which introduced the vector field $X$ and related commutator estimates. In the present problem, however, the degeneration of Kasner exponents additionally requires tracking temporal weights in $v$, and thus, $u$ and $v$ are not interchangeable. In this section, we will assume \eqref{hyp1}, \eqref{hyp2} and \eqref{hyp4}-\eqref{hyp3} are satisfied. The main results  can be stated as follows:
	
	\begin{prop} Assume \eqref{hyp1}, \eqref{hyp2} and \eqref{hyp4}-\eqref{hyp3} hold. Then
\begin{equation}
	|	\rd_v (r\lambda)|(u,v) \ls v^{-2s-1}.
\end{equation}
Moreover, for all $\ep>0$ small enough, and $(u,v)\in \{(u,v)\ r(u,v) \leq \ep r_0(v)\}$:
\begin{equation}
	r^2 |\rd_v \phi|(u,v) \geq D \cdot \ep\ v^{\frac{1}{2} -2s}.
\end{equation}

\end{prop}

We start with a lemma providing commutator estimates (see the analogous computation in \cite{Warren1}).

\begin{lemma}\label{lemma.com} Assume \eqref{hyp1}, \eqref{hyp2} and \eqref{hyp4} hold. Then the following estimates are satisfied: $[X,f(r)]=0$ and:

\begin{equation}\label{[du,X]}|	[\rd_u,X] F|(u,v) \ls  \exd r^{95}(u,v) \cdot \left[|\rd_u F|(u,v)+ |\rd_v F|(u,v) \right],\end{equation}
\begin{equation}\label{[dv,X]}		|	[\rd_v,X] F|(u,v) \ls  v^{2s-1} | \rd_v F|(u,v) +\exd r^{95}(u,v) \cdot  | \rd_u F|(u,v),	\end{equation}	\begin{equation}\label{[dvrdu,X]}		|	[\rd_v(r \rd_u),X] F|(u,v) \ls   v^{2s-1}|\rd_v( r \rd_u F )|+ \exd r^{95}(u,v)  [ |\rd_u( r \rd_v F )|+|\rd_{u u}^2 F|+ |\rd_{v v}^2 F|+|\rd_{u} F|+ |\rd_{v} F| ],	\end{equation} 
\begin{equation}\label{[durdv,X]}	
	|	[\rd_u(r \rd_v),X] F|(u,v) \ls   v^{2s-1}|\rd_u( r \rd_v F )|+ \exd r^{95}(u,v)  [ |\rd_v( r \rd_u F )|+|\rd_{u u}^2 F|+ |\rd_{v v}^2 F|+|\rd_{u} F|+ |\rd_{v} F| ].	\end{equation}  

\end{lemma}

\begin{proof}
\eqref{[du,X]} follows from \eqref{du.nu} and \eqref{Radius3}, and \eqref{[dv,X]} from \eqref{dv.lamba}, \eqref{lambda.eq}. 
For \eqref{[dvrdu,X]}, we first note the identity $$ 	[\rd_v(r \rd_u),X] F = \rd_v \left( r [\rd_u,X] F \right) + [\rd_v,X] \left( r \rd_u F \right).
$$   \eqref{[durdv,X]}  follows from similar considerations using \eqref{[du,X]}, \eqref{[dv,X]} and the other above estimates (most importantly those from Proposition~\ref{prop.crushing} and the bootstrap assumption \eqref{Omega.boot}).

\end{proof}
Then we obtain the $X$-commuted wave equation, up to errors controlled in the following corollary.
\begin{cor} Assume \eqref{hyp1}, \eqref{hyp2} and \eqref{hyp4} hold. For all $0\leq u \leq u_{\mathcal{S}}(v)$, $v\geq v_0$ we have
\begin{equation}\label{proof.X.1}
	|	\rd_u ( r \rd_v X \phi ) + \lambda \rd_u (X \phi)| \ls  r^{-3}(u,v) v^{-\frac{1}{2}},
\end{equation}
\begin{equation}\label{proof.X.2}
	|	\rd_v ( r \rd_u X \phi ) + \nu \rd_v (X \phi)|  \ls r^{-3}(u,v) v^{-\frac{1}{2}}.
\end{equation} or equivalently, introducing the notations consistent with sections~\ref{section.redshift} and \ref{section.crushing} $\varTheta[f]= \frac{r \rd_v f}{r|\lambda|}$ and  $\xi[f]  = \frac{r D_u f}{r|\nu|}$: \begin{equation}\label{proof.X.1'} \bigl|\rd_u( \varTheta[X\phi]) +\frac{\nu}{r} \xi[X\phi] \bigr| \lesssim r^{-3}(u,v) v^{-\frac{1}{2}+2s} \approx |\nu|(u,v) r^{-2}(u,v) v^{-\frac{1}{2}+2s} ,
\end{equation}
\begin{equation}\label{proof.X.2'}\bigl|\rd_v( \xi[X\phi] )+\frac{\lambda}{r} \varTheta[X\phi]\bigr| \ls r^{-3}(u,v) v^{-\frac{1}{2}} \approx  |\lambda|(u,v) r^{-2}(u,v) v^{-\frac{1}{2}+2s}.
\end{equation} 
\end{cor}

\begin{proof} Applying  $X$ to \eqref{Field2} and using Lemma~\ref{lemma.com} 
combined with 
Proposition~\ref{prop.crushing}, we get 	$$|	\rd_u ( r \rd_v X \phi ) + \lambda \rd_u (X \phi)| \ls \left( v^{2s-1} |\lambda|+ |X( \lambda)|\right) |\rd_u \phi|+  \exd r^{95}(u,v),$$ and \eqref{proof.X.1} then follows from \eqref{Duphi} and \eqref{lambda.eq}, \eqref{dv.lamba}. \eqref{proof.X.2} is obtained similarly, applying $X$ to \eqref{Field3} and using  \eqref{proof.X.1'},  \eqref{proof.X.2'} follow from \eqref{Radius}, \eqref{Omega.boot}, \eqref{lambda.eq}, \eqref{nu.S.eq3}.
\end{proof}

Finally, we turn to $X$-commuted scalar field estimates. We start from the red-shift region $\mathcal{R}$.

\begin{prop}\label{prop.commuted2} Assume \eqref{hyp1}, \eqref{hyp2} and \eqref{hyp4} hold. For all $(u,v) \in \mathcal{R}$, we have		\begin{equation}\label{X.highr.R}			|X\phi|(u,v),\ r^2|\rd_u X\phi|(u,v),\ \frac{r^2|\rd_v X\phi|(u,v)}{r|\lambda|(u,v)}   \ls v^{2s-\frac{1}{2}} .		\end{equation}	\end{prop}

\begin{proof}		Introduce the following bootstrap assumption, for a large constant $\Delta>0$ to be determined later. \begin{equation}\label{B.X}			r^2|\rd_v X\phi|(u,v)  \leq \Delta v^{-\frac{1}{2}}.		\end{equation} Note that by \eqref{Dvdvphi.est}, \eqref{lambda.est2} and \eqref{Field2}, \eqref{B.X} is satisfied at $u=0$. Then, using \eqref{B.X} in \eqref{proof.X.2} and integrating in $v$, we get, for some constant $C>0$ independent of $\Delta$: $$ |r\rd_u X \phi|(u,v) \leq  |r\rd_u X \phi|(u,v_0)+  C (1+\Delta) v^{2s-\frac{1}{2}} r^{-1}(u,v) \leq 2C (1+\Delta) v^{2s-\frac{1}{2}} r^{-1}(u,v).$$		Now using the above in \eqref{proof.X.1}, we obtain (also using \eqref{nu.S.eq3}), for some constant $C'>0$ independent of $\Delta$ $$| \rd_u ( r \rd_v X \phi ) | \leq C' (1+\Delta) r^{-2} |\nu|(u,v) v^{-\frac{1}{2}}.$$ which upon integrating in $u$ gives$$ | r \rd_v X \phi|(u,v) \leq  |r \rd_v X \phi|(0,v)+  2C' (1+\Delta)v^{-1/2} \delta\ r_0^{-1}(v),$$ hence  for some constant $C''>0$ independent of $\Delta$ $$ | r^2 \rd_v X \phi|(u,v) \leq   r_0(v)|r \rd_v X \phi|(0,v)+  C'' (1+\Delta)v^{-1/2} \delta ,$$ which retrieves bootstrap \eqref{B.X} for $\delta>0$ small enough and also shows the part of \eqref{X.highr.R} concerning $r^2|\rd_u X\phi|(u,v)$ and $r^2|\rd_v X\phi|(u,v)$ (also using \eqref{lambda.eq}). Then, integrating in $u$ the estimate for  $r^2|\rd_u X\phi|(u,v)$ (also using \eqref{nu.S.eq3}) gives the  part of \eqref{X.highr.R} concerning $X\phi$, also using \eqref{Dvphi.est}, \eqref{duphi.est} at $u=0$.	\end{proof}

Now, we continue with the crushing region $\mathcal{C}$ on which $ 0 \leq r\leq(1-\delta) r_0(v) $.
\begin{prop}\label{prop.commuted3}  \eqref{hyp1}, \eqref{hyp2} and \eqref{hyp4}. For any $ u_{\mathcal{C}}(v) \leq u \leq u_{\mathcal{S}}(v)$, $v\geq v_0$, we have:
\begin{equation}\label{X.highr}
	|X\phi|(u,v)   \ls v^{2s-\frac{1}{2}} (1+\log^2(\frac{r_0(v)}{r(u,v)})),
\end{equation}		\begin{equation}\label{X.highr2}
	r^2|\rd_u X\phi|(u,v),\ \frac{r^2|\rd_v X\phi|(u,v)}{r|\lambda|(u,v)}   \ls v^{2s-\frac{1}{2}} (1+\log(\frac{r_0(v)}{r(u,v)})),
\end{equation}
\begin{equation}\label{X.higherorder}
	|\rd_u^2 X\phi|(u,v)  \ls r^{-2 }v^{2s-\frac{1}{2}} (1+\log(\frac{r_0(v)}{r(u,v)})).
\end{equation}
\end{prop}
\begin{proof} We integrate \eqref{proof.X.1'} in $u$ making use of \eqref{X.highr.R} to control the boundary terms on the $\{(\uR(v),v)\}$ curve.
\begin{equation}\label{I.1}
	\bigl| \varTheta[X\phi](u,v)- \underbrace{\varTheta[X\phi](\uR(v),v)}_{\lesssim v^{-1+3s}}\bigr| \lesssim r^{-1}(u,v)v^{-\frac{1}{2}+2s}+ \int_{\uR(v)}^{u} |\nu|r^{-1} |\xi(X\phi)|(u',v) du'.
\end{equation} 
Similarly, we integrate \eqref{proof.X.2'} in $v$ and get \begin{equation}\label{I.2}
	\bigl| \xi[X\phi](u,v)- \underbrace{ \xi[X\phi](u,v_{\mathcal{R}}(u))}_{\lesssim \delta^{\frac{s-1}{2s-1}} u^{- \frac{s-1}{2s-1}} \lesssim v^{3s-1}}\bigr| \lesssim  r^{-1}(u,v)v^{-\frac{1}{2}+2s}+ \int_{v_{\mathcal{R}}(u)}^{v}|\lambda| r^{-1} |\varTheta[X\phi]|(u,v') dv'
\end{equation}

Now we proceed as in the proof of Proposition~\ref{prop.crushing}: let $(U,V)$ such that  $ \uR(V) \leq U \leq u_{\mathcal{S}}(V)$ and recall the definition of the coordinate $R(u,v) = \frac{r(u,v)}{r_0(v)}$, and $\Sigma_{R_0} =\{ (u,v) \in J^{-}(U,V),\ \frac{r(u,v)}{r_0(v)}= R_0\}$.   We introduce the following notation  $$\Phi[f](R)= \sup_{(u,v) \in \Sigma_R}\{ |\xi[f]|(u,v),\ |\varTheta[f]|(u,v)\}$$ as a generalization of the ones introduced in the proof of Proposition~\ref{prop.crushing}. Analogously, we obtain \begin{equation*}
	\Phi[X\phi](R) \leq C R^{-1} V^{-1+3s}+ \int^{(1-\delta)}_{R} \frac{\Phi[X\phi](R')}{ R' [1-(1+e^{\frac{K_-}{4} V}) (R')^2]} dR',
\end{equation*} where we have used the fact that $v^{-1+3s} \lesssim r^{-1}(u,v)v^{-\frac{1}{2}+2s}$ in this region. Then by Gr\"{o}nwall's inequality, we obtain, proceeding as in the proof of Proposition~\ref{prop.crushing}:

\begin{equation*}
	R\Phi[X\phi](R) \leq  C V^{-1+3s} \left[1+ \log(\frac{1-\delta}{R}) \right] (1-(1+e^{\frac{K_-}{4} V}) R^2)^{\frac{1}{2}}\ls \log( R^{-1}) V^{-1+3s},
\end{equation*} from which we get \eqref{X.highr2}. Then \eqref{X.highr} follows from integrating  \eqref{X.highr2} in $u$.
\end{proof}

Now we turn to the actual point-wise lower bound, which is obtained as a combination of  Proposition~\ref{lower.bound.prop1} and the commuted estimates of Proposition~\ref{prop.commuted3} in the region $\{(u,v),\ r(u,v) \leq \ep r_0(v)\}$.
\begin{prop}\label{lower.prop}
We assume \eqref{hyp1}, \eqref{hyp2}, \eqref{hyp4}-\eqref{hyp3} hold. Then, there exists $\ep>0$ sufficiently small such that for all $ u_{\ep}(v)\leq u\leq \uS $, $v\geq v_0$ the following point-wise lower bounds hold:
\begin{equation}\label{final}
	\frac{ r^2  \bigl|\rd_v \phi\bigr|}{r|\lambda|}(u,v) \gtrsim \ep \cdot v^{\frac{1}{2}},
\end{equation}
\begin{equation}\label{finalU}
	\frac{ r^2  \bigl|\rd_u \phi\bigr|}{r|\nu|}(u,v) \gtrsim \ep \cdot v^{\frac{1}{2}}.
\end{equation}
\end{prop}
\begin{proof}
By Proposition~\ref{lower.bound.prop1}, we have for $D= \frac{D_L}{8}$, $v\geq v_0$:
\begin{equation}\label{dyadic}
	|r\rd_v \phi|(u_{\ep}(v),v) \geq  D v^{-s}.
\end{equation}


Note that by \eqref{FieldX} and the estimates \eqref{Omega.boot}, \eqref{Q.est.S}, \eqref{phi.est.S}, \eqref{Dvphi}, \eqref{Duphi}, \eqref{lambda.eq}, \eqref{nu.S.eq3}, \eqref{A.S} we obtain: \begin{equation}\begin{split}\label{X.eq}
		&\bigl|\rd_u(r^2 \rd_v \phi) +[ r|\nu|][r|\lambda|]X\phi \bigr|\ls \exd r^{90}(u,v),\\ & \bigl|\rd_v(r^2 \rd_u \phi) +[ r|\nu|][r|\lambda|]X\phi \bigr| \ls \exd r^{90}(u,v).
	\end{split}
\end{equation}
Integrating \eqref{X.eq} in $u$, using also \eqref{X.highr} gives for $u\geq u_{\ep}(v)$ \begin{equation}\label{diff}
	\bigl|r^2 \rd_v \phi(u,v)- r^2 \rd_v \phi(u_{\ep}(v),v) \bigr| \lesssim \ep^2 \log^2(\ep)\cdot v^{\frac{1}{2}-2s}.
\end{equation}

Therefore, by \eqref{dyadic} and choosing $\ep>0$ small enough, we have \begin{equation}\label{lower.dvphi}
	\bigl|r^2 \rd_v \phi|(u,v) \geq \frac{D}{2} \ep v^{\frac{1}{2}-2s},
\end{equation} thus \eqref{final} holds. Now, note that   \eqref{X.highr} directly multiplied by $r^2$ and for any $u\geq u_{\ep}(v)$ gives \begin{equation}\label{dudv}
	\bigl|\frac{ r^2 \rd_v \phi}{r|\lambda|}(u,v) - r^2 \rd_u \phi(u,v) \bigr| \lesssim \ep^2 \log^2(\ep)\cdot v^{\frac{1}{2}},
\end{equation}from which we deduce \eqref{finalU}.
\end{proof}

\emph{Closing the remaining bootstrap assumption \eqref{omega.est.proof} and concluding the proof of Proposition~\ref{main.prop.spacelike}}. We first note that for $0 \leq u \leq u_{\frac{\ep}{10}}(v)$, we can simply use the monotonicity from \eqref{RaychU} and obtain, using \eqref{nu.S.eq3} \begin{equation}\label{omega.trivial}
\Omega^2(u,v) \leq |\nu|(u,v)  \frac{\Omega^2(u_0,v)}{|\nu|(u_0,v)} \lesssim r^{-1}(u,v) e^{1.9 K_- v}\ls e^{1.8 K_- v} r^{1000}(u,v),
\end{equation}  for $v_0$ large enough, which improves \eqref{omega.est.proof}. The delicate region is that where 	$ u \geq u_{\frac{\ep}{10}}(v)$. We will look at the larger region  $u \geq u_{\ep}(v)$ for now: we turn to \eqref{RaychV} which we write as \begin{equation*}
\rd_v ( \log[\frac{|\lambda|}{\Omega^2}])= \frac{r |\rd_v \phi|^2}{|\lambda|}
\end{equation*} and apply \eqref{final} together with \eqref{lambda.eq} to get that for some $D>0$ $$ 	\rd_v ( \log[\frac{|\lambda|}{\Omega^2}]) \gtrsim r^{-2} \ep^2 v^{1-2s} \geq D \ep^2 r^{-1} |\lambda| v. $$ Since $\ep>0$ can be fixed, we will ignore the $\ep$-dependence in what follows and conclude from integrating the above in $v$ that there exists $D'>0$ such that for all $u\geq u_{\ep}(v)$, $v\geq v_0$: \begin{equation}\label{omega.last}
\Omega^2(u,v) \leq \frac{\Omega^2(u,v_\ep(u))}{|\lambda|(u,v_{\ep}(u))} |\lambda|(u,v) [\frac{r(u,v)}{\ep r_0(v)}]^{ D' v} \ls v^{-s - \frac{1}{2}} [\frac{r(u,v)}{\ep r_0(v)}]^{ D' v-1} \ls [\frac{r(u,v)}{\ep r_0(v)}]^{ \frac{D' v}{2}},
\end{equation}  where we have used \eqref{omega.trivial} to control $\frac{\Omega^2(u,v_\ep(u))}{|\lambda|(u,v_{\ep}(u))}$. By monotonicity, we still have $\Omega^2(u,v) \ls r^{-1}(u,v) e^{1.9 K_- v}$ , so interpolating between that and \eqref{omega.last} gives $$ \Omega^2(u,v ) \ls e^{\frac{4K_-}{5} v} [\frac{r(u,v)}{\ep r_0(v)}]^{ \frac{D' v}{4} -\frac{1}{2}}\ls e^{\frac{3K_-}{4}v} r^{1000}(u,v),  $$ as long as  $r\leq \frac{\ep}{10} r_0(v)$, which  improves on  \eqref{omega.est.proof} for $u\geq u_{\frac{\ep}{10}}$. Therefore the part of \eqref{Omega.boot}  regarding $\Omega^2$ is retrieved. Integrating \eqref{Maxwell} in $v$ to estimate $A_u$ completes the improvement of \eqref{Omega.boot}. Also note that, treating the estimate a bit more carefully and taking $\ep>0$ small also gives \eqref{Omega.S} and \eqref{dvlogOmega}. \eqref{dulogOmega} follows analogously invoking \eqref{RaychU} and the estimates \eqref{du.nu} and \eqref{finalU}. Thus, the proof of Proposition~\ref{main.prop.spacelike} is concluded. 

\section{Applications to refined Kasner asymptotics}\label{kasner.section}

In this section, we provide more refined estimates near the Kasner singularity.
These estimates follow essentially immediately from Proposition~\ref{main.prop.spacelike}, but offer an interesting characterization of $\mathcal{S}$.

\subsection{Higher order geometric estimates and other refined estimates}

We start proving higher order commuted estimates that will be useful in the next Section~\ref{kasner.section}.
\begin{prop}\label{prop.commuted} The following estimates hold true for all $0\leq u \leq \uS$, $v\geq v_0$.
\begin{equation}\label{dudunu}
	\bigl|\rd^2_{u} (r\nu)\bigr|(u,v) \lesssim e^{1.9 K_- v} r^{D v}(u,v),
\end{equation}\begin{equation}\label{dudulogomega}
	\bigl|\rd^2_{u} \log(\Omega^2)\bigr|(u,v) \lesssim  r^{-4}(u,v) \cdot v,
\end{equation}
\begin{equation}\label{duduA.S}	\bigl| \rd_{u}^2 A_u(u,v) \bigr| \lesssim \exd r^{95}(u,v) ,\end{equation}
\begin{equation}\label{DuDuDuphi.S.est} |  D_{u}^3 \phi|(u,v)  \ls  \frac{v^{\frac{1}{2}}}{r^6(u,v)},
\end{equation}\begin{equation}\label{DvDvDvphi.S.est} |  \rd_{v}^3 \phi|(u,v)  \ls  \frac{v^{\frac{1}{2}-6s}}{r^6(u,v)},
\end{equation}
\begin{equation}\label{DvDuDvphi.S.est} |  D_{u}^2 \rd_v \phi|(u,v),\  \frac{|  D_{u } \rd_v^{2} \phi|(u,v)}{r|\lambda|(u,v)}   \ls  \frac{v^{\frac{1}{2}-2s}}{r^6(u,v)}.
\end{equation}
\end{prop}
\begin{proof}
We may formulate a bootstrap assumption of the form \begin{align*}
	& \bigl|\rd^2_{u} \log(\Omega^2)\bigr|(u,v) \lesssim  r^{-5}(u,v) v^{10},
\end{align*} and then commute \eqref{Radius3} with $\rd^2_{u}$ to obtain \eqref{dudunu}, using the estimates of Proposition~\ref{main.prop.spacelike} for $0 \leq u\leq \uS$ and the estimates of Proposition~\ref{main.prop.CH} for $u_0 \leq u\leq 0$. Then, we come back to \eqref{RaychU} which we commute with $\rd_u$ once and obtain \begin{equation*}
	\rd_u^2 (\log |\nu|)- \rd^2_{u} \log(\Omega^2)= \frac{\rd^2_{u }r}{|\rd_u r|^2} r|D_u \phi|^2- |D_u \phi|^2 + \frac{r}{|\rd_u r|} \rd_u ( |D_u \phi|^2).
\end{equation*} Then \eqref{dudulogomega} follows from \eqref{dudunu} and the other estimates of Proposition~\ref{main.prop.spacelike} and Proposition~\ref{main.prop.CH}, respectively. 
For \eqref{duduA.S} and \eqref{DuDuDuphi.S.est}, we  respectively commute \eqref{Maxwell} and \eqref{Field4} with $\rd_{uu}^2$ and use the above estimates. 
\eqref{DvDvDvphi.S.est} and \eqref{DvDuDvphi.S.est}  are obtained similarly.   Note that along the way, we have also improved on the bootstrap assumption.
\end{proof}

We then proceed to obtain refined estimates on the difference between $\frac{\rd_u\log(\Omega)}{r|\nu|}$ and $\frac{\rd_v\log(\Omega)}{r|\lambda|}$, i.e., $X \log(\Omega)$, which exploits the better estimate obtained on $X\phi$ in Section~\ref{section.commuted}. We formulate this as a corollary to Proposition~\ref{prop.commuted2} and Proposition~\ref{prop.commuted3}. These estimates correspond to the AVTD behavior mentioned in Section~\ref{proof.section5}.

\begin{cor} For any $0\leq u \leq \uS $, $v\geq v_0$:
\begin{equation}\label{domega.diff}
	r^2 v^{-1}|X\log(\Omega^2)|= \bigl| r^2 \frac{v^{-1}\rd_v \log(\Omega^2)}{r|\lambda|}- r^2 \frac{v^{-1}\rd_u \log(\Omega^2)}{r|\nu|} \bigr| \lesssim 
	\frac{r^2}{r_0^2(v)} (1+\log^2(\frac{r_0(v)}{r(u,v)})).
\end{equation}
\end{cor}
\begin{proof}
Using \eqref{RaychU} with \eqref{du.nu} from Proposition~\ref{main.prop.spacelike} we obtain  \begin{equation}\label{new.ray.U}
	\frac{|\nu|}{r} - \rd_u \log(\Omega^2) - r^3|\nu| [\frac{D_u \phi}{r|\nu|}]^2=O(e^{1.9 K_- v} r^{Dv}).
\end{equation} Similarly, from   \eqref{RaychV} with \eqref{dv.lamba} from Proposition~\ref{main.prop.spacelike}, we obtain 
\begin{equation}\label{new.ray.V}
	\frac{|\lambda|}{r} - \rd_v \log(\Omega^2) - r^3|\lambda| [\frac{\rd_v \phi}{r|\lambda|}]^2=O(v^{-1}).
\end{equation}
Then, we multiply \eqref{new.ray.U} by $r|\lambda|$ and \eqref{new.ray.V} by $r|\nu|$  and take the difference to obtain the new estimate: $$  - r|\lambda|\rd_u \log(\Omega^2) + r|\nu| \rd_v \log(\Omega^2) +r^4|\lambda||\nu| \left([\frac{\rd_v \phi}{r|\lambda|}]^2-[\frac{\rd_u \phi}{r|\nu|}]^2  \right)=O(v^{-1}). $$  Then, we notice that $$\bigl| - r|\lambda|\rd_u \log(\Omega^2) + r|\nu| \rd_v \log(\Omega^2) \bigr| \ls r^4|\lambda||\nu|  |X\phi| |\frac{\rd_v \phi}{r|\lambda|}+ \frac{\rd_u \phi}{r|\nu|} \bigr|+ O(v^{-1}) \lesssim (1+\log^2(\frac{r_0(v)}{r(u,v)})) +v^{-1},$$  where in the last line we have used \eqref{X.highr} from Proposition~\ref{prop.commuted3}. \eqref{domega.diff} then follows immediately.

\end{proof}

Further, we prove a higher-order commutator lemma, which will be useful below to commute the wave equation with the second order differential operator $f\rightarrow \rd_u [Xf]$, also denoted $\rd_u X$.
\begin{lemma}\label{lemma.com2}
We	introduce the notation $|\rd F|= |\rd_u F|+|\rd_v F|$ and $|\rd^2 F|= |\rd^2_{uu} F|+|\rd^2_{uv} F|+|\rd^2_{vv} F|$. 
\begin{equation}\label{[du,duX]}
	|	[\rd_u,\rd_u X] F|(u,v)\ls  \exd r^{95}(u,v) \cdot \left[|\rd_{uu}^2 F|(u,v)+ |\rd_{u} \rd_v F|(u,v)+ |\rd F|(u,v) \right],
\end{equation}
\begin{equation}\label{[du,dvX]}
	|	[\rd_u,\rd_v X] F|(u,v)\ls\exd r^{95}(u,v) \cdot \left[{|\rd_{v}^2 F|(u,v)+|\rd_u\rd_{v} F|(u,v)}+ |\rd F|(u,v) \right],
\end{equation}
\begin{equation}\label{[dv,duX]}
	|	[\rd_v,\rd_u X] F|(u,v)\ls {v^{2s-1} |\rd_u \rd_v F|(u,v)}+ \exd r^{95}(u,v) \cdot \left[|\rd_{u}^2 F|(u,v)+ |\rd F|(u,v) \right],
\end{equation}
\begin{equation}\label{[dvrdu,duX]}	\begin{split}&	|	[\rd_v(r \rd_u),\rd_uX] F|(u,v) \ls  v^{2s-1}|\rd_{u} \rd_v (r\rd_u F)|+|\rd_v( |\nu| \rd_u X F )|\\ &+ \exd r^{95}(u,v)  \left[|\rd_{u}^3 F|+ |\rd_{u} \rd_v^{2} F|+|\rd_{u}^2  \rd_v F|+ |\rd^2 F|+ |\rd F| \right],\end{split}	\end{equation} 
\begin{equation}\label{[durdv,duX]}\begin{split}	&	|	[\rd_u(r \rd_v),\rd_u X]			F|(u,v) \ls   v^{2s-1}|\rd^2_{u }(r\rd_v F)|+ |\rd_u( |\nu| \rd_v X F )|\\ &+ \exd r^{95}(u,v)  \left[|\rd_{u}^3 F|+ |\rd_{v}^2 \rd_u F|+|\rd_{u}^2  \rd_v F|+ |\rd^2 F|+ |\rd F| \right] .\end{split}	\end{equation}  
\end{lemma}
\begin{proof} First, \eqref{[du,duX]} follows from the identity $$[\rd_u,\rd_u X]= -\rd_u ([X,\rd_u]),$$ from which we can repeat the proof of \eqref{[du,X]}, after an extra $\rd_u$ commutation; \eqref{[du,dvX]}, \eqref{[dv,duX]} 
are derived similarly. For \eqref{[dvrdu,duX]}, we start noticing the standard commutator identity $$ [\rd_v(r\rd_u), \rd_u X] = \underbrace{[\rd_v(r\rd_u), \rd_u ]}_{=\rd_v(|\nu| \rd_u)}X + \rd_u[\rd_v(r\rd_u),  X]= \rd_v(|\nu| \rd_u X) +  \rd_u[\rd_v(r\rd_u),  X] . $$ Then, \eqref{[dvrdu,duX]}  follows from a combination of the estimates of Proposition~\ref{main.prop.spacelike} and Proposition~\ref{prop.commuted}, repeating the proof of Lemma~\ref{lemma.com}. \eqref{[durdv,duX]} 
is then obtained similarly.
\end{proof}

\begin{cor}\label{cor.ddX} For all $0 \leq u \leq \uS$, $v\geq v_0$:
\begin{equation}\label{dduXphi}
	|\rd^2_{u} X \phi|(u,v),\ \frac{|\rd_u \rd_v X \phi|(u,v) }{r|\lambda|(u,v)}
	\ls r^{-4} v^{2s-\frac{1}{2}} (1+\log(\frac{r_0(v)}{r(u,v)})).
\end{equation}
\end{cor}
\begin{proof} Commuting \eqref{Field2}, \eqref{Field3} with $\rd_u X$ gives rise to the following:
\begin{equation*}\begin{split}
		&	\bigl| \rd_u\varTheta[\rd_u X] +\frac{\lambda}{r} \zeta[\rd_u X]\bigr|\ls |\rd_u X \lambda| |\rd_u \phi| +  \bigl|\rd_u X\left( -iq_0 A_u(r\rd_v \phi-\lambda \phi)+\frac{ \Omega^{2} \cdot \phi}{4r} \cdot ( i q_{0} Q-m^2 r^2)\right)\bigr|\\ &+ |	[\rd_u(r \rd_v),\rd_uX] \phi|+ |\lambda||[\rd_u,\rd_u X] \phi|+ \exd r^{100}(u,v), \\ &  	\bigl| \rd_v \zeta[\rd_u X] +\frac{\nu}{r} \varTheta[\rd_u X]\bigr|\ls |\rd_u X \nu| |\rd_v \phi|+ \bigl|\rd_u X\left( -iq_0 A_u(r\rd_v \phi-\lambda \phi)+\frac{ \Omega^{2} \cdot \phi}{4r} \cdot ( -i q_{0} Q-m^2 r^2)\right)\bigr|\\ & + |	[\rd_v(r \rd_u),\rd_uX] \phi|+ |\nu||[\rd_v,\rd_u X] \phi|+ \exd r^{100}(u,v). \end{split}
\end{equation*} 
From the above commuted estimates and \eqref{[du,duX]}, \eqref{[dv,duX]} it is easy to see that $$ \bigl|\rd_u X\left( -iq_0 A_u(r\rd_v \phi-\lambda \phi)+\frac{ \Omega^{2} \cdot \phi}{4r} \cdot (\pm i q_{0} Q-m^2 r^2)\right)\bigr|,\ |\nu||[\rd_v,\rd_u X] \phi|,\ |\lambda||[\rd_u,\rd_u X] \phi| \ls  e^{\frac{K_-}{4} v} r^{95}(u,v).$$   Now, by Lemma~\ref{lemma.com}, Lemma~\ref{lemma.com2}, we have, by \eqref{Field4} and the estimates of Proposition~\ref{main.prop.spacelike} and Proposition~\ref{prop.commuted}: \begin{align*}
	&|	[\rd_v(r \rd_u),\rd_uX] \phi|(u,v) \ls  v^{2s-1}|\rd^2_{uv}(r\rd_u \phi)|+|\rd_v( |\nu| \rd_u X \phi )|+ {\exd r^{95}(u,v)}  \left[|\rd_{u u u}^3 \phi|+ |\rd_{u v v}^3 \phi|+|\rd_{u u v}^3 \phi|+ |\rd^2 \phi|+ |\rd \phi| \right]\\ &  \ls v^{2s-1} \bigl|\rd_u \left(-\nu \rd_v\phi-iq_0 A_u(r\rd_v \phi-\lambda \phi)+\frac{ \Omega^{2} \cdot \phi}{4r} \cdot (\pm i q_{0} Q-m^2 r^2) \right)\bigr| + |\rd_v\nu \cdot\rd_u X\phi|+ |\nu\cdot \rd_{uv}^2 X \phi|+{\exd r^{80}(u,v)}   \\& \ls v^{2s-1} r^{-1}[|\rd_u \rd_v \phi|+ r^{-2} |\rd_v \phi|]+r^{-3} v^{-2s}|\rd_u X\phi|+ r^{-2} |\rd_v(r\rd_u X\phi) - \lambda \rd_u X \phi|+  {\exd r^{80}(u,v)}\\ & \ls    r^{-5}{v^{-\frac{1}{2}} (1+\log(\frac{r_0(v)}{r(u,v)}))}+ r^{-2} |X\left( \rd_v(r\rd_u \phi)\right)|+ r^{-2}   |[X, \rd_v(r\rd_u )]\phi|+  {\exd r^{80}(u,v)} \\ & \ls     r^{-5}{v^{-\frac{1}{2}} (1+\log(\frac{r_0(v)}{r(u,v)}))}+ r^{-2} |X(\nu \rd_v \phi)|+ r^{-2}   v^{-12s}|\rd_v(r\rd_u \phi)|+  {\exd r^{80}(u,v)} \\ & \ls r^{-5}{v^{-\frac{1}{2}} (1+\log(\frac{r_0(v)}{r(u,v)}))}+ r^{-2} |X(\nu)| |\rd_v \phi|+ r^{-3}|X \rd_v \phi|+ r^{-2}   v^{2s-1}|\nu \rd_v \phi|+  {\exd r^{80}(u,v)}\\ & \ls r^{-5}{v^{-\frac{1}{2}} (1+\log(\frac{r_0(v)}{r(u,v)}))}+ r^{-4} v^{-2s+\frac{1}{2}} |[X,\rd_u] r| + r^{-3}|[X ,\rd_v ]\phi|+ r^{-3} |\rd_v X\phi|+  {\exd r^{80}(u,v)}\\ & \ls r^{-5}{v^{-\frac{1}{2}} (1+\log(\frac{r_0(v)}{r(u,v)}))}+ r^{-3}v^{2s-1} |\rd_v \phi|+{\exd r^{80}(u,v)}\ls  r^{-5}{v^{-\frac{1}{2}} (1+\log(\frac{r_0(v)}{r(u,v)}))}.
\end{align*} Similarly, $$ |	[\rd_u(r \rd_v),\rd_uX] \phi|(u,v)  \ls  r^{-5}{v^{-\frac{1}{2}} (1+\log(\frac{r_0(v)}{r(u,v)}))}.$$ The conclusion is 	\begin{equation*}\begin{split}
		&	\bigl| \rd_u\varTheta[\rd_u X] +\frac{\lambda}{r} \zeta[\rd_u X]\bigr|\ls r^{-5}{v^{-\frac{1}{2}} (1+\log(\frac{r_0(v)}{r(u,v)}))}, \\ &  	\bigl| \rd_v \zeta[\rd_u X] +\frac{\nu}{r} \varTheta[\rd_u X]\bigr|\ls r^{-5}{v^{-\frac{1}{2}} (1+\log(\frac{r_0(v)}{r(u,v)}))}. \end{split}
\end{equation*}  Then, we can repeat the proof of Proposition~\ref{prop.commuted3}, which gives  \eqref{dduXphi} 

\end{proof}

\subsection{Scalar field refined estimates and geometric applications} In this section, we provide more refined estimates for the scalar field near $r=0$, in the style of \cite{Warren1} from which the presentation and formalism are inspired.  Adopting the notations of \cite{Warren1}, we start defining $L= \frac{\rd_v}{-r\rd_v r}$ and $\barL= \frac{\rd_u}{-r\rd_u r}$; note that $X=L-\barL$. Then, we define the key quantity \begin{equation}\label{Psi.def}
\Psi(u,v) = r^2(u,v) \frac{L\phi(u,v)+\barL\phi(u,v)}{2}.
\end{equation} We start with a lemma, in which we show $\Psi$ has a limit as $r \rightarrow 0$ together with various estimates. \begin{lemma} \label{lemma.Psi}  The following estimates hold for all $0 \leq u \leq \uS$, $v\geq v_0$:
\begin{equation}\label{Psi.main}
	|\Psi|(u,v) \ls v^{\frac{1}{2}},
\end{equation}\begin{equation}\label{XPsi}
	|X\Psi|(u,v) \ls v^{2s-\frac{1}{2}},
\end{equation}\begin{equation}\label{dPsi.main}
	|\rd_u \Psi|(u,v),\ \frac{|\rd_v \Psi|(u,v)}{r|\lambda|(u,v)} \ls v^{2s-\frac{1}{2}}(1+\log(\frac{r_0(v)}{r(u,v)})),
\end{equation} \begin{equation}\label{LvsPsi}
	\frac{\bigl| r^2 L\phi(u,v)  - \Psi(u,v) \bigr|}{v^{\frac{1}{2}}},\ \frac{\bigl| r^2 \barL\phi(u,v)  - \Psi(u,v) \bigr|}{v^{\frac{1}{2}}}\ls  \frac{r^2(u,v)}{r_0^2(v)} (1+\log
	^2(\frac{r_0(v)}{r(u,v)})).
\end{equation}moreover, there exists $\Psi_{\mathcal{S}}(v) $ such that  $$  \Psi_{\mathcal{S}}(v) := \lim_{u \rightarrow u_{\mathcal{S}}(v)} \Psi(u,v) \text{ exists},$$ and \begin{equation}\label{Psi.diff.est}
	|\Psi_{\mathcal{S}}|(v) \ls v^{\frac{1}{2}},\   \frac{|\Psi(u,v)-\Psi_{\mathcal{S}}(v)|}{v^{\frac{1}{2}}}\ls \frac{r^2(u,v)}{r_0^2(v)} (1+\log(\frac{r_0(v)}{r(u,v)})),\ \bigl|\frac{d}{dv} \Psi_{\mathcal{S}}(v)\bigr| \ls v^{-\frac{1}{2}}.
\end{equation} Furthermore, for all $ u_{\ep}(v)\leq u \leq \uS$ \begin{equation}\label{Psi.lower}
	|\Psi|(u,v) \approx v^{\frac{1}{2}},
\end{equation} where the implicit constants in $\approx$ depend on $\ep$. Finally, we have the following higher order estimate \begin{equation}\label{duduPsi}
	|\rd_{uu}^2 \Psi|(u,v),\frac{|\rd_{uv}^2 \Psi|(u,v)}{r|\lambda|(u,v)} 
	\ls r^{-2} v^{2s-\frac{1}{2}}(1+\log(\frac{r_0(v)}{r(u,v)})).
\end{equation}
\end{lemma}

\begin{proof}
First, \eqref{Psi.main} follows immediately from  the estimates of Proposition~\ref{main.prop.spacelike}. Then, from \eqref{Field}, we obtain \begin{equation}\label{duPsi.eq}
	\rd_u  \Psi  = -\frac{1}{2} r^2 \rd_u X\phi + \frac{\Omega^2}{r^2|\nu||\lambda|} \left(\frac{Q^2}{r^2}-1+m^2 r^2 |\phi|^2\right)r^2 \rd_u \phi,
\end{equation} and the estimates of Proposition~\ref{main.prop.spacelike} together with those of Proposition~\ref{prop.commuted2} and Proposition~\ref{prop.commuted3} give \begin{equation}\label{duPsi}
	|\rd_u  \Psi |(u,v) \ls v^{2s-\frac{1}{2}} (1+\log(\frac{r_0(v)}{r(u,v)})).
\end{equation} The analogous estimate for $\rd_v \Psi$  is \begin{equation}\label{dvPsi.eq}
	\rd_v  \Psi  = \frac{1}{2} r^2 \rd_v X\phi + \frac{\Omega^2}{r^2|\nu||\lambda|} \left(\frac{Q^2}{r^2}-1+m^2 r^2 |\phi|^2\right)r^2 \rd_v \phi
\end{equation} and then repeating the argument, we establish \eqref{dPsi.main}. \eqref{LvsPsi} follows from the identity $\Psi + \frac{r^2 X\phi}{2} = r^2 L\phi$ (and the analogue for $\barL\phi$) together with the estimates of Proposition~\ref{prop.commuted2} and Proposition~\ref{prop.commuted3}.
Integrating \eqref{duPsi} in $u$ gives the existence of $\Psi_{\mathcal{S}}(v)$ and \eqref{Psi.diff.est}.  Moreover, for $u \geq u_{\ep}(v)$ (as in Section~\ref{section.commuted}), we obtain, noticing that $2\Psi(u,v) + r^2 X \phi = 2 r^2 \frac{\rd_v \phi}{r\lambda}$ and using the estimates of Proposition~\ref{prop.commuted2} and Proposition~\ref{prop.commuted3} together with Proposition~\ref{lower.prop}, we find a $D>0$ and $C>0$ (independent of $\ep$) such that: $$ |\Psi|(u,v) \geq D \ep v^{\frac{1}{2}} - C r^2 v^{2s-\frac{1}{2}} [1+ \log^2(\frac{r_0(v)}{r(u,v)})] \geq \frac{D}{2} \ep v^{\frac{1}{2}}, $$ for $\ep>0$ small enough (we used the fact that $x^2 [1+\log^2(x)]$ is decreasing near $x=0$), which gives \eqref{Psi.lower}. Finally, \eqref{duduPsi} follows from a commutation with $\rd_u$ or $\rd_v$ of \eqref{duPsi.eq} and combining it with all the previously proved estimates, most importantly those of Corollary~\ref{cor.ddX}.

\end{proof}

\begin{cor}\label{cor.phi.refined}
There exists  $\varXi_S(v)$ such that for all $0 \leq u \leq \uS$, $v\geq v_0$  \begin{align}				&\label{phi.S}  |\phi(u,v) -\Psi_S(v) \log(\frac{r_0(v)}{r(u,v)})- \varXi_S(v)| \ls v^{\frac{1}{2}}\cdot \frac{r^2(u,v)}{r_0^2(v)} (1+\log^2(\frac{r_0(v)}{r(u,v)})),\\  \label{varXi.S} &|\varXi_S|(v)\lesssim  v^{\frac{1}{2}}.		
\end{align} 

\end{cor}\begin{proof}
We start defining $$ \varXi(u,v):= \phi(u,v) - \Psi(u,v) \log(\frac{r_0(v)}{r(u,v)}),$$ we get $$ \rd_u \varXi= -\frac{r|\nu|}{2} X\phi - [\rd_u \Psi] \log(\frac{r_0(v)}{r(u,v)}),$$ which we estimate as such, using Proposition~\ref{main.prop.spacelike}: $$|\rd_u \varXi| \ls v^{2s-\frac{1}{2}}[1+ \log^2(\frac{r_0(v)}{r(u,v)})],$$ and integrating in $u$ shows that  $$ \varXi_{\mathcal{S}}(v) := \lim_{u \rightarrow u_{\mathcal{S}}(v)} \varXi(u,v) \text{ exists and }  |\varXi_{\mathcal{S}}|(v) \ls v^{\frac{1}{2}},\   \frac{|\varXi(u,v)-\varXi_{\mathcal{S}}(v)|}{v^{\frac{1}{2}}}\ls \frac{r^2(u,v)}{r_0^2(v)} (1+\log^2(\frac{r_0(v)}{r(u,v)})),$$ which then gives \eqref{varXi.S}, also using taking advantage of the estimates of Lemma~\ref{lemma.Psi}. 

\end{proof}

\begin{cor}\label{cor.Gamma}
The following estimates hold true: for all $0 \leq u \leq \uS$, $v\geq v_0$:
\begin{equation}\label{dudu.omega.detailed}\bigl|\rd_u(r^2\rd_u\log	\Omega)  \bigr|(u,v) \ls v^{2s} \cdot [1+\log(\frac{r_0(v)}{r(u,v)})],\end{equation}
\begin{equation}\label{dudv.omega.detailed}\bigl|\rd_u(r^2\rd_v\log	\Omega)  \bigr|(u,v),\ \bigl|\rd_v(r^2\rd_u\log	\Omega)  \bigr|(u,v) \ls  1+\log^2(\frac{r_0(v)}{r(u,v)}),\end{equation}
\begin{equation}\label{du.omega.detailed}\bigl|-r^2\rd_u\log	\Omega(u,v)  - \frac{|\Psi|^2(u,v)-1}{2} r|\nu|\bigr| \ls v \cdot \frac{r^2(u,v)}{r_0^2(v)}[1+\log^2(\frac{r_0(v)}{r(u,v)})],\end{equation}\begin{equation}\label{dv.omega.detailed}\bigl|-r^2\rd_v\log	\Omega(u,v)  - \frac{|\Psi|^2(u,v)-1}{2} r|\lambda|\bigr| \ls v^{1-2s} \cdot \frac{r^2(u,v)}{r_0^2(v)}[1+\log^2(\frac{r_0(v)}{r(u,v)})].\end{equation}
Moreover, defining $\Gamma(u,v)$ in the following fashion \begin{equation}\label{Gamma.def}
	\Gamma(u,v):=  \frac{r_0(v)}{r(u,v)}[	\Omega^2(u,v) r(u,v)]^{\frac{1}{|\Psi|^2(u,v)}},
\end{equation}
we find that there exists $\Gamma_{\mathcal{S}}(v)$ such that we have the following estimates \begin{equation}\label{omega.detailed}\begin{split} &  \log(\Gamma(u,v))=O(1),\ |\rd_u \Gamma|(u,v),\ \frac{|\rd_v \Gamma|(u,v)}{r|\lambda|(u,v)} \ls v^{2s-1} [ 1+\log^2(\frac{r_0(v)}{r(u,v)})],\\ & \Gamma_{\mathcal{S}}(v):= \lim_{u \rightarrow u_{\mathcal{S}}(v)}\Gamma(u,v), \\
		&\bigl| \Gamma(u,v)- \Gamma_{\mathcal{S}}(v)\bigr|\ls \frac{r^2(u,v)}{r_0^2(v)}(1+\log^2(\frac{r_0(v)}{r(u,v)})).\\ & \end{split}
\end{equation}
\end{cor}
\begin{proof}
We start with  \eqref{RaychU} that we write in the form $$ -r^2 \rd_u\log(\Omega^2) + r|\nu|= -r^2 \rd_u \log(r|\nu|) + r|\nu|\bigl| \Psi - \frac{r^2 X\phi}{2}\bigr|^2.$$ Then, taking a $\rd_u$ derivative of the above and using \eqref{du.nu.S}, \eqref{dudunu}, \eqref{Psi.main} and the estimates of Proposition~\ref{prop.commuted2} and Proposition~\ref{prop.commuted3} gives \eqref{dudu.omega.detailed}.  Then, we write \eqref{RaychU} in the form \begin{equation}\label{new.RaychU}
	\rd_u	\log(\Omega^2 r)-\rd_u\log(r|\nu|)=\rd_u	\log(\frac{\Omega^2}{|\nu|})=- r^{3}|\nu| |\barL\phi|^2.\end{equation}

Then, we use \eqref{nu.S.eq2} and \eqref{du.nu}, together with \eqref{LvsPsi} to obtain \begin{equation}\label{int.duomega}
	\bigl|\rd_u	\log(\Omega^2 r)+ \frac{|\nu|}{r}|\Psi|^2 \bigr|\ls  v^{2s} (1+\log^2(\frac{r_0(v)}{r(u,v)})).
\end{equation} \eqref{du.omega.detailed} then follows immediately from \eqref{int.duomega}. Now, note that by  \eqref{dPsi.main}, \eqref{Psi.lower} and the estimates of Proposition~\ref{main.prop.spacelike}:  $$\bigl|\rd_u \log( [\Omega^2 r]^{\frac{1}{|\Psi|^2(u,v)}})-|\Psi|^{-2}\rd_u \log(\Omega^2r)\bigr| \ls |\Psi|^{-3} |\rd_u \Psi|| \log(\Omega^2 r)|\ls  v^{2s-1}(1+\log^2(\frac{r_0(v)}{r(u,v)})),$$ and combining with \eqref{int.duomega}, we arrive at \begin{equation*}
	\bigl|\rd_u \log( r^{-1} [\Omega^2 r]^{\frac{1}{|\Psi|^2(u,v)}})\bigr|=\bigl|\rd_u \log( [\Omega^2 r]^{\frac{1}{|\Psi|^2(u,v)}})+\frac{|\nu|}{r}\bigr| \ls   v^{2s-1}(1+\log^2(\frac{r_0(v)}{r(u,v)})).
\end{equation*}  This clearly means that $$ \log( \Gamma_{\mathcal{S}}(v)):=\lim_{u \rightarrow u_{\mathcal{S}}(v)} \log( \frac{r_0(v)} {r(u,v)}[\Omega^2 r]^{\frac{1}{|\Psi|^2(u,v)}})  \text{ exists, and } \bigl|\log( \frac{r_0(v) [\Omega^2 r]^{\frac{1}{|\Psi|^2(u,v)}}}{r(u,v)\Gamma_{\mathcal{S}}(v)})\bigr|\ls \frac{r^2(u,v)}{r_0^2(v)}  (1+\log^2(\frac{r_0(v)}{r(u,v)})),$$ which is \eqref{omega.detailed}. Now, evaluating the above estimate at $u=u_{\ep}(v)$, i.e., $r=\ep r_0(v)$ gives that $\log(\Gamma_{\mathcal{S}}(v))=O(1)$, using \eqref{Omega.S}. Now, for \eqref{dv.omega.detailed} we proceed analogously with \eqref{RaychV} and obtain, using the estimates of Proposition~\ref{main.prop.spacelike}:\begin{equation}
	\bigl| \rd_v	\log(\Omega^2 r)+\frac{|\lambda|}{r} |\Psi|^2\bigr| \ls 1+\log^2(\frac{r_0(v)}{r(u,v)}),\end{equation} and \eqref{dv.omega.detailed} then follows. For \eqref{dudv.omega.detailed}, we rewrite  \eqref{Omega}, using the fact that $r^2 L\phi = \Psi+\frac{r^2 X\phi}{2}$ and $r^2 \barL\phi = \Psi-\frac{r^2 X\phi}{2}$: $$ \rd_v(r^2 \rd_u\log(\Omega^2)) = -2r|\lambda| \rd_u \log(\Omega^2) -2 |\nu||\lambda|\Re\left( [\Psi+\frac{r^2 X\phi}{2}] [\overline{\Psi}-\frac{r^2 \overline{X\phi}}{2}]\right)+ \frac{ \Omega^{2}}{2}+2\lambda\nu- \frac{\Omega^{2}}{r^{2}} Q^2. $$
Now using \eqref{dv.omega.detailed}, \eqref{Psi.main} and the estimates of Proposition~\ref{prop.commuted2} and Proposition~\ref{prop.commuted3}, we finally obtain \begin{equation}\begin{split}
		&\rd_u(r^2 \rd_v\log(\Omega^2)) = |\nu||\lambda| \left[ 1-|\Psi|^2(u,v) - 2r^2 \frac{\rd_v \log(\Omega)}{r|\lambda|} + O ( \frac{r^2(u,v)}{r_0^2(v)} [1+ \log(\frac{r_0^2(v)}{r^2(u,v)})])\right] \\ & = |\nu||\lambda|   O ( \frac{r^2(u,v)}{r_0^2(v)} [1+ \log(\frac{r_0^2(v)}{r^2(u,v)})]),
	\end{split}
\end{equation}  from which the desired \eqref{dudv.omega.detailed} follows easily, thanks to the main terms canceling (we also need to revert the roles of $u$ and $v$ to also obtain the analogous estimate on $\rd_u(r^2 \rd_v\log(\Omega^2))$).

\end{proof}

\subsection{Kasner asymptotics with degenerate coefficients}
We now turn to  the problem of putting the spacetime metric $g$ in the  asymptotic form (by which we mean as $r\rightarrow 0$) of a 	Kasner metric with variable exponents. To this effect, we will take advantage of the higher order commuted estimates of Proposition~\ref{prop.commuted}.


\begin{prop}\label{kasner.prop} We define the approximate proper time \begin{equation}\label{tau.def}
	\tau(u,v) = [\frac{\Omega^2(u,v) }{r|\nu|(u,v) \cdot r|\lambda|(u,v)}]^{\frac{1}{2}} r^2(u,v)[3+|\Psi|^2(u,v)]^{-1}= \frac{r^{\frac{3}{2}}(u,v)}{[\rho(u,v)-\frac{r}{2}(u,v)]^{\frac{1}{2}}}[3+|\Psi|^2(u,v)]^{-1},
\end{equation}  and $N^2(u,v):= -g(\nabla \tau,\nabla\tau)$. There exists $\ep>0$, $D>0$ such that for all $ u_{\ep^2}(v)\leq u \leq \uS$, $v\geq v_0$,

\begin{equation}\label{tau.est}
	\tau(u,v) = \frac{\Omega(u,v) r^2(u,v)}{[r|\lambda|]_{\mathcal{S}}^{\frac{1}{2}}(v)} \frac{1}{3+ |\Psi_{\mathcal{S}}|^2(v)} \left[1+  O(\frac{r^2(u,v)}{r^2_0(v)} [1+\log^2(\frac{r^2_0(v)}{r^2(u,v)})])\right] \approx [\frac{r(u,v)}{r_0(v)} \Gamma_{\mathcal{S}}(v) ]^{\frac{|\Psi|^2_{\mathcal{S}}(v)}{2}}r^{\frac{3}{2}}(u,v) v^{s-1}.
\end{equation}
\begin{equation}\label{x.def}  					x(u,v) =  \int_{v}^{+\infty} (3+|\Psi_{\mathcal{S}}|^2(v')) (r|\lambda|)_{\mathcal{S}}(v') dv' -(3+|\Psi|^2(u,v))\frac{r^2(u,v)}{4}\cdot  				 [1+ O( \frac{r^2}{r_0^2(v)}(1+\log^2(\frac{r_0(v)}{r})))] \approx v^{2-2s},
\end{equation} 
\begin{equation}\label{kasner.form}
	g= - N^2(u,v) d\tau^2 +  a_{x x}(u,v)dx^2 + r^2(u,v) ( d\theta^2 + \sin^2(\theta)d\varphi^2),
\end{equation} 
\begin{equation}\label{N.est}
	\bigl| N^2(u,v)-1\bigr| \ls \frac{r^2(u,v)}{r^2_0(v)} [1+\log^2(\frac{r_0(v)}{r(u,v)})],
\end{equation} and defining  $p(u,v)$ by the expression \begin{equation}\label{p.def}
	r^2(u,v) =  \tau^{2p(u,v)},
\end{equation} we also have
\begin{equation}\label{axx}
	a_{xx}(u,v)= \tau^{2(1-2p(u,v))} (1+O(\frac{r^2(u,v)}{r^2_0(v)} [1+{\log^2}(\frac{r_0(v)}{r(u,v)})])),
\end{equation}
\begin{equation}\label{p.est}
	p(u,v) = \frac{2}{3+|\Psi|^2(u,v)} \left(1+O(\frac{\log(v)}{v |\log|(r)}) \right)\approx v^{-1}.
\end{equation}
Moreover, we have the following estimates on $\rd_x$: \begin{equation}\label{dx.est}
	\rd_x = \frac{1}{2} \left[\frac{\rd_u}{\rd_u x}+\frac{\rd_u}{\rd_v x} \right] = \frac{1}{2[3+|\Psi|^2(u,v)]} \left(-[1+ O( \frac{r^2}{r_0^2(v)}(1+{\log^2}(\frac{r_0(v)}{r})))]L + [1+ O( \frac{r^2}{r_0^2(v)}(1+{\log^2}(\frac{r_0(v)}{r})))] \barL \right).
\end{equation} 
\end{prop}

\begin{proof} First, note that \eqref{tau.est} follows directly from a re-organization of the estimates of Corollary~\ref{cor.Gamma} and Lemma~\ref{lemma.Psi}. Then, proceeding to the proof of \eqref{N.est}, we start writing
\begin{equation*}
	-	N^2(u,v):=	g(\nabla \tau,\nabla \tau) = -\frac{4\rd_u \tau \cdot \rd_v \tau}{\Omega^2}.\end{equation*} Then, we compute that

\begin{align*}
	& \rd_u \tau = [  \rd_u \log \Omega-\frac{2|\nu|}{r}-\frac{\rd_u(|\Psi|^2)}{3+|\Psi|^2}-\frac{1}{2}\rd_u\log(r|\nu|)-\frac{1}{2}\rd_u\log(r|\lambda|)] \tau,\\ &\rd_v \tau =[  \rd_v \log \Omega-\frac{2|\lambda|}{r}-\frac{\rd_v(|\Psi|^2)}{3+|\Psi|^2}-\frac{1}{2}\rd_v\log(r|\nu|)-\frac{1}{2}\rd_v\log(r|\lambda|)]\tau.
\end{align*} Now, by Lemma~\ref{lemma.Psi} and Corollary~\ref{cor.Gamma}, we obtain \begin{align}
	&\label{dutau}\rd_u \tau = -\frac{|\Psi|^2+3}{2} \frac{|\nu|}{r}[1+ O( \frac{r^2}{r_0^2(v)}(1+\log^2(\frac{r_0(v)}{r})))]\tau=-\frac{\Omega}{2}[\frac{r|\nu|}{r|\lambda|}]^{\frac{1}{2}}[1+ O( \frac{r^2}{r_0^2(v)}(1+\log^2(\frac{r_0(v)}{r})))],\\  &\label{dvtau}\rd_v \tau = -\frac{|\Psi|^2+3}{2} \frac{|\lambda|}{r}[1+ O( \frac{r^2}{r_0^2(v)}(1+\log^2(\frac{r_0(v)}{r})))]\tau=-\frac{\Omega}{2}[\frac{r|\lambda|}{r|\nu|}]^{\frac{1}{2}}[1+ O( \frac{r^2}{r_0^2(v)}(1+\log^2(\frac{r_0(v)}{r})))],
\end{align} 
which immediately establishes \eqref{N.est}. 
Now, we want to define the variable $x$ which is ``orthogonal'' to $\tau$. Before doing so, note the following general formulae:  \begin{align*}
	& dv = \frac{(\rd_u x) d\tau - (\rd_u \tau) dx}{\rd_v \tau \rd_u x - \rd_v x \rd_u \tau},\\ & du = \frac{-(\rd_v x) d\tau +(\rd_v \tau) dx}{\rd_v \tau \rd_u x - \rd_v x \rd_u \tau},\\ & -du dv = \frac{(\rd_u x) (\rd_v x) d\tau^2 +(\rd_u \tau) (\rd_v \tau) dx^2+(\rd_v \tau \rd_u x + \rd_v x \rd_u \tau) d\tau dx}{\left(\rd_v \tau \rd_u x - \rd_v x \rd_u \tau \right)^2}.
\end{align*} Now, recalling \eqref{rlambdaS.prop} and its notations, we set the initial value for $x(u,v)$ to be \begin{equation}\label{xS.def}
	x(\uS,v) = \int_v^{+\infty}[ 3+|\Psi_{\mathcal{S}}(v)]^2(v') r|\lambda|_{\mathcal{S}}(v') dv' \approx v^{2-2s},
\end{equation}  where the last estimate follows from Lemma~\ref{lemma.Psi}. Then, we will introduce the vector field\begin{equation}\label{T.def}
	T(u,v)  = \frac{\Omega(u,v)}{[r|\nu|(u,v)\cdot r|\lambda|]^{\frac{1}{2}}} \left[\frac{\rd_u }{\rd_u \tau}+ \frac{\rd_v }{\rd_v \tau} \right] =[1+O(\frac{r^2(u,v)}{r_0^2(v)}[1+ \log^2(\frac{r_0(v)}{r(u,v)})])] L + [1+O(\frac{r^2(u,v)}{r_0^2(v)} [1+\log^2(\frac{r_0(v)}{r(u,v)})])] \barL,
\end{equation} where for the second estimate we have used \eqref{dutau} and \eqref{dvtau}. We then define $x(u,v)$ as the solution of \eqref{xS.def} and \begin{equation}\label{Tx}
	Tx=0.
\end{equation} Note that, in view of \eqref{duS.eq}, \eqref{xS.def}, \eqref{T.def} and \eqref{Tx}, that \begin{equation}\label{dxS} \rd_u x (\uS,v)= \frac{3+|\Psi_{\mathcal{S}}|^2(v)}{2}\approx v,\ \rd_v x (\uS,v)= - (r|\lambda|)_{\mathcal{S}}(v) \frac{3+|\Psi_{\mathcal{S}}|^2(v)}{2}\approx -v^{1-2s}.  \end{equation}

Integrating \eqref{Tx} along the integral curves of $T$ using \eqref{xS.def}, we obtain \begin{equation}\label{x}
	\bigl| x(u,v) - x(\uS,v)\bigr| \ls  r^2(u,v)[1+{\log^2}(\frac{r_0(v)}{r(u,v)})] \cdot v \lesssim \frac{r^2(u,v)}{r_0^2(v)}[1+{\log^2}(\frac{r_0(v)}{r(u,v)})] \cdot x(\uS,v).
\end{equation}


Now, we want to estimate $\rd_u x$ and $\rd_v x$. We will first commute \eqref{Tx} with $\rd_u$ to estimate $\rd_u x$
: $$ T(\rd_u x)= [T,\rd_u](x) .$$

To compute $[T,\rd_u]$, note from \eqref{dutau}, \eqref{dvtau} that \begin{align*}
	& [T^{u}]^{-1}=\sqrt{r|\nu|\cdot r|\lambda|} \Omega^{-1} \rd_u \tau= [3+|\Psi|^2]^{-1} \left[ r^2 \rd_u \log(\Omega)- 2|\nu| r -r^2 \frac{\rd_u(|\Psi|^2)}{3+|\Psi|^2}-\frac{r^2}{2} \rd_u\log(r|\nu|)-\frac{r^2}{2} \rd_u\log(r|\lambda|)\right],\\ &[T^{v}]^{-1}=\sqrt{r|\nu|\cdot r|\lambda|} \Omega^{-1} \rd_v \tau= [3+|\Psi|^2]^{-1} \left[ r^2 \rd_v \log(\Omega)- 2|\lambda| r -r^2 \frac{\rd_v(|\Psi|^2)}{3+|\Psi|^2}-\frac{r^2}{2} \rd_v\log(r|\nu|)-\frac{r^2}{2} \rd_v\log(r|\lambda|)\right]. 
\end{align*} 
From the above estimates, particularly those of Lemma~\ref{lemma.Psi} and Corollary~\ref{cor.Gamma}, we obtain \begin{align*}
	& [T^{u}]^{-1} \approx 1,\  |\rd_u([T^{u}]^{-1})|\ls v^{2s-1} [1+{\log^2}(\frac{r_0(v)}{r(u,v)})],\\ &[T^{v}]^{-1} \approx v^{-2s},\ |\rd_u([T^{v}]^{-1})|\ls v^{-1} [1+{\log^2}(\frac{r_0(v)}{r(u,v)})],
\end{align*}  from which we obtain $$ \bigl|T \rd_u x\bigr|\ls  
v^{2s-1} [1+{\log^2}(\frac{r_0(v)}{r(u,v)})]\bigl| \rd_u x\bigr|+ v^{2s-1} [1+{\log^2}(\frac{r_0(v)}{r(u,v)})]  \bigl| L x\bigr|  .$$ Then, we use \eqref{T.def} the definition  of $T$  to express $L$ in terms of $\barL$ and $T$, and in view of \eqref{Tx}, we then obtain $$ \bigl|T(\rd_u x)\bigr|\ls  v^{2s-1} [1+{\log^2}(\frac{r_0(v)}{r(u,v)})]\bigl| \rd_u x\bigr|.$$ Then, we can integrate the above along the integral curves of $T$ as we did above and use Gr\"{o}nwall's inequality (note that $v^{2s-1} [1+\log(\frac{r_0(v)}{r(u,v)})]$ is indeed integrable with respect to the parameter-time of $T$) to obtain, recalling \eqref{dxS}: \begin{equation} \label{dux.est}\bigl|\rd_u x (u,v)-\rd_u x (\uS,v) \bigr| \lesssim   v \cdot \frac{r^2(u,v)}{r_0^2(v)} [1+{\log}^2(\frac{r_0(v)}{r(u,v)})].\end{equation}  Integrating \eqref{dux.est} in $u$ then gives \eqref{x.def}. To obtain the analogous estimate to \eqref{dux.est} on $\rd_v x$, we combine \eqref{T.def}, \eqref{Tx} with \eqref{dux.est}, \eqref{dxS}. From this, we conclude that \begin{equation}\label{dx}
	\rd_u x=  \frac{3+|\Psi|^2(u,v)}{2}r|\nu| [1+O(\frac{r^2(u,v)}{r_0^2(v)} [1+{\log^2}(\frac{r_0(v)}{r(u,v)})])],\ \rd_v x= \frac{-3+|\Psi|^2(u,v)}{2}r |\lambda|[1+O(\frac{r^2(u,v)}{r_0^2(v)} [1+{\log^2}(\frac{r_0(v)}{r(u,v)})])].
\end{equation}

Therefore, we get $$ a_{x x}(u,v) = \frac{\Omega^2 }{4|\rd_u x\rd_vx|}= \frac{\Omega^2}{r|\lambda| \cdot r|\nu| \cdot [3+|\Psi|^2(u,v)]^2} (1+O(\frac{r^2}{r_0^2(v)}[1+{\log^2}(\frac{r_0(v)}{r(u,v)})]),$$ from which \eqref{axx} follows (after invoking \eqref{p.def}). For \eqref{p.est}, note that \eqref{p.def} takes the form $$ p(u,v) = \frac{\log(r(u,v))}{\log(\tau(u,v))}.$$  \eqref{p.est} then follows from \eqref{tau.est} together with the estimates of Proposition~\ref{main.prop.spacelike}. 

\end{proof}

\begin{rmk}\label{rmk.Weingarten}
Instead of casting the metric in the asymptotic Kasner form \eqref{kasner.form} and defining the (variable) Kasner exponents to be $(1-2p(u,v),p(u,v),p(u,v))$ defined by $r^2(u,v) = \tau^{2p(u,v)}$ as in \eqref{p.def}, one could have instead adopted the Weingarten formalism to define the Kasner exponents (see e.g. \cite{Warren1}). More precisely, in view of $\tau$ being the approximate proper time defined in \eqref{tau.def}, we can introduce an orthonormal frame as such \begin{align*}
	& e_0 = \frac{\nabla \tau}{\sqrt{-g(\nabla\tau,\nabla \tau)}}= \Omega^{-1}\left( \sqrt{\frac{\rd_u \tau}{\rd_v \tau} }\rd_v+ \sqrt{\frac{\rd_v \tau}{\rd_u \tau} }\rd_u\right), \\ &  e_1 = e_0^v  \rd_v - e_0^u \rd_u =\Omega^{-1}\left(\sqrt{ \frac{\rd_u \tau}{\rd_v \tau}} \rd_v- \sqrt{\frac{\rd_v \tau}{\rd_u \tau}} \rd_u\right),\\  &e_3=\frac{1}{r} \rd_{\theta},\ e_4=\frac{1}{r \sin(\theta)} \rd_{\varphi}.
\end{align*}  The Weingarten map is given by  $k_{i j}(u,v)=	g(\nabla_{e_i}e_0,e_j)$, which is a diagonal tensor with eigenvalues $k_{11}(u,v)$, $k_{22}(u,v)$, $k_{33}(u,v)$. One  then alternatively defines the Kasner exponents $p_1(u,v)$, $p_2(u,v)=p_3(u,v)$ as the renormalized eigenvalues, namely: $$ p_I(u,v) = \frac{k_{II}(u,v)}{k_{11}(u,v)+k_{22}(u,v)+k_{33}(u,v)}=\frac{k_{II}(u,v)}{k_{11}(u,v)+2k_{22}(u,v)}.$$ Note that $p_1(u,v)+p_2(u,v)+p_3(u,v)=1$, by definition. As it turns out, our estimates also allow to prove that \begin{equation*}\begin{split}
		&		k_{11}(u,v) = (|\Psi|^2(u,v)-1)  \Omega^{-1}(u,v) r^{-2}(u,v) \sqrt{r|\nu|(u,v) \cdot r|\lambda|(u,v) }(1+ O(\frac{r^2}{r_0^2(v)}[1+\log^2(\frac{r_0(v)}{r(u,v)})])),\\ & k_{22}(u,v) =k_{33}(u,v) =2 \Omega^{-1}(u,v) r^{-2}(u,v) \sqrt{r|\nu|(u,v) \cdot r|\lambda|(u,v) }(1+ O(\frac{r^2}{r_0^2(v)}[1+\log^2(\frac{r_0(v)}{r(u,v)})])).
	\end{split},
\end{equation*} from which we deduce

\begin{equation}\label{k.est}\begin{split}
		p_2(u,v) =p_3(u,v) = \frac{2}{3+|\Psi|^2(u,v)}(1+ O(\frac{r^2}{r_0^2(v)}[1+\log^2(\frac{r_0(v)}{r(u,v)})])),
	\end{split}
\end{equation} 		\eqref{k.est} is consistent with  \eqref{p.est} (which uses another definition of the Kasner exponents).

\end{rmk}

The Kasner asymptotics of Proposition~\ref{kasner.prop} finally conclude the proof of  Statement~\ref{main.thm.II} of Theorem~\ref{main.thm}.

\section{Construction of the initial data}\label{section.data}
In this section, we seek a proof of Statement~\ref{main.thm.III} of Theorem~\ref{main.thm}, which follows from the construction of initial data  such that $\Cin$ is trapped with $r(u,v_0) \rightarrow 0$ as $u\rightarrow u_F$. Note indeed that, in this case, one applies Proposition~\ref{localbreak.prop}, therefore $\uch<u_F$ so the assumptions of Statement~\ref{main.thm.II} of Theorem~\ref{main.thm} are satisfied, and its conclusions  are thus valid.
Hence, we intend to construct bicharacteristic initial data  $\Cin=\{ u_0\leq u< u_F,\ v=v_0\} \cup C_{out} =\{v_0 \leq v <+\infty,\ u=u_0\}$, on which we impose the electromagnetic gauge \begin{equation}\label{A.gauge.data.all}
A_v \equiv 0.
\end{equation}
First, we construct the ingoing cone $\Cin=\{ u_0\leq u< u_F,\ v=v_0\}$ with the gauge choice\footnote{Note that \eqref{gauge.data.u}, \eqref{A.gauge.data.u} chosen for convenience in this section, are slightly different from the gauge choices  of  \eqref{gauge2}, \eqref{A.gauge2}. However,  it is not difficult to switch between these two gauges since we have systematically written gauge-invariant estimates.} \begin{align}\label{gauge.data.u}
&	-r\rd_u r(u,v_0) \equiv 1,\\ & A_u(u,v_0) \equiv 0, \label{A.gauge.data.u}
\end{align} and we arrange $r(u,v_0)\rightarrow 0$ as $u\rightarrow u_F$ by setting $$ u_F = u_0 +\frac{r^2(u_0,v_0)}{2}. $$ Under this condition,  $\Cin$ is an ingoing spherical cone collapsing to a tip $(u_F,v_0)$ with $r(u_F,v_0)=0$ (however, strictly speaking, this tip  $(u_F,v_0)$ does not belong to $\Cin$).

The following proposition is the key to constructing initial data to which Theorem~\ref{main.thm}, Statement~\ref{main.thm.III} applies. The main condition to obtain is that $\Cin$ is trapped, in the sense that $\lambda(u,v_0)<0$, for all $u_0 \leq u < u_F$. While this condition is always satisfied if $m^2=0$, and $Q(u,v_0) \equiv 0$ (the so-called Christodoulou model) by monotonicity of \eqref{Radius3}, it requires specific initial data constructions in general, carried out below. Basically, if $\rho(u_0,v_0)$ is large compared to $|Q|(u_0,v_0)$ and $\phi_{\Cin}$, then $\Cin$ is trapped. In Proposition~\ref{prop.data}, we provide five large classes  of examples under which $\Cin$ is trapped, including (surprisingly) two in which $\phi_{\Cin}$ is allowed to be arbitrarily large. For the benefit of the reader, we already describe in words these five scenarios. In what follows, we denote $\rho_0=\rho(u_0,v_0)$, $Q_0=Q(u_0,v_0)$, and we fix the constant $r_0=r(u_0,v_0)>0$.   \begin{enumerate}[A.]
\item (Large initial Hawking mass) $\rho_0>0$ is large  compared to $|Q_0|$ and $\phi$.
\item (Small  coupling parameters and initial charge) $|Q_0|$, $m^2$ and $|q_0|$ are small compared to $\rho_0$ and $\phi$.
\end{enumerate} The next three examples are formulated under  the assumption of stable scalar-field Kasner asymptotics on $\Cin$, i.e., $\phi =~ \Psi_0 \log(r^{-1})+\tphi$, with a Kasner exponent $|\Psi_0|>1$, and $\tphi = O(1)$ as $r\rightarrow 0 $. 
\begin{enumerate}[A.]\setcounter{enumi}{2}

\item (Small scalar field  and initial charge)  $|Q_0|$ and $\phi$ are small compared to $\Psi_0$ and $\rho_0$.
\item (Large Kasner exponents) $|\Psi_0|$ is large compared to $\rho_0$, $|Q_0|$ and $\tphi$.
\item (Perturbation of exact Kasner asymptotics, small initial charge and small Klein--Gordon mass) $|Q_0|$, $\tphi$ and $m^2$ are small compared to $\rho_0$ and $|\Psi_0|$.
\end{enumerate} These five classes of examples will be stated more precisely and proven in Proposition~\ref{prop.data} below. In light of the discussion at the beginning of this section, Statement~\ref{main.thm.III} of Theorem~\ref{main.thm} follows from Proposition~\ref{prop.data} below.
\begin{prop}\label{prop.data} Under the gauge conditions \eqref{gauge.data.u}, \eqref{A.gauge.data.u} on $\Cin$, we additionally assume
\begin{equation}\label{phi.upper.data}\begin{split}
		&	|\phi|(u,v_0) \lesssim |\log|\left(r^{-1}(u,v_0)\right),\\ & |D_u \phi|(u,v_0) \lesssim r^{-2}(u,v_0).
	\end{split}
\end{equation}
Moreover, we assume that	there exists $\Psi_0 \in \mathbb{C}$ with $|\Psi_0|>1$ such that  \begin{equation}\label{lower.data}
	|\Psi_0|\leq  \liminf_{u \rightarrow u_F}r^2|D_u \phi|(u,v_0).
\end{equation}
Then, the following (gauge-invariant) quantities are finite (re-expressed in the gauge \eqref{gauge.data.u}, \eqref{A.gauge.data.u}) \begin{equation}\label{I.def}\begin{split}
		I(\phi)=  &    \sup_{u_0 \leq u < u_F } \left( r^2(u,v_0)[m^2 r^2 |\phi|^2(u,v_0)-1] +  [Q_0+q_0\int_{u_0}^{u} r^2 \Im(\bar{\phi}D_u\phi)(u',v_0) du' ]^2\right)\\= &    \sup_{u_0 \leq u < u_F } \left( r^2(u,v_0)[m^2 r^2 |\phi|^2(u,v_0)-1] +  [Q_0+q_0\int_{u_0}^{u} r^2 \Im(\bar{\phi}\rd_u\phi)(u',v_0) du' ]^2\right), \end{split}
\end{equation}
\begin{equation}\label{N.def}
	N_k(\phi)= \int_{u_0}^{u_F} r^{-2}(u,v_0)  \mathcal{F}^k(u) \exp(-\mathcal{F}(u)) |\nu|(u,v_0)  du,
\end{equation}  for any $k\geq 0$, where we have introduced the notation $\mathcal{F}(u)=\int_{u_0}^{u} \frac{r|D_u \phi|^2(u',v_0)}{|\nu|(u',v_0)}  du'=\int_{u_0}^{u} r^2|\rd_u \phi|^2(u',v_0) du'$. Moreover, introducing the following norm which is critical with respect to scaling \begin{equation}\label{norm.def}
	\| F\|_{crit} = \sqrt{ \int_{u_0}^{u_F} r^{-1} [1+ \log(\frac{r_0}{r})]^{-4} |F|^2(u,v_0)|\nu|(u,v_0) du}.
\end{equation} Then, the following  (gauge-invariant) quantity is finite (re-expressed in the gauge \eqref{gauge.data.u}, \eqref{A.gauge.data.u})
\begin{equation}\label{norm.finite}
	\| \frac{ rD_u\phi}{\nu} \|_{crit} = \sqrt{ \int_{u_0}^{u_F} r  [1+ \log(\frac{r_0}{r})]^{-4} \frac{|D_u \phi|^2(u,v_0)}{|\nu|(u,v_0)} du}=\sqrt{ \int_{u_0}^{u_F} r^2  [1+ \log(\frac{r_0}{r})]^{-4} |\rd_u \phi|^2(u,v_0) du}.
\end{equation}

Denote $\phi_0= \phi(u_0,v_0)$, $\rho_0=\rho(u_0,v_0)$, $r_0=r(u_0,v_0)$, $Q_0=Q(u_0,v_0)$. If the following condition is satisfied \begin{equation}\label{condition1}
	2\rho_0 -r_0 > I(\phi) N_0(\phi),
\end{equation} then $\Cin$ is trapped, in the sense that $\lambda(u,v_0)<0$, for all $u_0 \leq u < u_F$. In particular, \begin{enumerate}[A.] 
	\item \label{example.I} fixing $(\phi(u_0,v),Q_0,r_0)$, there exists $A>0$ sufficiently large, such that if $\rho_0>A$, then \eqref{condition1} holds. Therefore, $\Cin$ is trapped, in the sense that $\lambda(u,v_0)<0$, for all $u_0 \leq u < u_F$.
\end{enumerate} 
Moreover,  fixing  $\rho_0>0$ and $r_0>0$, if $|Q_0|$ and $\phi_{\Cin}(u,v_0)$ are sufficiently small in the following sense \begin{equation}\label{condition2}	2\rho_0 -r_0 > 2\left( Q_0^2+  m^2 r_0^4 |\phi_0|^2 \right) N_0(\phi)  +2 r_0^4 \left(2q_0^2  |\phi_0|^2+  m^2
	\right) N_1(\phi)+ 4q_0^2 r_0^4 N_2(\phi),\end{equation}  then \eqref{condition1} is satisfied, thus $\Cin$ is trapped, in the sense that $\lambda(u,v_0)<0$, for all $u_0 \leq u < u_F$.  In particular, 
\begin{enumerate}[A.]\setcounter{enumi}{1}
	\item\label{example.II} fixing $(\phi(u_0,v),r_0,\rho_0)$, there exists  $\delta>0$ sufficiently small, such that if $\max\{|Q_0|, |q_0|, |m^2|\}<\delta$, then $\Cin$ is trapped, in the sense that $\lambda(u,v_0)<0$, for all $u_0 \leq u < u_F$.
\end{enumerate}

We now turn to the construction of more specific examples with Kasner-like asymptotics. 
Assume   $\phi(u,v_0)$ takes the following form for some constant $\Psi_0 \in \mathbb{C}$ with $|\Psi_0|>1$, $\epsilon\in \mathbb{C}$,  arbitrary \begin{equation}\label{Kasner.profile.data}
	\phi(u,v_0) = \Psi_0 \log(  \frac{r_0}{r(u,v_0)})+ \epsilon\ \tilde{\phi}(u),
\end{equation}where   $\tilde{\phi_0}$ is a  complex-valued function  such that $\tilde{\phi}(u)=O(1),\  \rd_u\tilde{\phi}(u)=o\left(r^{-2}(u,v_0)\right) \text{ as } u\rightarrow u_F$, and  denote \begin{equation}\label{phi0.ext}|\tilde{\phi}^{F}| := \underset{u_0 \leq u \leq u_F}{\sup}  |\tilde{\phi}|(u,v_0)<\infty.\end{equation}    

All the following examples will assume $(\rho_0,r_0)$ are fixed and $\phi$ takes the form \eqref{Kasner.profile.data}.  We start with a small data construction as measured in the scale-critical norm \eqref{norm.def}.
\begin{enumerate}[A.]\setcounter{enumi}{2}
	\item \label{example.III} Take  $\ep=1$  (with no loss of generality). Fixing $(r_0,\rho_0,\Psi_0)$ and $\Delta>0$, there exists  $\delta>0$ sufficiently small (depending on $(r_0,\rho_0,\Psi_0,\Delta)$), such that if \begin{equation}\label{condition.exampleIII}
		\max\{|Q_0|,\ |\phi_0|,\  	\| \frac{ rD_u\phi}{\nu} \|_{crit}\}<\delta \text{ and }  |\tilde{\phi}^{F}|\leq\Delta,
	\end{equation} then $\Cin$ is trapped, in the sense that $\lambda(u,v_0)<0$, for all $u_0 \leq u < u_F$. Moreover, for any $(\Psi_0,\Delta)$, there  exists a large class of examples $\phi(u_0,v)$ of  the form  \eqref{Kasner.profile.data}  such that $|\phi_0|,\ \| \frac{ rD_u\phi}{\nu} \|_{crit}<\delta$ and $|\tilde{\phi}^{F}|<\Delta$. 
\end{enumerate}
We continue with a  large-data construction, fixing $\ep$  and other quantities, but taking $|\Psi_0|$ as large as needed.

\begin{enumerate}[A.]\setcounter{enumi}{3}
	\item \label{example.IV} Fixing $(\tilde{\phi_0},\epsilon, Q_0,\rho_0,r_0)$, there exists $A>0$ sufficiently large, such that if $|\Psi_0|>A$, then  $\Cin$ is trapped, in the sense that $\lambda(u,v_0)<0$, for all $u_0 \leq u < u_F$.
\end{enumerate}
Finally, we turn to small perturbations of asymptotics of the form \eqref{Kasner.profile.data} (fixing the constant $|\Psi_0|>1$), which are still large-data examples, although they are more restrictive in that they assume $m^2$ to be small (or $m^2=0$).
\begin{enumerate}[A.]\setcounter{enumi}{4}
	\item \label{example.V}Fixing $(\Psi_0,\tilde{\phi_0},\rho_0,r_0)$, there exists $\epsilon_0>0$ sufficiently small, such that if  $\max\{|\ep|,Q_0,m^2\}< \ep_0$, then $\Cin$ is trapped, in the sense that $\lambda(u,v_0)<0$, for all $u_0 \leq u < u_F$. 
\end{enumerate}

\end{prop}

\begin{proof} 
We want to arrange this construction such that \begin{equation}\label{trapped.condition}
	r\lambda(u,v_0)<0 \text{ for all } u_0 \leq u < u_F.
\end{equation}
Re-writing \eqref{RaychU} as $ \rd_u \log(\frac{-\rd_u r}{\Omega^2}) = \frac{r |D_u \phi|^2}{|\nu|} $ and using \eqref{lower.data}, \eqref{gauge.data.u}, \eqref{A.gauge.data.u}  gives, for $u$ close enough to $u_F$  \begin{equation*}
	\Omega^2(u,v_0) 
	\ls r^{-1+|\Psi_0|^2}(u,v_0).
\end{equation*} Therefore, since $|\Psi_0|^2>1$, we have \begin{equation}\label{omega.data2}
	\int_{u_0}^{u_F}r^{-2}\Omega^2(u,v_0)  du <\infty.
\end{equation} 
Moreover, \eqref{phi.upper.data} show that $I(\phi)$ from \eqref{I.def} is well-defined and \eqref{chargeUEinstein}    gives $|\rd_u Q|\lesssim r|\nu|\log(r^{-1})$, thus   $\underset{u_0 \leq u < u_F}{\sup} |Q|(u,v_0)<~\infty$. \eqref{omega.data2} shows that $N_0(\phi)$ from \eqref{N.def} is  finite. In view of the fact (using \eqref{Radius3}) that \begin{equation}\label{lambda.data.u}
	-r\lambda(u,v_0)= 	-r\lambda(u_0,v_0) + \int_{u_0}^{u} r^{-2}\Omega^2  [ r^2 - Q^2 -m^2 r^4 |\phi|^2](u,v_0) du
\end{equation} and by \eqref{omega.data2}, \eqref{phi.upper.data}, and $\underset{u_0 \leq u < u_F}{\sup} |Q|(u,v_0)<~\infty$ shows that $\sup r|\lambda|(u,v_0)<\infty$, and moreover \begin{equation*}
	-r\lambda(u,v_0) \geq  -r\lambda(u_0,v_0) + \frac{\Omega^2(u_0,v_0)}{|\nu|(u_0,v_0)} I(\phi) N_0(\phi)=  -r\lambda(u_0,v_0) \left[ 1-  [2\rho(u_0,v_0)-r(u_0,v_0)]^{-1} I(\phi)N_0(\phi)\right].
\end{equation*} From the above inequality we indeed infer that  Condition~\eqref{condition1}   implies \eqref{trapped.condition}, from which Example~\ref{example.I} follows immediately. Condition~\eqref{condition2} is a bit more subtle, so we come back to \eqref{lambda.data.u}. Recall from the proof of  Lemma~\ref{lemma.hardy} and the notation $\mathcal{F}$ that $$ \Omega^2(u,v_0) =  \frac{\Omega^2(u_0,v_0)}{|\nu|(u_0,v_0)}|\nu|(u,v_0)  \exp(-\mathcal{F}(u))=  \frac{r|\lambda|(u_0,v_0)}{2\rho_0-r_0}\ r^{-1}(u,v_0)\exp(-\mathcal{F}(u)).$$ Now, from the estimates of Lemma~\ref{lemma.hardy} (in particular \eqref{Hardy.phi} and \eqref{Hardy.Q}), we obtain \begin{align*} 
	&	\int_{u_0}^{u_F} r^{-2}\Omega^2 Q^2(u,v_0 ) du \leq \frac{r|\lambda|(u_0,v_0)}{2\rho_0-r_0}\left( 2 Q_0^2 N_0(\phi)+ 4q_0^2 r_0^4 |\phi_0|^2  N_1(\phi)+ 4q_0^2 r_0^4   N_2(\phi)\right),\\ & \int_{u_0}^{u_F} r^{2}\Omega^2 |\phi|^2(u,v_0 ) du \leq 2\frac{r|\lambda|(u_0,v_0)}{2\rho_0-r_0}\ r_0^4\left(|\phi_0|^2 N_0(\phi) + \int_{u_0}^{u_F} \frac{\mathcal{F}(u)}{r^{2}(u,v_0)} \exp(-\mathcal{F}(u)) \frac{r^2(u,v_0)}{r_0^2}\log(\frac{r_0}{r(u,v_0)}) |\nu|(u,v_0)du\right),
\end{align*} and thus Condition~\eqref{condition2} indeed implies \eqref{trapped.condition}, remarking that $x^2 \log(x^{-1}) \leq 1$ for all $x \in (0,1]$. Example~\ref{example.II} follows immediately.

For Example~\ref{example.III}, we want to show that Condition~\ref{condition2} is satisfied. Note that
\begin{align*}  &  N_k(\phi)	\leq  \| \frac{rD_u\phi}{\nu}\|_{crit}^{2k}\int_{u_0}^{u_F} r^{-2}|\nu|[ 1+ \log(\frac{r_0}{r})]^{4k} \exp(-\mathcal{F}(u)) du\\ & \leq   \| \frac{rD_u\phi}{\nu}\|_{crit}^{2k}\ r_0^{-|\Psi_0|^2}\int_{u_0}^{u_F} r^{-2+|\Psi_0|^2}|\nu| [ 1+ \log(\frac{r_0}{r})]^{4k} \exp\left(-2\int_{u_0}^u  \Re(\ep \bar{\Psi}_0\  \rd_u \tilde{\phi})(u',v_0) du' \right) du\\ & \leq   \| \frac{rD_u\phi}{\nu}\|_{crit}^{2k}\ r_0^{-1} \exp\left(2|\ep||\Psi_0| [ |\tilde{\phi}^{F}|+ |\tilde{\phi}|(u_0,v_0)]\right)\underbrace{\int_0^{1} y^{-2+|\Psi_0|^2}[1+\log(y^{-1})]^{4k} dy}_{=J_k(|\Psi_0|^2)} \\ & \leq  \| \frac{rD_u\phi}{\nu}\|_{crit}^{2k}\ r_0^{-1} \exp\left(2|\Psi_0| [ \Delta+ |\phi_0|]\right) J_k(|\Psi_0|^2),  \end{align*} 



where in the last line we have used $|\ep|\leq 1$ and $|\tilde{\phi}^{F}|\leq \Delta$, which then shows that  \eqref{condition.exampleIII} implies  \eqref{condition2} which, in turn, implies \eqref{trapped.condition}. To conclude the discussion of Example~\ref{example.III}, we must also construct a class of $\phi(u_0,v)$ satisfying \eqref{condition.exampleIII}: this construction will be a generalization of profiles of the following form for some small $\eta>0$: \begin{align*}
	&  \phi(r)= \Psi_0 \log(\frac{r_0}{r}) \text{ for } r \leq \eta,\text{ i.e., } \tphi \equiv 0 \text{ for } r \leq \eta \\ & \phi(r)  \equiv 0 \text{ for } r \geq 2\eta .
\end{align*}
Starting from the form \eqref{Kasner.profile.data} with fixed $|\Psi_0|>1$, $\Delta>0$, we impose  $|\tilde{\phi}^{F}| <\Delta$ and a smallness condition of $\tilde{\phi}$ restricted to $\{r\leq \eta\}$ in the scale-critical norm \eqref{norm.def}, specifically (consistently with  $O(r^{-2})$ asymptotics of $D_u\tphi$) \begin{equation*}
	\| 1_{r \leq \eta}\ \frac{rD_u\tphi}{\nu}\|_{crit}=O(\log^{-\frac{3}{2}}(\eta^{-1})).
\end{equation*}(a particular case would be to impose $\tphi\equiv 0$ for $r\leq \eta$). Defining  $F_K = \Psi_0 \log(\frac{r_0}{r})= \phi-\tphi$, we also have  \begin{equation*}
	\| 1_{r \leq \eta}\ \frac{rD_uF_K}{\nu}\|_{crit}=O(\log^{-\frac{3}{2}}(\eta^{-1})),
\end{equation*} from which we get by triangular inequality  \begin{equation}\label{exIII1}
	\| 1_{r \leq \eta}\ \frac{rD_u\phi}{\nu}\|_{crit}=O(\log^{-\frac{3}{2}}(\eta^{-1})).
\end{equation} Now we impose a smallness condition on $\phi$ on  $\{r\geq 2\eta\}$ in the scale-critical norm \eqref{norm.def}, in the sense that  \begin{equation}\label{exIII3}
	|\phi_0|+\| 1_{r \geq 2\eta}\ \frac{rD_u \phi}{\nu}\|_{crit}=O(\log^{-10}(\eta^{-1})).
\end{equation} (a particular case would be to impose $\phi\equiv 0$ for $r\geq 2\eta$). Note that for $\eta>0$ small enough, $\phi \approx \Psi_0 \log(\eta^{-1})$ when $r=\eta$, hence $\rd_r \phi=\frac{\rd_u \phi}{\nu}$ is at least of size $\eta^{-1} \log(\eta^{-1})$ when $ \eta \leq r\leq 2\eta$, yet due to the definition of \eqref{norm.def}, it is possible to arrange that \begin{equation}\label{exIII4}
	\| 1_{\eta \leq r \leq 2\eta}\ \frac{rD_u \phi}{\nu}\|_{crit}=O(\log^{-1}(\eta^{-1})).
\end{equation} Combining \eqref{exIII1}, \eqref{exIII3}, \eqref{exIII4} gives \begin{equation*}
	|\phi_0|+\|\frac{rD_u \phi}{\nu}\|_{crit}=O(\log^{-1}(\eta^{-1})),
\end{equation*} which is enough to satisfy  \eqref{condition.exampleIII}, taking $\delta=\log^{-1}(\eta^{-1})$ (assuming, of course, that $|Q_0|$ is appropriately small).

For Example~\ref{example.IV}, we want to show that \eqref{condition2} is satisfied by proving that $N_k(\phi)$ can be made arbitrarily small by increasing $|\Psi_0|$. Then, using the inequality $x^k e^{-\frac{x}{2}} \leq C_k$ for some $C_k>0$, we write the key estimate \begin{align*}
	N_k(\phi) \leq &  \int_{u_0}^{u_F}   r^{-2}[\mathcal{F}(u)^{k}\exp(- \frac{\mathcal{F}(u)}{2})]\exp(- \frac{\mathcal{F}(u)}{2})|\nu| du \leq   C_k \int_{u_0}^{u_F}   r^{-2}\exp(- \frac{\mathcal{F}(u)}{2})|\nu| du\\ & \lesssim  \int_{u_0}^{u_F}   r^{-2+\frac{|\Psi_0|^2}{4}}|\nu| du=  O(|\Psi_0|^{-2}),
\end{align*} assuming that $|\Psi_0|$ is sufficiently large, thus \eqref{condition2} is satisfied, which implies \eqref{trapped.condition}.

For Example~\ref{example.V}, we want to show instead that Condition~\ref{condition1} holds. Note that $N_0(\phi)$ obeys a fixed upper bound (in view of $\Psi_0$ and $\tphi$ being fixed and $\ep$ being in a bounded set), so the objective is to make $I(\phi)$ small enough. Let $\ep_0>0$ small, then  we can make $m^2$ and $|Q_0|$ sufficiently small depending on $\ep_0$. Thus, only the term $\int r^2 \Im(\bar{\phi} D_u\phi ) du$ remains in $I(\phi)$: the key is to notice the following cancellation \begin{equation*}
	\int_{u_0}^{u_F}r^2 |\Im(\bar{\phi} D_u\phi )| du \leq  \int_{u_0}^{u_F}r^2 \left(|\Psi_0||\ep|[r^{-2}|\tphi|+ \log(\frac{r_0}{r}) |\rd_u \tphi|]+|\ep|^2 |\tphi \rd_u\tphi | \right) du=O(\ep),
\end{equation*} where the last $O(\ep)$ estimate follows from the fact that $\Psi_0$ and $\tphi$ are fixed. Thus, \eqref{condition2} and the claim of Example~\ref{example.V} are established.

\end{proof}

\bibliographystyle{plain}
\bibliography{bibliography.bib}

\end{document}